\documentclass[runningheads,cleveref,thm-restate]{lipics-v2021}
\usepackage[T1]{fontenc}
\usepackage{float}
\usepackage{amsmath}
\usepackage{amssymb}
\usepackage{algorithm, algpseudocode, float}
\usepackage{xspace}
\usepackage{graphicx}
\usepackage{capt-of}
\usepackage{tikz}
\usetikzlibrary{positioning}
\usetikzlibrary{decorations.markings}
\usetikzlibrary{matrix}
\usepackage{complexity}
\usepackage{hyperref}

\providecommand {\calO}      {{\mathcal O}\xspace}

\providecommand {\calO}      {{\mathcal O}\xspace}

\newcommand{\T}{T_G}
\newcommand{\Ty}{T_{G'}}

\newcommand{\twh}{\tau_H}
\newcommand{\rnode}{S}
\newcommand{\tnode}{T}

\newcommand{\parConf}{\mathcal{P}}
\newcommand{\parConfBad}{{\mathcal{P}}_{bad}}
\newcommand{\bk}[1]{{#1}_{bad}}
\newcommand{\bkc}[1]{\widetilde{#1}_{bad}}
\newcommand{\subc}{\subseteq_c}

\newcommand{\main}{main}
\newcommand{\branch}{branch}
\newcommand{\lnode}{L}
\newcommand{\lnodex}{L'}

\newcommand{\csum}[1]{\oplus_{#1}}

\newcommand{\dirs}[1]{\emph{Dir}_{#1}}
\newcommand{\vias}[1]{\emph{Via}_{#1}}
\newcommand{\dir}[3]{\text{\emph{Dir}}_{#1}(#2,#3)}
\newcommand{\via}[4]{\text{\emph{Via}}_{#1}(#2,#3,#4)}

\newcommand{\projset}[1]{\mathcal{N}(#1)}

\newcommand{\leftR}[1]{left(#1)}
\newcommand{\rightR}[1]{right(#1)}

\newcommand{\leftm}[1]{{#1}^{L}}
\newcommand{\rightm}[1]{{#1}^{R}}

\newlength{\ml}
\newlength{\ld}
\newlength{\hd}
\newcommand{\phidisj}[2]{\phi_{disj}(#1,#2)}
\newcommand{\phipart}[3]{\phi_{part}(#1,#2,#3)}
\newcommand{\phieven}[1]{\phi^{e}_{s,t}(#1)}
\newcommand{\phiodd}[1]{\phi^{o}_{s,t}(#1)}

\newcommand{\phidppeven}{\phi^{e}_{s_1,t_1,s_2,t_2}}

\title{The Even-Path Problem in Directed Single-Crossing-Minor-Free Graphs}

\author{Archit Chauhan}{Chennai Mathematical Institute, India}{archit@cmi.ac.in}{}
{Partially supported by a grant from Infosys foundation and TCS PhD fellowship.}
\author{Samir Datta}{Chennai Mathematical Institute \& UMI ReLaX, India}{sdatta@cmi.ac.in}{}
{Partially supported by a grant from Infosys foundation}
\author{Chetan Gupta}{Indian Institute of Technology, Roorkee, India}
{chetan.gupta@cs.iitr.ac.in}{}{}
\author{Vimal Raj Sharma}{Indian Institute of Technology, Jodhpur, India}
{vimalraj@iitj.ac.in}{}{}

\authorrunning{A. Chauhan, S. Datta, C. Gupta and V. \,R. Sharma}
\Copyright{Archit Chauhan, Samir Datta, Chetan Gupta, Vimal Raj Sharma} 


\acknowledgements{We would like to thank anonymous reviewers for 
  their useful comments and corrections in past submissions. 
  We would also like to thank Martin Grohe 
for pointing out some references. 
We are also grateful to Geevarghese Philip and Vishwa Prakash H .V. 
for their help in reading some sections and 
improving the presentation of this paper.}

\nolinenumbers 

\begin{document}

\maketitle              
\begin{abstract}
Finding a simple path of even length between two designated vertices in a 
directed graph is a fundamental \NP-complete problem \cite{LP84} known as 
the $\EP$ problem.
Nedev \cite{Nedev99} proved in 1999, that for directed planar 
graphs, the problem can be solved in polynomial time. More than two decades since then, 
we make the first progress in extending the tractable classes of  
graphs for this problem. We give a polynomial time algorithm to solve the 
$\EP$ problem for classes of
$H$-minor-free directed graphs,\footnote{Throughout this paper, when 
referring to concepts like treewidth 
or minors of directed graphs, we intend them to apply to the underlying undirected graph.} 
where $H$ is a single-crossing graph. 

We make two new technical contributions along the way, that might be of independent interest. 
The first, and perhaps our main, contribution is the construction of 
small, planar, \emph{parity-mimicking networks}. 
These are graphs that mimic parities of all possible paths
between a designated set of terminals of the original graph.

Finding vertex disjoint paths between given source-destination pairs of vertices is 
another fundamental problem, known to be \NP-complete in directed 
graphs~\cite{FHW80}, though known 
to be tractable in planar directed graphs~\cite{Schrijver94}. 
We encounter a natural variant of this problem, that of finding disjoint paths 
between given pairs of vertices, but with constraints on parity of the total 
length of paths. The other significant contribution of our paper is 
to give a polynomial time algorithm for the \emph{$3$-disjoint paths with total parity 
problem}, in directed planar graphs with some restrictions (and also in 
directed graphs of bounded treewidth).

\end{abstract}

\section{Introduction}
Given a directed graph $G$, and two vertices $s$ and $t$ in it, checking for the 
existence of a 
a simple directed path from $s$ to $t$ is a fundamental problem in graph theory, 
known as the $\REACH$ problem.  
The $\EP$ problem is a variant of $\REACH,$  where given a directed 
graph $G$ and two 
vertices $s$ and $t$ we need to answer whether there exists a simple path of even 
length from $s$ to $t$. 
$\EP$ was shown to be $\NP$-$\mbox{\textsf{complete}}$ by 
LaPaugh and Papadimitriou \cite{LP84} via a reduction from 
an $\NP$\textsf{-complete} 
problem, the \textsf{Path}-\textsf{Via}-\textsf{A}-\textsf{Vertex} problem. 
On the other hand, they also show in~\cite{LP84} that its undirected counterpart 
is solvable in linear time. 
Several researchers have recently studied both the space, and simultaneous time-space complexity 
of $\EP$ for special classes of graphs \cite{CT15,DGKT12}. A similar problem, that of
finding a simple directed cycle of even length, called $\EC$ (which easily reduces
to $\EP$), has also received significant attention. While polynomial-time algorithms
have been known since long for the undirected version (\cite{LP84,YU97}), 
the question of tractablility of the directed version was open for over 
two decades before polynomial-time algorithms were given by McCuiag, and by 
Robertson, Seymour and Thomas~\cite{MRST97}. 
More recently, Bj\"orklund, Husfeldt and Kaski~\cite{BHK22} gave a randomized 
polynomial-time algorithm for finding a \emph{shortest} even directed cycle in directed 
graphs.

Although $\EP$ is $\NP$-$\textsf{complete}$ for general directed 
graphs, it is natural and interesting to investigate the classes 
of graphs for which it can be solved efficiently.
In 1994, before the algorithm of~\cite{MRST97}, Gallucio and Loebl~\cite{GL94}
gave a polynomial-time algorithm for $\EC$ in planar directed graphs. They did
so by developing a routine for a restricted variant of $\EP$ (when
$s,t$ lie on a common face, and there are no even directed cycles left on removal
of that face).
Following that, Nedev in 1999, showed that $\EP$ in planar graphs is 
polynomial-time solvable~\cite{Nedev99}. 
Planar graphs are an example of a \emph{minor-closed} family, which are families of 
graphs that are closed under edge contraction and deletion. 
Minor-closed families include many 
more natural classes of graphs, like graphs of bounded genus, 
graphs of bounded treewidth, apex graphs.
A theorem of Robertson-Seymour~\cite{RS04} shows that 
every minor-closed family can be characterized by a set of finite forbidden 
minors. Planar graphs, for example, are exactly graphs with $K_{3,3,},K_5$ as 
forbidden minors~\cite{Wagner1937}. 
In this paper, we consider the family of $H$-minor-free graphs, 
where $H$ is any fixed single-crossing graph, i.e., $H$ can be drawn on the plane with 
at most one crossing. 
Such families are called single-crossing-minor-free graphs.
They include well-studied classes of graphs like $K_5$-minor-free graphs, $K_{3,3}$-minor-free 
graphs.\footnote{Both $K_5,K_{3,3}$ have crossing number one. Also, note that 
both the families, $K_5$-minor-free, and $K_{3,3}$-
minor-free, have graphs of $O(n)$ genus.} 

Robertson and Seymour showed that single-crossing-minor free graphs admit 
a decomposition by (upto) $3$-clique-sums, into pieces that are either of 
bounded treewidth, or planar~\cite{RS93}. This is a simpler version of their 
more general theorem regarding decomposition of $H$-minor free graphs, 
(where $H$ is any fixed graph) by clique sums, into more complex pieces, 
involving apices and vortices~\cite{RS03}.
Solving $\EP$ on single-crossing-minor free graphs would therefore be 
a natural step to build an attack on more general minor closed familes.

Many results on problems like reachability,
matching, coloring, isomorphism, for planar graphs have been extended to 
$K_{3,3}$-minor-free graphs and $K_{5}$-minor-free graphs as a next step
(see \cite{TW09,Khuller90_k5,Khuller90_k33,Vazirani89,STF14,DNTW09,AGGT14}). 
Chambers and Eppstein showed in \cite{CE} that 
using the results of \cite{BK09,Hagerup-et-al} for maximum flows in planar and
bounded treewidth graphs, respectively, maximum flows in
single-crossing-minor-free graphs can be computed efficiently. 
Following the result of \cite{AV20} which showed that perfect matching
in planar graphs can be found in $\NC$, Eppstein and Vazirani in \cite{EV}, 
extended the result to single-crossing-minor-free graphs.  
\subsection{Our Contributions}
From here onwards, we will drop the term `directed' and assume by default that 
the graphs we are referring to are directed, unless otherwise stated. 
Operations like clique sums, decomposing the graphs along separating triples, 
pairs, etc., will be applied on the underlying undirected graphs.
The following is the main theorem we prove in this paper:
\begin{theorem}
 Given an $H$-minor-free graph $G$ for any fixed single-crossing graph $H$, the $\EP$
 problem in $G$ can be solved in polynomial time.
\end{theorem}
We first apply the theorem of Robertson-Seymour (theorem~\ref{thm:RS_scm_free}), and decompose 
$G$ using $3$-clique sums into pieces that are either planar or of bounded 
treewidth.
Though $\EP$ is tractable in planar graphs, and can also be solved in bounded treewidth 
graphs by Courcelle's theorem, straightforward dynamic programming  
does not yield a polynomial-time algorithm for the problem, as we will explain in 
subsequent sections.
One of the technical ingredients that we develop to overcome the issues is that of
\emph{parity-mimicking networks}, which are graphs that preserve the parities of
various paths between designated terminal vertices of the graph it mimics. We
construct them for upto three terminal vertices.
The idea of mimicking networks has been used in the past in other 
problems, like flow computation \cite{CE,CSWZ,Hagerup-et-al,KR14,KR13}, 
and in perfect matching \cite{EV}. 
The ideas we use for constructing
parity mimicking networks however, do not rely on any existing work that we know of.
For technical reasons, we require our parity mimicking networks to be of bounded treewidth
\emph{and} planar, with all terminals lying on a common face. 
These requirements make it more challenging to construct them (or 
even to check their existence), than might seem at a first glance. 
One of our main contributions is to show (in lemma~\ref{lem:par_mim_net}) the 
construction of such networks, for upto three terminals. 
It might be of independent interest to see if a more simpler construction exists 
(perhaps a constructive argument to route paths, that has eluded us so far), 
that avoids the hefty case analysis we do, and also if they can be constructed 
for more than three terminals. 

We also come across a natural variant of another famous problem.
Suppose we are given a graph $G$ and vertices $s_1,t_1,s_2,t_2 \ldots s_k,t_k$ 
(we may call them terminals) in it. 
The problem of finding pairwise vertex disjoint paths, from each $s_i$ to 
$t_i$ is a well-studied problem called the disjoint paths problem. In undirected 
graphs, the problem is in $\P$ when $k$ is fixed~\cite{RS95,RRSS91}, but 
$\NP$-$\mbox{\textsf{complete}}$ otherwise~\cite{Lynch75}. 
For (directed) graphs, the problem is $\NP$-$\mbox{\textsf{complete}}$ even 
for $k=2$~\cite{FHW80}. In planar graphs, it is known to be in $\P$ for 
fixed $k$~\cite{Schrijver94,CDPP13,LMPSZ20}. We consider this problem, with the 
additional constraint that the sum of lengths of the $s_i$-$t_i$ paths must 
be of specified parity. We hereafter refer to the parity of the sum of lengths as 
total parity, and refer to the problem as $\DPTP$.  
In the undirected setting, a stricter version of this 
problem has been studied, where each $s_i$-$t_i$ path must have parity $p_i$ 
that is specified in input. This problem was shown to be in $\P$ for fixed 
$k$, by Kawarabayashi et al.~\cite{KRW11}. However much less is known in directed 
setting.
While $\DPTP$ can be solved for fixed $k$ in bounded treewidth graphs using Courcelle's 
theorem~\cite{COURCELLE}, 
we do not yet know if it is tractable in planar graphs, even for $k=2$. 
The other main technical contribution of our paper is in lemma~\ref{lem:3dpp}, where 
we show that in some special cases, i.e., when 
there are four terminals, three of which lie on a common face of a planar graph, 
$\DPTP$ can be solved in polynomial time for $k=3$. 
We do this by showing that under the extra constraints,
the machinery developed by \cite{Nedev99} can be further generalized 
and applied to find a solution in polynomial
time. The question of tractability of $\DPTP$ in planar graphs, without 
any constraint of some terminals lying on a common face is open, and would be 
interesting to resolve.  
A polynomial-time algorithm for it (for fixed $k$), would 
yield a polynomial-time algorithm for 
$\EP$ in graphs with upto $k$ crossings, which 
is currently unknown.

Though the proofs of lemmas~\ref{lem:par_mim_net},\ref{lem:3dpp} form the meat 
of technical contributions 
of the paper, we defer them to~\cref{app:par_mim_net_proof,app:3dpp_planar}.

\section{Preliminaries}\label{prelims}
From now onwards we will refer to simple, directed paths as just paths.
For a path $P$, and a pair of vertices $u$ and $v$ on $P$, such that $u$ 
occurs before $v$ in $P$, $P[u,v]$ 
denotes the subpath of $P$ from $u$ to $v$. If $P_1$ and $P_2$ are two paths that are 
vertex disjoint, except possibly sharing starting or ending vertices, then we say that
$P_1,P_2$ are \emph{internally disjoint} paths. If $P_1$'s ending vertex is same as the 
starting vertex of $P_2$, then we denote the concatenation of 
$P_1$ and $P_2$ by $P_1.P_2$. We will use the numbers $0,1$ to refer to 
  \emph{parities}, $0$ for even parity and $1$ for odd
  parity. We say a path $P$ is of \emph{parity} $p$ $(p\in \{0,1\})$, if its length 
  modulo $2$ is $p$.
We will use a well-known structural decomposition of $H$-minor-free graphs due to 
\cite{RS93}. 
We recall the definition of clique sums:
\begin{definition}
 A $k$-clique-sum of two graphs $G_1$, $G_2$ can be obtained from the
 disjoint union of $G_1,G_2$ by identifying a clique in $G_1$ of at
 most $k$ vertices with a clique of the same number of vertices in $G_2$,
 and then possibly deleting some of the edges of the merged clique.
\end{definition}
Thus when separating $G$ along a separating pair/triplet, we can add
virtual edges if needed, to make the separating pair/triplet a clique. 
The virtual edges will not be used in computation of path parities, 
they are only used to compute the decomposition. We can keep track of which edges in the 
graph are virtual edges and which are the real edges throughout the algorithm. 
We can repeatedly apply this procedure to decompose any graph $G$ into smaller 
pieces.
The following is a theorem from \cite{RS93}.
\begin{theorem}[Robertson-Seymour \cite{RS93}]\label{thm:RS_scm_free}
 For any single-crossing graph $H$, there is an integer $\twh$ such 
 that every graph with no minor isomorphic to H is either 
 \begin{enumerate}
  \item the proper $0$-, $1$-, $2$- or $3$-clique-sum of two graphs, or
  \item planar 
  \item of treewidth $\leq \twh$. 
 \end{enumerate}
\end{theorem}
Thus, every $H$-minor-free graph, where $H$ is a single-crossing graph, 
can be decomposed by $3$-clique sums into graphs that are either planar or
have treewidth at most $\twh$. 
Polynomial time algorithms are known to compute this decomposition
~\cite{DHNRT04,KW11,GKR13} (and also $\NC$ algorithms~\cite{EV}). 
The decomposition can be thought of as a two colored tree (see 
\cite{DHNRT04,CE,EV} for further details on the decomposition), where the blue colored 
nodes represent \emph{pieces} (subgraphs that are either planar or have bounded treewidth), 
and the red nodes represent cliques at which two or more pieces
are attached. We call these nodes of the tree decomposition as $piece$ $nodes$ and 
$clique$ $nodes$, respectively.
The edges of the tree describe the incidence relation between pieces and cliques 
(see~\cref{fig:clique_decomp}).
We will denote this decomposition tree by $\T$.
We will sometimes abuse notation slightly and refer to a piece of $\T$ (and also phrases 
like \emph{leaf} piece, \emph{child} piece), 
when it is clear from
the context that we mean the piece represented by the corresponding node of $\T$.
Note that the bounded treewidth and planarity condition 
on the pieces we get in the decomposition, is along with their virtual edges. 
As explained in~\cite{CE,EV}, we can assume that in any planar piece of the decomposition, 
the vertices of a separating pair or triplet lie on a common face
(Else we could decompose the graph further).

Suppose $G$ decomposes via a $3$-clique sum at clique $c$ into $G_1$ and $G_2$. Then we 
write $G$ as $G_1{\oplus}_{c} G_2$.
More generally, if $G_1, G_2, \ldots, G_{\ell}$ all share a common clique $c$, then we 
use $G_1\oplus_c G_2 \oplus_c \ldots \oplus_c G_{\ell}$ 
to mean $G_1, G_2, \ldots, G_{\ell}$ are glued together 
at the shared clique. 
If it is clear from the context which clique we are referring to, we will sometimes
drop the subscript and simply use $G_1\oplus G_2 \oplus \ldots \oplus G_{\ell}$ instead.
Suppose $G'_2$ is a graph that contains the vertices of the clique $c$ shared by 
$G_1$ and $G_2$. 
We denote by $G[G_2 \rightarrow G'_2]$, the graph $G_1{\oplus}_c G'_2$, 
i.e., replacing the subgraph $G_2$ of $G$, by $G'_2$, keeping the clique vertices intact. 
We will also use the notion of \emph{snapshot} of a path in a subgraph. If $G$ can be decomposed 
into 
$G_1$ and $G_2$ as above, and $P$ is an $s$-$t$ path in $G$, its snapshot in $G_1$ 
is the set of maximal subpaths of $P$, restricted to vertices of $G_1$.
Within a piece, we will sometimes refer to the vertices of 
separating cliques, and $s$ and $t$, as $terminals$. 

In figures, we will generally use the convention that a single arrow denotes
a path segment of odd parity and double arrow denotes a path segment of even
parity, unless there is an explicit expression for the parity mentioned beside
the segment.

\begin{figure}
\begin{minipage}{0.4\textwidth}
  \includegraphics[scale=0.7]{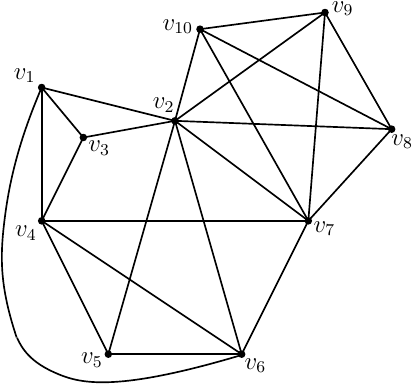}
  \caption{An example of a graph $G$. We ignore directions here.}
\end{minipage}
\begin{minipage}{0.6\textwidth}
  \includegraphics[scale=0.5]{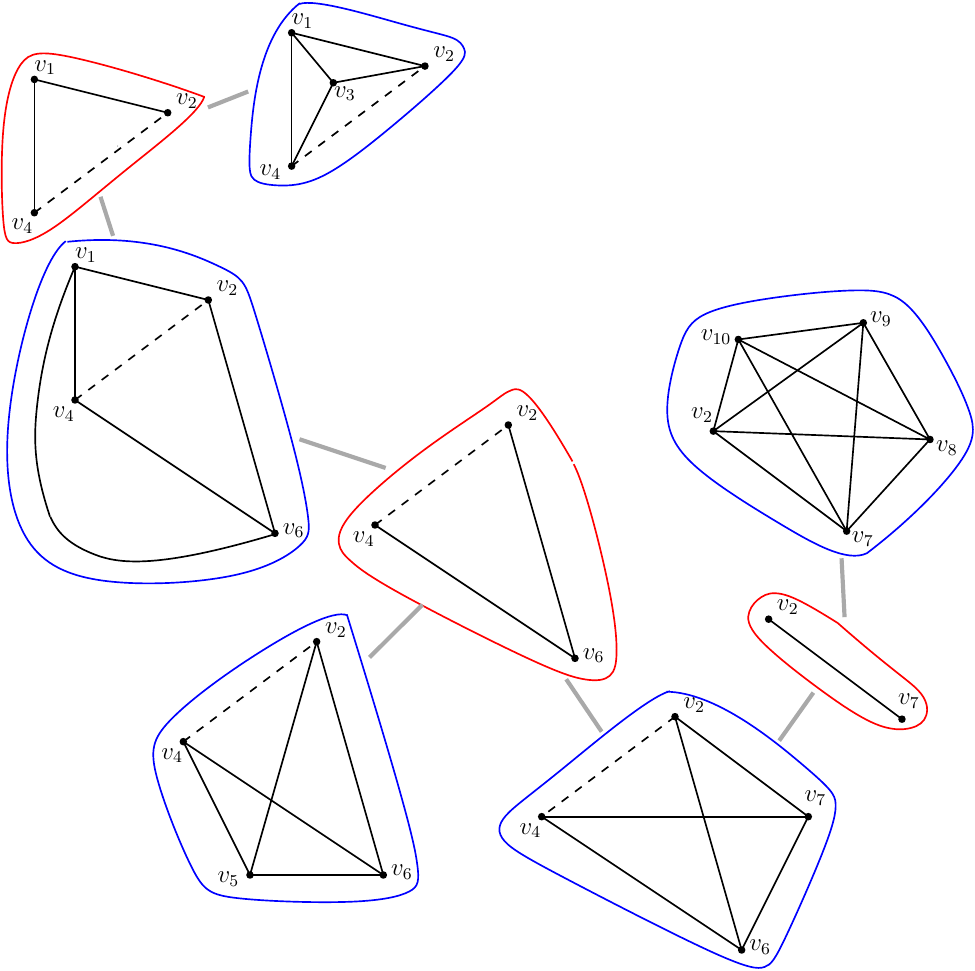}
  \caption{A clique sum decomposition of $G$. Red nodes
   are the clique nodes and blue node the piece nodes.
   Dashed edges denote virtual edges.}\label{fig:clique_decomp}
\end{minipage}
\end{figure}

\section{Overview and Technical Ingredients}\label{sec:overview}
We first compute the $3$-clique sum decomposition tree of $G$, $\T$. 
We can assume that $s,t,$ each occur in only one piece of $\T$,
$S$ and $T$, respectively.\footnote{If they are part
of a separating vertex/pair/triplet then they may occur in multiple pieces of $\T$. Say $s$
is a part of many pieces in $\T$. 
To handle that case, we can introduce a dummy $s'$ and add an edge
from $s'$ to $s$ and reduce the problem to finding an odd length path from $s'$ to $t$.
The vertex $s'$ now will occur in a unique piece in $\T$. Vertex $t$ can be handled
similarly.} We call the pieces $S$ and $T$, along with the 
pieces corresponding to nodes that lie in the unique path in $\T$ 
joining $S$ and $T$, as the $\main$ pieces of $\T$, and the remaining pieces are 
called the $\branch$ pieces of $\T$.  We will assume throughout that $\T$ is rooted at $S$. 

The high level strategy of our algorithm follows that of \cite{CE}. The algorithm has two phases. 
In the first phase, we simplify the branch pieces of the decomposition tree. 
Any $s$-$t$ even path $P$ must start and end inside the main pieces $S$ and $T$, respectively. 
However, it may take a detour into the branch pieces. 
Suppose $L$ is a leaf branch piece of $\T$, attached to its parent piece, 
say $G_i$, via a $3$-clique $c$.
Using Nedev's algorithm or Courcelle's theorem, we can find paths of various parities 
between vertices of $c$ in $L$, which constitutes the \emph{parity configuration} 
of $L$ with respect to $c$ (formally defined in next subsection).
We will replace $L$ by a parity mimicking network of $L$ with respect to
vertices of $c$, $L'$. 
$L'$ will mimic the parity configuration of $L$ and hence preserve the 
parities of all $s$-$t$ paths of original graph.
The parity mimicking networks we construct are small 
and planar, with 
the terminals (vertices of $c$) all lying on a common face, as decribed 
in lemma~\ref{lem:par_mim_net}. 
Therefore, if $G_i$ is of bounded treewidth, then $G_i\oplus L'$ will be 
of bounded treewidth. And if $G_i$ is planar, then we can plug $L'$ 
in the face of $G_i$ that is common to vertices of $c$, 
and $G_i\oplus L'$ will be planar. 
This allows us compute the parity configurations of the merged piece, and repeat 
this step until a single branch, i.e. a path, 
remains in the decomposition tree, consisting only of the main pieces 
(connected by cliques), 
including $\rnode$ and $\tnode$.

In the second phase, we start simplifying the main pieces, 
starting with the leaf piece $\tnode$. 
Instead of a single mimicking network for $\tnode$, 
we will store a set of small networks, each of them mimicking a particular 
snapshot of a solution. We call them \emph{projection networks}. Since a snapshot 
of an $s$-$t$ even path in $T$ can possibly be a set of disjoint paths between the (upto) four 
terminals in $T$, we require the 
$\DPTP$ routine of lemma~\ref{lem:3dpp} to compute these projection networks. 
We combine the parent piece with each possible projection network. 
The merged piece will again be either planar or of bounded treewidth, allowing us to 
continue this operation 
towards the root node until a single
piece containing both $s$ and $t$ remains. We query
for an $s$-$t$ even path in this piece and output yes iff there exists one. 
At each step, the number of projection networks used to replace the leaf piece, 
and their combinations with its parent piece will remain bounded 
by a constant number. 

%
Once we have the decision version of $\EP$, we show a poly-time self-reduction 
using the decision oracle of $\EP$ to construct a solution, 
in~\cref{app:search_to_dec}. 
\paragraph*{Necessity of a two phased approach}
We mention why we have two phases and different technical ingredients for each.
\begin{itemize}
 \item Instead of a single parity mimicking network, we need a set of projection
  networks for the leaf piece in the second phase because it
  can have upto four terminals (three vertices of the separating
  clique and the vertex $t$), and we do not yet know how to find
  (or even the existence of) parity mimicking networks with the 
  constraints we desire, for graphs with four terminals. 
 \item We cannot however use a set of networks for each piece in phase I because of 
  the unbounded degree of $\T$.
  Suppose a branch piece $G_i$ is connected to its parent piece by clique $c$, and suppose $G_i$
  has child pieces $L_1,L_2,\ldots, L_{\ell}$, attached to $G_i$ via disjoint cliques
  $c_1,c_2,\ldots, c_{\ell}$, respectively.
  An even $s$-$t$ path can enter $G_i$ via a vertex of $c$, then visit
  any of $L_1,L_2,\ldots, L_{\ell}$ in any order and go back to the 
  parent of $G_i$ via another vertex of $c$. 
  If we store information regarding parity configurations of $L_1,L_2,\ldots, L_{\ell}$ as
  sets of projection networks, we could have to try exponentially many combinations to compute
  information of parity configurations between vertices of $c$ in the subtree rooted at $G_i$
  (note that $\ell$ could be $O(n)$). 
  Therefore, we compress the information related to parity configurations of
  $L_i$ into a single parity mimicking network $L'_i$, while preserving solutions,
  so that the combined graph $(((G_i\csum{c_1} L'_1)\csum{c_2} L'_2\ldots)
 \csum{c_{\ell}} L'_{\ell})$ is either planar or of bounded treewidth.
 \end{itemize}
We will now describe these ingredients formally in the remaining part of this section.
\subsection{Parity Mimicking Networks}
We first define the parity configuration of a graph, which consists of subsets of 
$\{0,1\}$ for each pair, triplet of terminals, depending on whether there exists
`direct' or `via' paths of parity even, odd, or both 
(we use $0$ for even parity and $1$ for odd). 
We formalise this below.
\begin{definition}\label{def:parity_config}
  Let $\lnode$ be a directed graph and $T(L)=\{v_1,v_2,v_3\}$ be the set of
 terminal vertices of $\lnode$.
 Then, $\forall i,j,k \in \{1,2,3\}$, such that $i,j,k$ are distinct, 
 we define the sets $\dir{\lnode}{v_i}{v_j}$, and 
 $\via{\lnode}{v_i}{v_k}{v_j}$ as:
 \begin{itemize}
   \item $\dir{\lnode}{v_i}{v_j} =$\{$p$ $\mid$ there exists a path of 
    parity $p$ from $v_i$ to $v_j$ in $L-v_k$ \}  
   \item $\via{\lnode}{v_i}{v_k}{v_j} =$\{$p$ $\mid$ there exists a path of 
    parity $p$ from $v_i$ to $v_j$ via $v_k$ \emph{and}\\ 
    \hspace{5cm} there does not exist a path of parity p from $v_i$ to 
    $v_j$  in  $L - v_k\}$  
 \end{itemize}
We say that the $\dirs{\lnode},\vias{\lnode}$ sets constitute the 
\emph{parity configuration} of the graph $\lnode$ with respect to $T(L)$. 
We call the paths corresponding to elements in $\dirs{\lnode},\vias{\lnode}$ 
sets as \emph{Direct paths} and \emph{Via paths}, respectively. 
\end{definition}
The parity configuration of a graph can be visualised as a table. 
We have defined it for three terminals, it can be defined in a 
similar way for two terminals. 
It is natural to ask the question that given a parity 
configuration $\parConf$ independently with respect to some terminal vertices, 
does there exist a graph with those terminal vertices, 
realising that parity configuration. If not, we say that $\parConf$ is unrealisable.
It is easy to see that the number of parity configurations for a set of
three terminals is bounded by $4^{12}$, many of which are unrealizable.
We now define parity mimicking networks.
\begin{definition}\label{def:par_mim_net}
  A graph $\lnodex$ is a parity mimicking network of a another graph $\lnode$ 
  (and vica versa),
  if they share a common set of terminals, and have the same parity 
  configuration, $\parConf$, w.r.t. the terminals. 
  We also call them parity mimicking networks of parity configuration $\parConf$.
\end{definition}
\begin{figure}
  \includegraphics[scale=0.7]{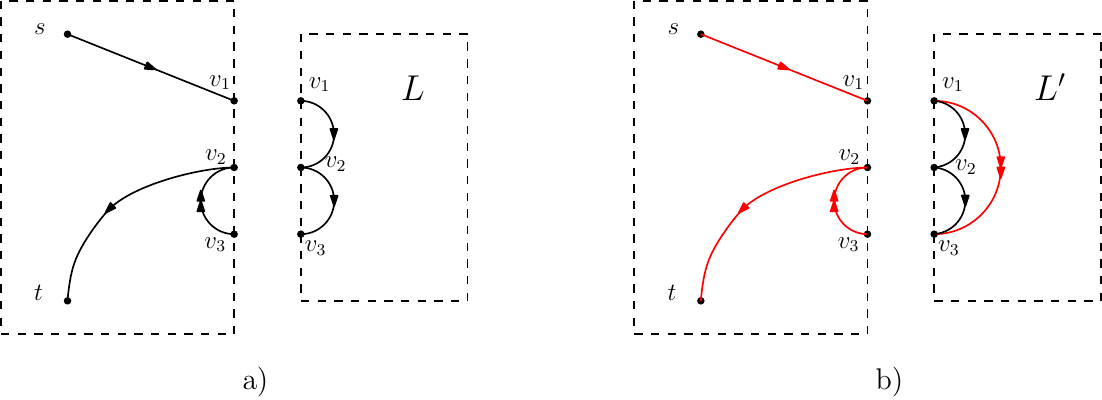}
\caption{Figure a) shows the input graph and b) shows the graph with $L$ replaced by an 
   erroneous mimicking network $L'$. 
   Suppose the original graph in a) has no $s$-$t$ even path 
   but does have an $s$-$t$ even walk as shown in the figure, using vertex $v_2$ twice. 
   If we query for a path from $v_1$ to $v_3$ in $L$, and add a direct $v_1$ 
   to $v_3$ path of that parity in $L'$, 
   we end up creating a false solution since $v_2$ is freed up to be used 
   outside $L'$. Hence there must
   be equality between corresponding direct sets.}\label{fig:via1} 
\end{figure}
The reason we differentiate between direct paths and via paths while defining parity
configurations is to ensure that no false solutions are introduced on replacing a 
leaf piece of $\T$ by its mimicking network (see~\cref{fig:via1}).
Note that in our definition of \emph{Via} sets, 
we exclude parity entries of via paths between two terminals if that parity is already 
present in \emph{Dir} set between the same terminals. 
We do so because this makes the parity configurations easier to 
enumerate in our construction of parity mimicking networks.
In~\cref{fig:via2}, we describe why doing this will still preserve 
solutions.

We also need to consider the case where multiple leaf pieces, in $\T$ are attached to a 
common parent piece via a shared clique (as seen in \cref{fig:clique_decomp}). 
In this case, we will replace the entire subgraph corresponding to the clique sum
of the sibling leaf pieces by one parity mimicking network.
To compute the parity configuration of the combined subgraph of leaf pieces, we make
the following observation:
\begin{restatable}{myObs}{obsMergeSiblings}\label{obs:obs1}
 Let $L_1, L_2, \ldots, L_{\ell}$ be leaf branch pieces that are 
 pairwise disjoint except for a common set 
 of terminal vertices, say $\{v_1,v_2,v_3\}$. 
 Let $L= L_1\oplus L_2 \oplus \ldots \oplus L_{\ell}$. Then
 the parity configuration of $L$ with respect to $\{v_1,v_2,v_3\}$ can be computed by:
 \begin{flalign}
  &\dir{\lnode}{v_i}{v_j} = \bigcup\limits_{a=1}^{\ell} \dir{L_a}{v_i}{v_j}\\
  &\via{\lnode}{v_i}{v_k}{v_j} =  \nonumber\\
  &\left (\bigcup\limits_{a=1}^{\ell} \via{L_a}{v_i}{v_k}{v_j}
   \cup \bigcup\limits_{a,b=1}^{\ell} 
   (\dir{L_a}{v_i}{v_k}\boxplus\dir{L_b}{v_k}{v_j}) \right) \backslash \dir{L}{v_i}{v_j} 
 \end{flalign}
 where $A \boxplus B$ denotes the set formed by addition modulo $2$ between 
 all pairs of elements 
 in sets $A,B$, and $i,j,k\in \{1,2,3\}$ are distinct.
\end{restatable}
The intuition behind the observation is simple.
Any direct path in $L$ from $v_i$ to $v_j$
must occur as a direct path in one of $L_1,L_2\ldots L_{\ell}$ since they are disjoint except
for terminal vertices. Any via path in $L$ from $v_i$ to $v_j$ via $v_k$ can occur in two ways,
either as a $v_i$-$v_k$-$v_j$ via path in one of
$L_1,L_2\ldots L_{\ell}$, or as a concatenation of two
direct paths, one from $v_i$ to $v_k$ in some piece $L_i$,
and another from $v_k$ to $v_j$ in another piece $L_j$.
Note that although the observation is for the case when all $L_1, L_2, \ldots, L_{\ell}$ share
a common $3$-clique $\{v_1,v_2,v_3\}$, it is easy to see it can be tweaked easily to handle the
cases when some of the $L_i's$ are attached via a $2$-clique that is a subset of the $3$-clique.

The next lemma states that replacing leaf piece nodes 
in $\T$ by parity mimicking networks obeying some planarity conditions, 
will preserve the existence of $s$-$t$ paths of any particular parity, 
and also preserve conditions
on treewidth and planarity for the combined piece.

\begin{restatable}{lemmma}{parMimCorrectness}\label{lem:par_mim_net_correct}
  Let $G$ be a graph with clique sum decomposition tree $\T$, and let $L_1,L_2\ldots,L_{\ell}$ 
be set of leaf branch pieces of $\T$, attached to their 
parent piece $G_1$ via a common clique $c$. Let $\lnodex$ be a parity mimicking 
network of $L_1\oplus L_2 \oplus \ldots L_{\ell}$ with respect to $c$, such that $\lnodex$ is
planar, and vertices of $c$ lie on a common face in $\lnodex$. 
Then: 
 \begin{enumerate}
  \item There is a path of parity $p$ from $s$ to $t$ in $G$ iff there is a 
    path of parity $p$ from $s$ to $t$ in $G[L_1 \oplus L_2 \oplus \ldots 
    L_{\ell} \rightarrow \lnodex]$.
  \item If $G_1$ is planar, then $G_1\oplus \lnodex$ is also planar.
  \item If $G_1$ has treewidth $\twh$, and $\lnodex$ has treewidth $\tau_{\lnodex}$,
   then $G_1\oplus \lnodex$ has treewidth max($\twh,\tau_{\lnodex}$)
 \end{enumerate}
\end{restatable}
 \begin{proof}
 \begin{enumerate}
   \item The proof essentially follows from the definition of parity mimicking networks and
     observation~\ref{obs:obs1}, since we can replace the snapshot of any $s$-$t$ path $P$
     in $L_i$ by a path of corresponding parity in $L'_i$ and vice-versa. 
   \item This follows since in the decomposition, vertices of separting cliques in every piece lie
     on the same face, and so is the case for $L'$ by assumption. 
     Therefore we can embed $L'$ inside
     the face in $G_1$, on the boundary of which $v_1,v_2,v_3$ lie.
   \item This follows since we can merge tree decompositons of $G_1,L'$ along bags
     consisting of the common clique.
 \end{enumerate}
\end{proof}
\begin{figure}
  \includegraphics[scale=0.7]{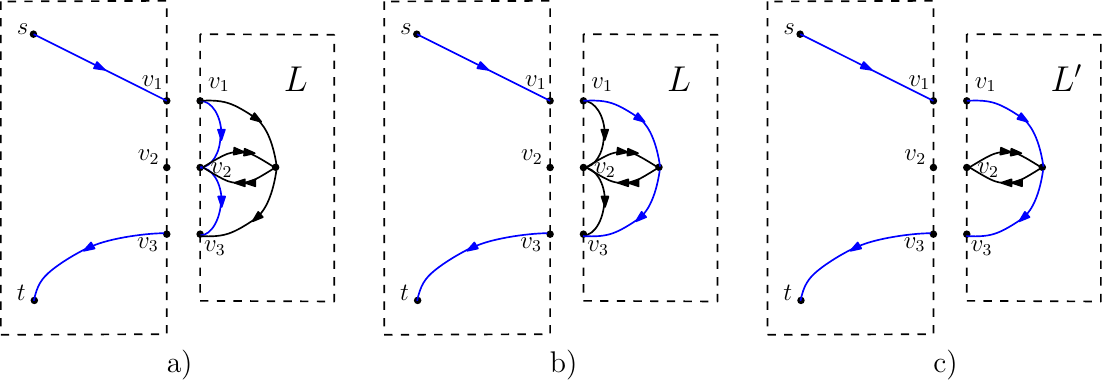}
  \caption{Figure a) denotes the original graph which has both a direct path,
  as well as a via path of even parity from $v_1$ to $v_3$. 
  Suppose the via path is part of an even
  $s$-$t$ path solution, as marked by blue. Then in $L$ itself, we could replace the
  via path by the direct path and it would still be a valid even $s$-$t$ path, as marked
  in blue in b). Hence in the mimicking network $L'$, too (shown in c)), 
  we could use the direct $v_1$ to $v_3$ path of the
  same parity. Therefore we do not need to put the parity of the $v_1$-$v_2$-$v_3$ path in 
  $\via{L}{v_1}{v_2}{v_3}$, since the same parity is already present in 
  $\dir{L}{v_1}{v_3}$, and $\dir{L}{v_1}{v_3} = \dir{L'}{v_1}{v_3}$.}\label{fig:via2}
\end{figure}
Now we will show how to compute parity mimicking networks that are small 
in size (and hence of bounded treewidth), and also planar, with terminal vertices 
lying on the same face, for a given parity configuration of a graph $\lnode$. 
\begin{figure}[h]
\begin{minipage}{\textwidth}
 \includegraphics[scale=0.56]{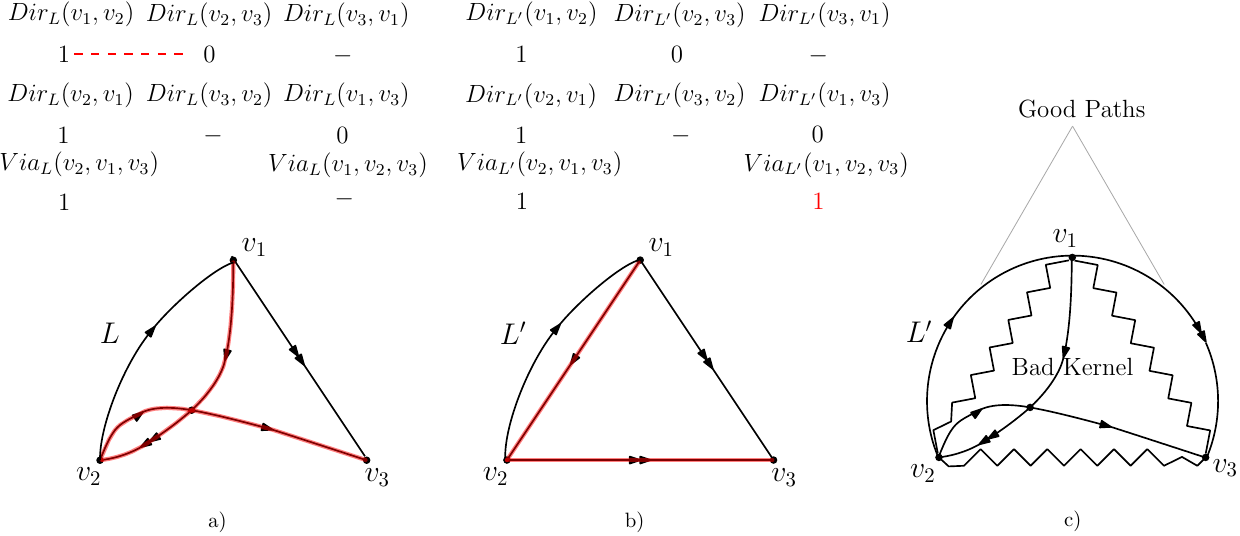}
 \caption{Fig a) denotes a graph $L$ with its parity configuration table (only relevant sets). 
 Fig b) denotes a `parity mimicking network', if for each pair of terminals, 
 we just independently put 
 paths of correct parity, disjoint from each other. It leads to an extra 
 path (highlighted in red) from 
 $v_1$ to $v_3$ via $v_2$ in $L'$, of odd parity. 
 Pairs of such entries, for which we cannot add disjoint paths are called bad entries as marked
 by the dashed red line in the parity configuration table in a). Fig c) outlines 
 the approach used to construct the correct mimicking network. The two paths corresponding to
 bad pair entries, form the bad kernel, for which we construct a mimicking network 
 by enumerating cases. The remaining paths
 can be added iteratively, disjoint from all existing paths, 
 on the outer face.}\label{fig:bad_kernel_eg}
\end{minipage}
 \begin{minipage}{\textwidth}
  \includegraphics[scale=0.55]{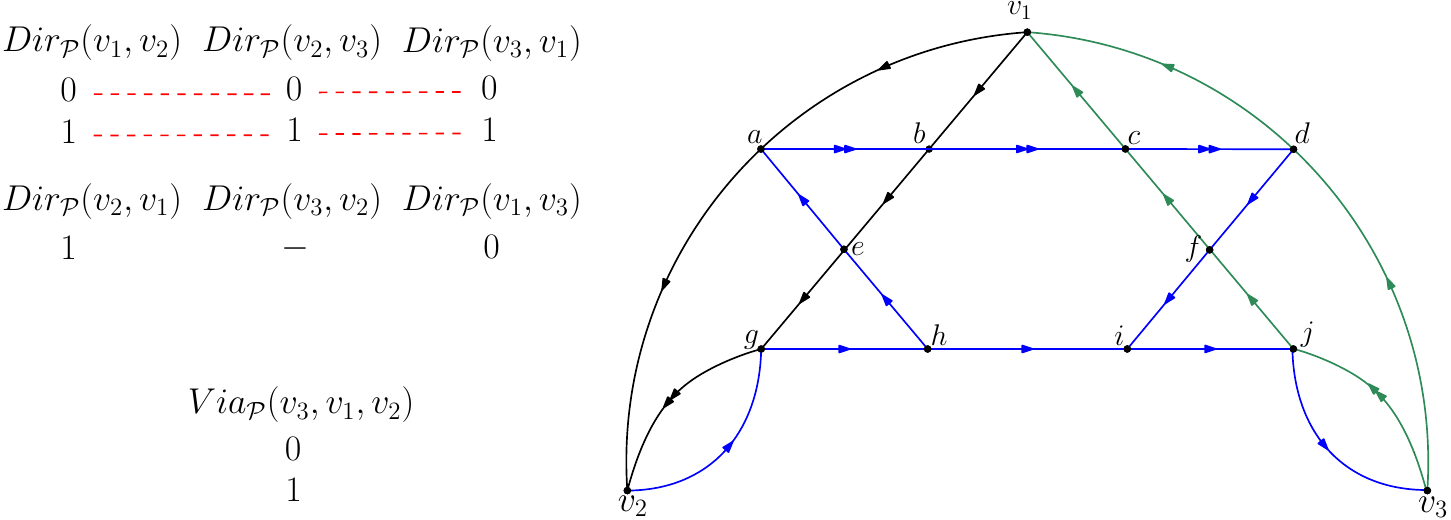}
  \caption{An example of a more non-trivial bad kernel, and a mimicking network realising 
  its closure. This is a subcase of case $(4,0)$ described in the full proof 
  in~\cref{app:par_mim_net_proof}. We give a list of paths along with their lengths, 
  for ease of reader to check that the network obeys the parity configuration.} 
 \end{minipage}
\end{figure}
\begin{restatable}{lemmma}{abc}\label{lem:par_mim_net}
  Suppose $\lnode$ is a graph with terminals $T(L) = \{v_1,v_2,v_3\}$, 
  and suppose we know the
  parity configuration of $\lnode$ with respect to $\{v_1,v_2,v_3\}$. 
  We can in polynomial-time, find a parity mimicking network $\lnodex$ of 
  $\lnode$, with respect to $\{v_1,v_2,v_3\}$ which 
 consists of at most $18$ vertices, and is also planar, with 
 $v_1,v_2,v_3$ lying on a common face.
\end{restatable}
\begin{proof}
 We give a brief idea of the proof and defer the full proof to~\cref{app:par_mim_net_proof}.
 As noted above, the number of possible parity configurations are finite (bounded by $4^{12}$
 for three terminals), but the number is too large to 
 enumerate over all of them and individually construct the mimicking networks.
 We use some observations to make the case analysis tractable.
 We refer to elements of sets $\dir{\lnode}{v_i}{v_j}$ as \emph{entries}. But we 
 abuse notation slightly and distinguish them from the boolean values $0,1$. 
 For example, we always distinguish between an entry of $\dir{\lnode}{v_1}{v_2}$, 
 and an entry of $\dir{\lnode}{v_2}{v_3}$, even if they have the same 
 \emph{value} ($0$ or $1$).
 A natural constructive approach would be to iteratively do the following 
 step for all $i,j,p$ : add a path of length $2-p$ from  
 $v_i$ to $v_j$ in $\lnodex$, disjoint from existing paths of $\lnodex$, 
 if there is an entry of parity $p$ present in \emph{$\dir{\lnodex}{v_i}{v_j}$}. 
 Its easy to check that this will result in a
 planar $\lnodex$ with terminals on a common face. However, this could lead to wrong parity 
 configurations in $\lnodex$. For example, $\lnode$ could have a direct paths of parity 
 $1$ from $v_1$ to $v_2$, and a direct path of parity $0$ from $v_2$ to $v_3$,  
 but no path of parity $1$ from $v_1$ to $v_3$, either direct or
 via $v_2$ (see~\cref{fig:bad_kernel_eg}). We will call pairs of such entries 
 as \emph{bad pairs}. The entries that are part of \emph{any}
 bad pair are called \emph{bad entries}. Though the example in~\cref{fig:bad_kernel_eg} 
 has a simple fix for the bad pair, it becomes more complicated to maintain the 
 planarity conditions as the number of bad pairs increase.
 Let $\parConf_{\lnode}$ be the parity configuration of $L$.
 The idea of the proof is to define a \emph{bad kernel} of  
 $\parConf_{\lnode}$, as the \emph{sub-configuration} consisting 
 of all the \emph{bad entries} of $\parConf_{\lnode}$.
 The \emph{closure} of the bad kernel is defined as the parity 
 configuration obtained from it by adding `minimal' number of entiries 
 to make it realizable. We observe that the closure remains a sub-configuration 
 of $\parConf_{\lnode}$.
 Suppose that we can somehow construct a planar mimicking network for the 
 \emph{closure} of the \emph{bad kernel} of $\parConf_{\lnode}$, 
 with terminals lying on a common face. 
 Then we show that paths corresponding to leftover parity entries of $\lnode$ 
 can be safely added using the constructive approach described above. 
 Hence it suffices to construct parity mimicking networks for closures of 
 all possible bad kernels. We use some observations to show that the number of possible 
 types of bad kernels
 cannot be too large, and enumerate over each type, explicitly constructing
 the parity mimicking networks of their closures.
\end{proof}
\subsection{Disjoint Paths with Parity Problem}\label{sec:3dpp}
In this section, we will define and solve the $\DPTP$ problem for some 
special cases and types of graphs.  
We define the problem for three paths between four terminals.
\begin{definition}
  Given a graph $G$ and four distinct terminals $v_1, v_2, v_3,$ and $v_4$ in $V(G)$, the 
  $\DPTP$ problem is to find a set of three pairwise disjoint 
  paths, from $v_1$ to $v_2$, $v_2$ to $v_3$, and from $v_3$ to $v_4$, such that the total parity
  is even, if such a set of paths exist, and output no otherwise.
\end{definition}	
The problem where total parity must be odd can be easily reduced to this by adding a dummy
neighbour to $v_4$. The problem is $\NP$-$\textsf{hard}$ in general graphs since the even path
problem trivially reduces to this.
We show that the above problem can be solved in polynomial time in following
two cases:
\begin{restatable}{lemmma}{threeDpp}\label{lem:3dpp}
  Let $G$ be a graph, and $v_1,v_2,v_3,v_4$ be four vertices of $G$. 
 Both decision as well as search 
 versions of $\DPTP$ for these vertices as defined above can be solved 
 in polynomial time in the following cases:
 \begin{enumerate}
  \item If $G$ has constant treewidth. 
  \item If $G$ is planar and $v_1,v_2,v_3$ lie on a common face of $G$. 
 \end{enumerate}
\end{restatable}
\begin{proof}
 Proofs of both parts can be found in~\cref{app:3dpp_planar},~\cref{app:3dpp_btw}. 
 We give a high level idea of the proof of the second part. The argument of Nedev 
 for $\EP$ uses two main lemmas. One lemma states that if there are two paths $P_1,P_2$, of 
 different parities from $s$ to $t$, then their union forms a (at least one) 
 structure, which they call an \emph{odd list superface}. 
 It (roughly) consists of two internally disjoint paths 
 of different parities, with a common starting vertex, say $b$ and a common 
 ending vertex, say $e$. Let $F$ denote such a superface. 
 They show that there exist two disjoint paths in 
 $P_1\cup P_2 - F$, one from $s$ to $b$, and one from $e$ to $t$. 
 This provides a `switch' in $P_1 \cup P_2$, and if we can 
 find this switch efficiently, then we can use existing $2$-disjoint path algorithms 
 to connect $s$ and $t$ via this switch. But the number of odd list superfaces in a 
 graph can be exponential. The second lemma of Nedev says that we can exploit the 
 structure of planarity and show that each of the odd list superfaces formed by 
 $P_1\cup P_2$, `contain' a `minimal' odd list superface, which they call a 
 \emph{simple} odd list superface, that obeys the 
 same conditions. The set of simple odd list superfaces is small and can be 
 enumerated in polynomial time. In our setting, we start from the case that two 
 instances of three disjoint paths between the specified terminals exist, 
 such that they have different total parity. Say the instances are $P_1,P_2,P_3$, 
 and $P'_1,P'_2,P'_3$. At least one of $P_i,P'_i$ must be of different parity. 
 We show that using the constraints of three terminals on a face, and using 
 ideas of $\emph{leftmost}$ (and rightmost) paths of $P_i\cup P'_i$, for each 
 case of $i\in \{1,2,3\}$, there does exist an analogous structure: a simple 
 odd list super face, and four disjoint path segments connecting the required 
 vertices.
 A point to note is that in Nedev's argument, \emph{any} odd list 
 superface formed by $P_1,P_2$ could be trimmed to a simple odd list superface that 
 would give a valid solution. That does not hold true here. 
 We generalise their lemma, and argue that there does exist \emph{at least one} odd 
 list superface between $P_i,P'_i$ that will work in our setting.
\end{proof}
\section{Main Algorithm}\label{sec:main_algo}
We now explain the two phases of the algorithm.
\subsection{Phase 1}
\begin{enumerate}
 \item Find the $3$-clique sum decomposition tree $\T$. 
  Mark the piece that contains the vertex $s$ as the root of $\T$.
 \item Pick any maximal set of leaf branch pieces of $\T$, say 
   $L_1,L_2,\ldots, L_{\ell}$, which are attached to a parent piece $G_i$ via a common clique.
  Compute their parity configurations using Nedev's algorithm, or using Courcelle's
    theorem, as explained in Appendix~\ref{app:3dpp_btw}. 
  Then compute the parity configuration of $L_1 \oplus L_2 \oplus \ldots \oplus L_{\ell}$
  using observation~\ref{obs:obs1}.
\item Compute the parity mimicking network, $\lnodex$, of $L_1 \oplus L_2 \oplus 
    \ldots \oplus L_{\ell}$ using lemma~\ref{lem:par_mim_net}. 
    Replace $L_1 \oplus L_2 \oplus 
  \ldots \oplus L_{\ell}$ by $\lnodex$ and merge it with $G_i$.
 \item Since $G_i \oplus \lnodex$ is either of bounded treewidth
  or is planar by lemma~\ref{lem:par_mim_net_correct}, we can 
  repeat this step until no branch pieces remain.
\end{enumerate}
\subsection{Phase II}\label{subsec:phase_2}
Let $G'$ denote the graph after phase I. After phase I, the modified tree $\Ty$ 
looks like a path of pieces, $G_1,G_{2},\ldots, G_m$, joined
at cliques $c_1,c_2\ldots c_{m-1}$.\footnote{Note that the vertices of a clique, 
say $c_i$ need no longer lie on the same face of $G_i$ after phase I, 
since we might have merged the parity mimicking network 
of the branch pieces incident at $c_i$ into the face corresponding to $c_i$.}
The vertex $s$ is in root piece $G_1$, and $t$ in leaf piece $G_m$ (we use $G_1,G_m$ instead 
of $S,T$ here for notational convenience). 
We can write $G'=G_1\csum{c_1} G_{2}\csum{c_2} \ldots 
\csum{c_{m-1}} G_m$. Since it is clear in this phase that $c_i$ is the 
clique joining $G_{i},G_{i+1}$, we will omit the subscript for notational 
convenience and just write $G_1+ G_2+\ldots +G_m$ instead.  
Let $c_{m-1}=\{v_1,v_2,v_3\}$ and let $i,j,k\in \{1,2,3\}$ be distinct.
The snapshot of any even $s$-$t$ path $P$ in $G_m$, can be one of the following four types
(see figure~\ref{fig:phase2}):
\begin{itemize}
 \item Type $1$ : A path from $v_i$ to $t$ without using $v_j,v_k$.
 \item Type $2$ : A path from $v_i$ to $t$ via $v_j$, without using $v_k$.
 \item Type $3$ : A path from $v_i$ to $t$ via $v_j,v_k$.
 \item Type $4$ : A path from $v_i$ to $v_j$ and a path from $v_k$ to $t$, 
   both disjoint from each other.
\end{itemize}
We call any path/set of paths in $G_m$ of one of the above types as a 
\emph{potential snapshot} of $G_m$. 

We now construct the \emph{projection networks} of potential snapshots of $G_m$. 
\begin{definition}
  Let $G_m$ be the leaf piece as described above with clique $c_{m-1} = \{v_1,v_2,v_3\}$,
  and vertex $t$ present in $G_m$.
  \begin{itemize}
    \item For each of the types described above, for all $i\in \{1,2,3\}$, 
      and for all $p\in \{0,1\}$,
      find a potential snapshot (if it exists) in $G_m$ from $v_i$ to $t$, 
      of total parity $p$, using lemma~\ref{lem:3dpp}. 
    \item Let $J$ be a potential snapshot found in the previous step, 
      Its \emph{projection network}, is defined as the \emph{graph} obtained from
      $J$ by keeping terminal vertices intact, and replacing every terminal to 
      terminal path in $J$ by a path of length $2-p$.  
   \end{itemize}   
  The type of the projection network is the type of the corresponding potential snapshot.\\
  The \emph{set of projection networks} of $G_m$, denoted by $\projset{G_m}$, 
  is the set of all projection networks obtained for $G_m$ by the above procedure. 
\end{definition}
See~\cref{fig:phase2} for an example. 
Since the total number of terminals is at most $4$ (with one fixed as $t$), it is easy to see 
that the number of possible projections networks for $G_m$ is bounded. 
Therefore $\projset{G_m}$ can be computed in polynomial time.
Note that $\projset{G_m}$ is not uniquely defined. 
But it is sufficient for our purpose, to compute any one of the various 
possible choices of the set $\projset{G_m}$ as explained in~\cref{fig:phase2}. 
\begin{figure}
  \includegraphics[scale=0.6]{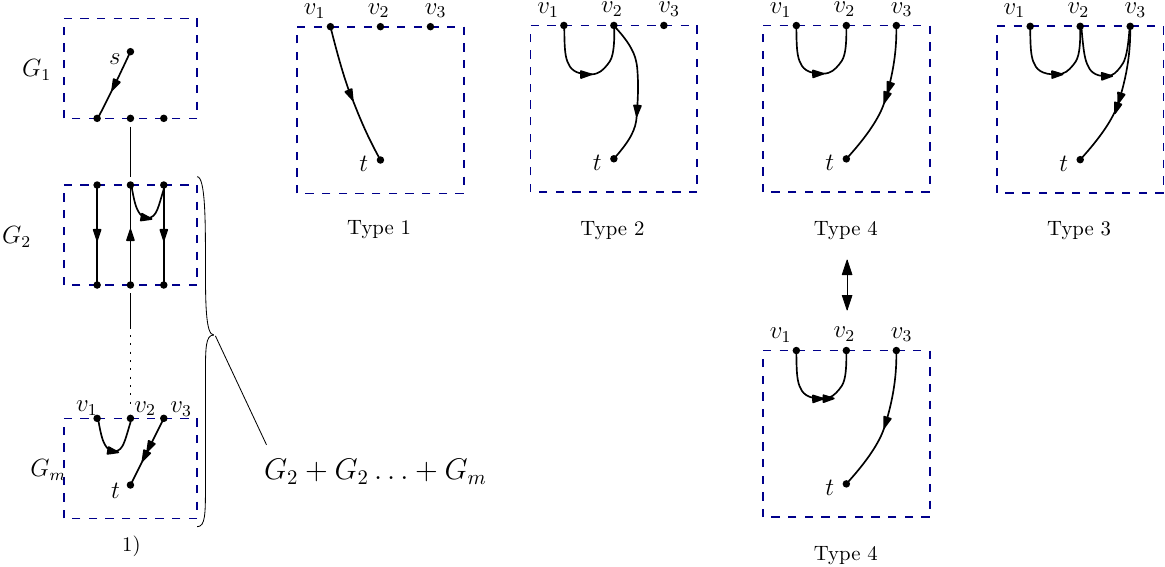}
  \caption{ Fig 1) Denotes the decomposition tree after phase $1$, with $G_1,G_2\ldots,G_m$ 
  denoting the pieces. We skip drawing clique nodes here. On the right are examples of 
  projection networks of different types. 
  In fig 1), the snapshot of the
  $s$-$t$ path in $G_m$ is of type $4$. 
  The two projection networks of type $4$ drawn on the right have 
  same total parity, but different parities of
  individual segments. 
  It is sufficient for our purpose to find any one of them since 
  they are interchangeable.}
  \label{fig:phase2}
\end{figure}
The next lemma shows that the projection networks of $G_m$ preserve solutions, 
and also maintain invariants on planarity and treewidth, when merged with the 
parent piece.
\begin{restatable}{lemmma}{projProp}\label{lem:proj_prop}
 Given $G'= G_1+\ldots +G_m$ as described above. 
  \begin{enumerate}
  \item Given a $\projset{G_m}$, there is an $s$-$t$ path in $G'$ of parity $p$ iff $\exists
    N \in \projset{G_m}$ such that $G'[G_m\rightarrow N]$ has an $s$-$t$ path of 
    parity $p$.
   \item If $G_{m-1}$ is planar/of bounded treewidth, then for 
     any projection network $N\in \projset{G_m}$, 
   $G_{m-1}+N$ is planar/of bounded treewidth, respectively.
  \end{enumerate}
\end{restatable}
\begin{proof}
 \begin{enumerate}
   \item This follows from the definition of projection networks. 
    The only minor technical point to note is that $\projset{G_m}$ is not
    unique. For example, suppose there are two potential
    snaphots in $G_m$ of type 4. One is $J_1$, consisting of a path $P_1$ from $v_1$ to $v_2$
    of even parity, and a path $P_2$ from $v_3$ to $t$ of odd parity. The other is $J_2$,
    consisting of a path $P'_1$ from $v_1$ to $v_2$ of odd parity, and a path $P'_2$ from $v_3$ to
    $t$, of even parity. Since the total parity of $J_1$ and $J_2$ is same,~\cref{lem:3dpp} could
    output either one of them. We don't have control over it to find both. But finding any one of
    them is sufficient for us, since if $J_1$ is a snapshot of an actual solution, then replacing $J_1$
    by $J_2$ would also give a valid solution and vice versa (See~\cref{fig:phase2}). 
   \item Suppose $c_{m_1}=\{v_1,v_2,v_3\}$ is the clique where $G_{m-1},G_m$ are attached. 
    The argument of treewidth bound
    is same as that of~\cref{lem:par_mim_net_correct} in previous phase when we attached mimicking networks to parent pieces. 
    However if $G_{m-1}$ is planar, there could have been a parity mimicking 
    network $L'$ attached to $G_{m-1}$ via $c_{m-1}$ during phase I.
    Hence $v_1,v_2,v_3$ might not lie on a common face in $G_{m-1}$ after phase I. 
    We observe however, 
    since $L'$ was attached at a $3$-clique, $c_{m-1}$, 
    every pair $v_i,v_j$ of vertices of $c_{m-1}$, must share a common face in $G_{m-1}$.
    Now, the projection networks consist of at most three paths, 
    two between $v_1,v_2,v_3$, and one from them to $t$. 
    For any $v_i,v_j$, we can embed the path between $v_i,v_j$ in $N$, in the face in 
    $G_{m-1}$ shared by $v_i,v_j$, and finally just add the path leading to $t$. 
    Therefore if $G_{m-1}$ is planar, all projection networks of $G_m$ 
    can be embedded in their parent nodes. 
 \end{enumerate} 
\end{proof}
 We make the following observation to compute a $\projset{G_i+\ldots G_m}$ recursively:
\begin{equation}
 \projset{G_i+\ldots G_m} = \bigcup\limits_{N\in 
 \projset{G_{i+1}+\ldots G_m}} \projset{G_i + N} 
\end{equation}
Thus we can proceed as follows:
\begin{enumerate}
 \item Compute $\projset{G_m}$ using lemma~\ref{lem:3dpp}
 \item For all $N\in \projset{G_m}$, compute $G_{m-1}+N$, and hence compute
   $\projset{G_m+G_{m-1}}$ using the observation above.
 \item For all $N\in \projset{G_m+G_{m-1}}$, compute $G_{m-2}+N$, and hence compute
   $\projset{G_m+G_{m-1}+G_{m-2}}$. Repeat until we reach $G_1$.
\end{enumerate}

\bibliographystyle{plainurl}
\bibliography{references}

\begin{thebibliography}{10}

\bibitem{AV20}
Nima Anari and Vijay~V. Vazirani.
\newblock Planar graph perfect matching is in nc.
\newblock {\em J. ACM}, 67(4), may 2020.
\newblock URL: \url{https://doi.org/10.1145/3397504}, \href
  {http://dx.doi.org/10.1145/3397504} {\path{doi:10.1145/3397504}}.

\bibitem{AGGT14}
Rahul Arora, Ashu Gupta, Rohit Gurjar, and Raghunath Tewari.
\newblock Derandomizing isolation lemma for \$k\_\{3, 3\}\$-free and
  \$k\_5\$-free bipartite graphs.
\newblock In {\em Symposium on Theoretical Aspects of Computer Science}, 2014.
\newblock URL: \url{https://api.semanticscholar.org/CorpusID:1484963}.

\bibitem{BHK22}
Andreas Bj\"{o}rklund, Thore Husfeldt, and Petteri Kaski.
\newblock The shortest even cycle problem is tractable.
\newblock In {\em Proceedings of the 54th Annual ACM SIGACT Symposium on Theory
  of Computing}, STOC 2022, page 117–130, New York, NY, USA, 2022.
  Association for Computing Machinery.
\newblock URL: \url{https://doi.org/10.1145/3519935.3520030}, \href
  {http://dx.doi.org/10.1145/3519935.3520030}
  {\path{doi:10.1145/3519935.3520030}}.

\bibitem{BK09}
Glencora Borradaile and Philip Klein.
\newblock An o(n log n) algorithm for maximum st-flow in a directed planar
  graph.
\newblock {\em J. ACM}, 56(2), apr 2009.
\newblock URL: \url{https://doi.org/10.1145/1502793.1502798}, \href
  {http://dx.doi.org/10.1145/1502793.1502798}
  {\path{doi:10.1145/1502793.1502798}}.

\bibitem{CT15}
Diptarka Chakraborty and Raghunath Tewari.
\newblock Simultaneous time-space upper bounds for red-blue path problem in
  planar dags.
\newblock In M.~Sohel Rahman and Etsuji Tomita, editors, {\em {WALCOM:}
  Algorithms and Computation - 9th International Workshop, {WALCOM} 2015,
  Dhaka, Bangladesh, February 26-28, 2015. Proceedings}, volume 8973 of {\em
  Lecture Notes in Computer Science}, pages 258--269. Springer, 2015.
\newblock URL: \url{https://doi.org/10.1007/978-3-319-15612-5\_23}, \href
  {http://dx.doi.org/10.1007/978-3-319-15612-5\_23}
  {\path{doi:10.1007/978-3-319-15612-5\_23}}.

\bibitem{CE}
{Erin Wolf} {Chambers} and {David} {Eppstein}.
\newblock Flows in one-crossing-minor-free graphs.
\newblock {\em Journal of Graph Algorithms and Applications}, 17(3):201--220,
  2013.
\newblock \href {http://dx.doi.org/10.7155/jgaa.00291}
  {\path{doi:10.7155/jgaa.00291}}.

\bibitem{CSWZ}
Shiva Chaudhuri, K.~Subrahmanyam, Frank Wagner, and Christos Zaroliagis.
\newblock Computing mimicking networks.
\newblock {\em Algorithmica}, 26:31--49, 01 2000.
\newblock \href {http://dx.doi.org/10.1007/s004539910003}
  {\path{doi:10.1007/s004539910003}}.

\bibitem{COURCELLE}
Bruno Courcelle.
\newblock The monadic second-order logic of graphs. i. recognizable sets of
  finite graphs.
\newblock {\em Information and Computation}, 85(1):12--75, 1990.
\newblock URL:
  \url{https://www.sciencedirect.com/science/article/pii/089054019090043H},
  \href {http://dx.doi.org/https://doi.org/10.1016/0890-5401(90)90043-H}
  {\path{doi:https://doi.org/10.1016/0890-5401(90)90043-H}}.

\bibitem{CDPP13}
Marek Cygan, Daniel Marx, Marcin Pilipczuk, and Michal Pilipczuk.
\newblock The planar directed k-vertex-disjoint paths problem is
  fixed-parameter tractable.
\newblock In {\em 2013 IEEE 54th Annual Symposium on Foundations of Computer
  Science}, pages 197--206, 2013.
\newblock \href {http://dx.doi.org/10.1109/FOCS.2013.29}
  {\path{doi:10.1109/FOCS.2013.29}}.

\bibitem{DGKT12}
Samir Datta, Arjun Gopalan, Raghav Kulkarni, and Raghunath Tewari.
\newblock Improved bounds for bipartite matching on surfaces.
\newblock In Christoph D{\"{u}}rr and Thomas Wilke, editors, {\em 29th
  International Symposium on Theoretical Aspects of Computer Science, {STACS}
  2012, February 29th - March 3rd, 2012, Paris, France}, volume~14 of {\em
  LIPIcs}, pages 254--265. Schloss Dagstuhl - Leibniz-Zentrum f{\"{u}}r
  Informatik, 2012.
\newblock URL: \url{https://doi.org/10.4230/LIPIcs.STACS.2012.254}, \href
  {http://dx.doi.org/10.4230/LIPIcs.STACS.2012.254}
  {\path{doi:10.4230/LIPIcs.STACS.2012.254}}.

\bibitem{DNTW09}
Samir Datta, Prajakta Nimbhorkar, Thomas Thierauf, and Fabian Wagner.
\newblock Graph isomorphism for k\_\{3, 3\}-free and k\_5-free graphs is in
  log-space.
\newblock {\em Electron. Colloquium Comput. Complex.}, TR10, 2009.
\newblock URL: \url{https://api.semanticscholar.org/CorpusID:7978883}.

\bibitem{DHNRT04}
Erik~D Demaine, MohammadTaghi Hajiaghayi, Naomi Nishimura, Prabhakar Ragde, and
  Dimitrios~M Thilikos.
\newblock Approximation algorithms for classes of graphs excluding
  single-crossing graphs as minors.
\newblock {\em Journal of Computer and System Sciences}, 69(2):166--195, 2004.
\newblock URL:
  \url{https://www.sciencedirect.com/science/article/pii/S0022000004000194},
  \href {http://dx.doi.org/https://doi.org/10.1016/j.jcss.2003.12.001}
  {\path{doi:https://doi.org/10.1016/j.jcss.2003.12.001}}.

\bibitem{EV}
David Eppstein and Vijay~V. Vazirani.
\newblock Nc algorithms for computing a perfect matching and a maximum flow in
  one-crossing-minor-free graphs.
\newblock {\em SIAM Journal on Computing}, 50(3):1014--1033, 2021.
\newblock URL: \url{https://doi.org/10.1137/19M1256221}, \href
  {http://arxiv.org/abs/https://doi.org/10.1137/19M1256221}
  {\path{arXiv:https://doi.org/10.1137/19M1256221}}, \href
  {http://dx.doi.org/10.1137/19M1256221} {\path{doi:10.1137/19M1256221}}.

\bibitem{FHW80}
Steven Fortune, John Hopcroft, and James Wyllie.
\newblock The directed subgraph homeomorphism problem.
\newblock {\em Theoretical Computer Science}, 10(2):111--121, 1980.
\newblock URL:
  \url{https://www.sciencedirect.com/science/article/pii/0304397580900092},
  \href {http://dx.doi.org/https://doi.org/10.1016/0304-3975(80)90009-2}
  {\path{doi:https://doi.org/10.1016/0304-3975(80)90009-2}}.

\bibitem{GL94}
Anna Galluccio and Martin Loebl.
\newblock Even/odd dipaths in planar digraphs.
\newblock {\em Optimization Methods and Software}, 3(1-3):225--236, 1994.
\newblock URL: \url{https://doi.org/10.1080/10556789408805566}, \href
  {http://arxiv.org/abs/https://doi.org/10.1080/10556789408805566}
  {\path{arXiv:https://doi.org/10.1080/10556789408805566}}, \href
  {http://dx.doi.org/10.1080/10556789408805566}
  {\path{doi:10.1080/10556789408805566}}.

\bibitem{GKR13}
Martin Grohe, Ken-ichi Kawarabayashi, and Bruce Reed.
\newblock A simple algorithm for the graph minor decomposition: logic meets
  structural graph theory.
\newblock In {\em Proceedings of the Twenty-Fourth Annual ACM-SIAM Symposium on
  Discrete Algorithms}, SODA '13, page 414–431, USA, 2013. Society for
  Industrial and Applied Mathematics.

\bibitem{Hagerup-et-al}
Torben Hagerup, Jyrki Katajainen, Naomi Nishimura, and Prabhakar Ragde.
\newblock Characterizing multiterminal flow networks and computing flows in
  networks of small treewidth.
\newblock {\em Journal of Computer and System Sciences}, 57(3):366--375, 1998.
\newblock URL:
  \url{https://www.sciencedirect.com/science/article/pii/S0022000098915926},
  \href {http://dx.doi.org/https://doi.org/10.1006/jcss.1998.1592}
  {\path{doi:https://doi.org/10.1006/jcss.1998.1592}}.

\bibitem{THstack}
Tony~Huynh (https://mathoverflow.net/users/2233/tony huynh).
\newblock Mso2-expressible graph properties unexpressible in mso1.
\newblock MathOverflow.
\newblock URL:https://mathoverflow.net/q/257819 (version: 2017-01-03).
\newblock URL: \url{https://mathoverflow.net/q/257819}, \href
  {http://arxiv.org/abs/https://mathoverflow.net/q/257819}
  {\path{arXiv:https://mathoverflow.net/q/257819}}.

\bibitem{KRW11}
Ken-ichi Kawarabayashi, Bruce Reed, and Paul Wollan.
\newblock The graph minor algorithm with parity conditions.
\newblock In {\em 2011 IEEE 52nd Annual Symposium on Foundations of Computer
  Science}, pages 27--36, 2011.
\newblock \href {http://dx.doi.org/10.1109/FOCS.2011.52}
  {\path{doi:10.1109/FOCS.2011.52}}.

\bibitem{KW11}
Ken-ichi Kawarabayashi and Paul Wollan.
\newblock A simpler algorithm and shorter proof for the graph minor
  decomposition.
\newblock In {\em Proceedings of the Forty-Third Annual ACM Symposium on Theory
  of Computing}, STOC '11, page 451–458, New York, NY, USA, 2011. Association
  for Computing Machinery.
\newblock URL: \url{https://doi.org/10.1145/1993636.1993697}, \href
  {http://dx.doi.org/10.1145/1993636.1993697}
  {\path{doi:10.1145/1993636.1993697}}.

\bibitem{KR14}
Arindam Khan and Prasad Raghavendra.
\newblock On mimicking networks representing minimum terminal cuts.
\newblock {\em Information Processing Letters}, 114(7):365--371, 2014.
\newblock URL:
  \url{https://www.sciencedirect.com/science/article/pii/S0020019014000350},
  \href {http://dx.doi.org/https://doi.org/10.1016/j.ipl.2014.02.011}
  {\path{doi:https://doi.org/10.1016/j.ipl.2014.02.011}}.

\bibitem{Khuller90_k5}
Samir Khuller.
\newblock Coloring algorithms for k5-minor free graphs.
\newblock {\em Information Processing Letters}, 34(4):203--208, 1990.
\newblock URL:
  \url{https://www.sciencedirect.com/science/article/pii/002001909090161P},
  \href {http://dx.doi.org/https://doi.org/10.1016/0020-0190(90)90161-P}
  {\path{doi:https://doi.org/10.1016/0020-0190(90)90161-P}}.

\bibitem{Khuller90_k33}
Samir Khuller.
\newblock Extending planar graph algorithms to k3,3-free graphs.
\newblock {\em Information and Computation}, 84(1):13--25, 1990.
\newblock URL:
  \url{https://www.sciencedirect.com/science/article/pii/089054019090031C},
  \href {http://dx.doi.org/https://doi.org/10.1016/0890-5401(90)90031-C}
  {\path{doi:https://doi.org/10.1016/0890-5401(90)90031-C}}.

\bibitem{KR13}
Robert Krauthgamer and Inbal Rika.
\newblock {\em Mimicking Networks and Succinct Representations of Terminal
  Cuts}, pages 1789--1799.
\newblock 2013.
\newblock URL:
  \url{https://epubs.siam.org/doi/abs/10.1137/1.9781611973105.128}, \href
  {http://arxiv.org/abs/https://epubs.siam.org/doi/pdf/10.1137/1.9781611973105.128}
  {\path{arXiv:https://epubs.siam.org/doi/pdf/10.1137/1.9781611973105.128}},
  \href {http://dx.doi.org/10.1137/1.9781611973105.128}
  {\path{doi:10.1137/1.9781611973105.128}}.

\bibitem{LP84}
Andrea~S. LaPaugh and Christos~H. Papadimitriou.
\newblock The even-path problem for graphs and digraphs.
\newblock {\em Networks}, 14(4):507--513, 1984.
\newblock URL: \url{https://doi.org/10.1002/net.3230140403}, \href
  {http://dx.doi.org/10.1002/net.3230140403}
  {\path{doi:10.1002/net.3230140403}}.

\bibitem{LMPSZ20}
Daniel Lokshtanov, Pranabendu Misra, Micha\l{} Pilipczuk, Saket Saurabh, and
  Meirav Zehavi.
\newblock An exponential time parameterized algorithm for planar disjoint
  paths.
\newblock In {\em Proceedings of the 52nd Annual ACM SIGACT Symposium on Theory
  of Computing}, STOC 2020, page 1307–1316, New York, NY, USA, 2020.
  Association for Computing Machinery.
\newblock URL: \url{https://doi.org/10.1145/3357713.3384250}, \href
  {http://dx.doi.org/10.1145/3357713.3384250}
  {\path{doi:10.1145/3357713.3384250}}.

\bibitem{Lynch75}
James~F. Lynch.
\newblock The equivalence of theorem proving and the interconnection problem.
\newblock {\em SIGDA Newsl.}, 5(3):31–36, sep 1975.
\newblock URL: \url{https://doi.org/10.1145/1061425.1061430}, \href
  {http://dx.doi.org/10.1145/1061425.1061430}
  {\path{doi:10.1145/1061425.1061430}}.

\bibitem{MRST97}
William McCuaig, Neil Robertson, P.~D. Seymour, and Robin Thomas.
\newblock Permanents, pfaffian orientations, and even directed circuits
  (extended abstract).
\newblock In {\em Proceedings of the Twenty-Ninth Annual ACM Symposium on
  Theory of Computing}, STOC '97, page 402–405, New York, NY, USA, 1997.
  Association for Computing Machinery.
\newblock URL: \url{https://doi.org/10.1145/258533.258625}, \href
  {http://dx.doi.org/10.1145/258533.258625} {\path{doi:10.1145/258533.258625}}.

\bibitem{Nedev99}
Zhivko~Prodanov Nedev.
\newblock Finding an even simple path in a directed planar graph.
\newblock {\em {SIAM} J. Comput.}, 29(2):685--695, 1999.
\newblock URL: \url{https://doi.org/10.1137/S0097539797330343}, \href
  {http://dx.doi.org/10.1137/S0097539797330343}
  {\path{doi:10.1137/S0097539797330343}}.

\bibitem{RRSS91}
Bruce Reed, Neil Robertson, Alexander Schrijver, and Paul Seymour.
\newblock Finding disjoint trees in planar graphs in linear time.
\newblock volume 147, pages 295--302, 01 1991.
\newblock \href {http://dx.doi.org/10.1090/conm/147/01180}
  {\path{doi:10.1090/conm/147/01180}}.

\bibitem{RS93}
N.~Robertson and P.D. Seymour.
\newblock Excluding a graph with one crossing.
\newblock {\em Graph struc- ture theory (Seattle, WA, 1991)}, 1993.
\newblock URL: \url{https://www.ams.org/books/conm/147/}, \href
  {http://dx.doi.org/http://dx.doi.org/10.1090/conm/147}
  {\path{doi:http://dx.doi.org/10.1090/conm/147}}.

\bibitem{RS95}
N.~Robertson and P.D. Seymour.
\newblock Graph minors .xiii. the disjoint paths problem.
\newblock {\em Journal of Combinatorial Theory, Series B}, 63(1):65--110, 1995.
\newblock URL:
  \url{https://www.sciencedirect.com/science/article/pii/S0095895685710064},
  \href {http://dx.doi.org/https://doi.org/10.1006/jctb.1995.1006}
  {\path{doi:https://doi.org/10.1006/jctb.1995.1006}}.

\bibitem{RS03}
Neil Robertson and P.D Seymour.
\newblock Graph minors. xvi. excluding a non-planar graph.
\newblock {\em Journal of Combinatorial Theory, Series B}, 89(1):43--76, 2003.
\newblock URL:
  \url{https://www.sciencedirect.com/science/article/pii/S009589560300042X},
  \href {http://dx.doi.org/https://doi.org/10.1016/S0095-8956(03)00042-X}
  {\path{doi:https://doi.org/10.1016/S0095-8956(03)00042-X}}.

\bibitem{RS04}
Neil Robertson and P.D. Seymour.
\newblock Graph minors. xx. wagner's conjecture.
\newblock {\em Journal of Combinatorial Theory, Series B}, 92(2):325--357,
  2004.
\newblock Special Issue Dedicated to Professor W.T. Tutte.
\newblock URL:
  \url{https://www.sciencedirect.com/science/article/pii/S0095895604000784},
  \href {http://dx.doi.org/https://doi.org/10.1016/j.jctb.2004.08.001}
  {\path{doi:https://doi.org/10.1016/j.jctb.2004.08.001}}.

\bibitem{Schrijver94}
Alexander Schrijver.
\newblock Finding k disjoint paths in a directed planar graph.
\newblock {\em {SIAM} J. Comput.}, 23(4):780--788, 1994.
\newblock URL: \url{https://doi.org/10.1137/S0097539792224061}, \href
  {http://dx.doi.org/10.1137/S0097539792224061}
  {\path{doi:10.1137/S0097539792224061}}.

\bibitem{STF14}
Simon Straub, Thomas Thierauf, and Fabian Wagner.
\newblock Counting the number of perfect matchings in k5-free graphs.
\newblock In {\em 2014 IEEE 29th Conference on Computational Complexity (CCC)},
  pages 66--77, 2014.
\newblock \href {http://dx.doi.org/10.1109/CCC.2014.15}
  {\path{doi:10.1109/CCC.2014.15}}.

\bibitem{TW09}
Thomas Thierauf and Fabian Wagner.
\newblock Reachability in k3,3-free graphs and k5-free graphs is in unambiguous
  log-space.
\newblock In Miros{\l}aw Kuty{\l}owski, Witold Charatonik, and Maciej
  G{\k{e}}bala, editors, {\em Fundamentals of Computation Theory}, pages
  323--334, Berlin, Heidelberg, 2009. Springer Berlin Heidelberg.

\bibitem{Vazirani89}
Vijay~V. Vazirani.
\newblock Nc algorithms for computing the number of perfect matchings in
  k3,3-free graphs and related problems.
\newblock {\em Information and Computation}, 80(2):152--164, 1989.
\newblock URL:
  \url{https://www.sciencedirect.com/science/article/pii/0890540189900175},
  \href {http://dx.doi.org/https://doi.org/10.1016/0890-5401(89)90017-5}
  {\path{doi:https://doi.org/10.1016/0890-5401(89)90017-5}}.

\bibitem{Wagner1937}
Klaus~Von Wagner.
\newblock {\"U}ber eine eigenschaft der ebenen komplexe.
\newblock {\em Mathematische Annalen}, 114:570--590, 1937.
\newblock URL: \url{https://api.semanticscholar.org/CorpusID:123534907}.

\bibitem{YU97}
Raphael Yuster and Uri Zwick.
\newblock Finding even cycles even faster.
\newblock {\em SIAM Journal on Discrete Mathematics}, 10(2):209--222, 1997.

\end{thebibliography}
\appendix

\section{Appendix 1}\label{appendix1}
\subsection{Proof of lemma~\ref{lem:3dpp}, planar case}\label{app:3dpp_planar}
We restate the lemma for the disjoint path problem with parity
in a special case of planar graphs.
\threeDpp*
\begin{proof}{\emph{(Of part 2)}}
 We first reiterate some of the machinery developed
 in \cite{Nedev99}.
 We give the definitions of \emph{list superface} and \emph{simple list superface} in planar graphs as defined in \cite{Nedev99}.
 \begin{definition}
 Let $G$ be a planar graph and $G'$ be its planar embedding. Let $P$ and 
 $P'$ be two paths in $G$, such 
 that (1) both $P$ and $P'$ start from the same vertex, say $b$, and end at the same vertex, 
 say $e$, and (2) $P$ and $P'$ are vertex disjoint except at $b$ and $e$. Then, we call $P$ and $P'$ 
 together with their interior region (the finite region enclosed by $P$ and $P'$) or exterior 
 region (the region on the plane apart from the interior), a list superface.
 \end{definition}
 Note that $P$ and $P'$ as described above can form two list superfaces, one with the interior 
 region on the plane and one with the exterior region on the plane with respect to the boundary formed by them.
 \begin{definition}
 A list superface $F$ is called a simple list superface if for every directed path 
 $Q = (q_1, q_2, \dots, q_k)$, where $q_1$ and $q_k$ lie on the boundary of $F$ and 
 $q_2, q_3, \dots, q_{k - 1}$ lie on the region of $F$, there exists a directed path from 
 $q_k$ to $q_1$ using only a subset of the edges of the boundary of $F$.
 \end{definition}
 \begin{figure}[H]
  \centering
  \begin{minipage}[b]{0.3\textwidth}
   \begin{tikzpicture}[scale=.5]
   \begin{scope}[thick,decoration={markings, mark=at position .5 with {\arrow{>}}}] 
   \draw[postaction={decorate}][draw, thick] (0,0) to[out=45,in=130, distance=4cm] (8,0);
   \draw[postaction={decorate}][draw, thick] (0,0) to[out=330,in=230, distance=4cm] (8,0);
   \draw[postaction={decorate}][draw, dotted] (2,1.6) to[out=330,in=140, distance=2cm] (5.5,-1.82);
   \draw[postaction={decorate}][draw, dotted] (1,-.55) to[out=60,in=130, distance=2cm] (4.8,-1.9);
   \end{scope}
   
   \draw[fill] (0,0) circle [radius=0.1];
   \draw[fill] (8,0) circle [radius=0.1];
   \draw[fill] (2,1.6) circle [radius=0.07];
   \draw[fill] (5.5,-1.82) circle [radius=0.07];
   \draw[fill] (4.8,-1.9) circle [radius=0.07];
   \draw[fill] (1,-.55) circle [radius=0.07];
   
   \node at (-.4,0) {$b$};
   \node at (8.4,0) {$e$};
   \node at (1.6,1.9) {$u_1$};
   \node at (5.9, -2.1) {$v_1$};
   \node at (4.8,-2.25) {$v_2$};
   \node at (.8,-.85) {$u_2$};
   \end{tikzpicture}
   \caption{A list superface}
   \label{nonsimpleLS}
  \end{minipage}
  \hspace{1in}
  \begin{minipage}[b]{0.3\textwidth}
   \begin{tikzpicture}[scale=.5]
   \label{simplels}
   \begin{scope}[thick,decoration={markings, mark=at position 0.5 with {\arrow{>}}}] 
   \draw[postaction={decorate}][draw, thick] (0,0) to[out=45,in=130, distance=4cm] (8,0);
   \draw[postaction={decorate}][draw, thick] (0,0) to[out=330,in=230, distance=4cm] (8,0);
   \draw[postaction={decorate}][draw] (7,1) to[out=230,in=310, distance=2cm] (3,2);
   \draw[postaction={decorate}][draw] (4.8,-1.9) to[out=130,in=60, distance=2cm] (1,-.55);
   \end{scope}
   
   \draw[fill] (0,0) circle [radius=0.1];
   \draw[fill] (8,0) circle [radius=0.1];
   \draw[fill] (7,1) circle [radius=0.07];
   \draw[fill] (3,2) circle [radius=0.07];
   \draw[fill] (4.8,-1.9) circle [radius=0.07];
   \draw[fill] (1,-.55) circle [radius=0.07];
   
   \node at (-.4,0) {$b$};
   \node at (8.4,0) {$e$};
   \node at (7.35,1.25) {$u_2$};
   \node at (3,2.45) {$v_2$};
   \node at (4.8,-2.25) {$u_1$};
   \node at (.8,-.85) {$v_1$};
   \end{tikzpicture}
   \caption{A simple list superface}
  \end{minipage}
  \caption{List superface in \textsf{\textbf{(a)}} is non-simple because of paths from $u_2$ to $v_2$ and $u_1$ to $v_1$. On the other     hand, \textsf{\textbf{(b)}} is a simple list superface.}
 \end{figure}
 
 As noted earlier, if we can solve the above
 problem for total parity even, we can also find if such paths of total
 parity odd exist, by adding a dummy out neighbour $v'_4$ of $v_4$ and
 finding disjoint paths of total parity even from $v_1$to $v_2$,
 $v_2$ to $v_3$ and from $v_3$ to $v'_4$. Hence we assume that we want
 to find disjoint paths with total parity even.
 
 We first find a set of pairwise disjoint paths from $v_1$ to $v_2$,
 $v_2$ to $v_3$ and from $v_3$ to $v_4$(if it exists) using \cite{Schrijver94}, but the 
 total parity might be odd (If no such paths exist, we output negative, and if we get 
 an instance of paths of total parity even then we output that instance). 
 Suppose a solution exists, and let $P_1,P_2,P_3$, 
 be one set of pairwise disjoint paths from $v_1$to $v_2$,
 $v_2$ to $v_3$ and from $v_3$ to $v_4$ respectively, of total parity even. 
 Let $P'_1,P'_2,P'_3$ be another set of pairwise disjoint paths of
 total parity odd between same corresponding vertices.
  Let $i,j,k$ refer to distinct integers from the set $\{1,2,3\}$ hereafter. 
 At least one of the pairs $P_i,P'_i$ must have
 different parities, otherwise the total parities of both sets would be
 the same. Suppose $P_1,P'_1$ are of different parities. 
 
 Using this we are going to extend the ideas of \cite{Nedev99} and show
 the following lemma:
 \begin{lemma}\label{lem:3dpp_olsf}
  Let $G$ be a plane graph and $v_1,v_2,v_3,v_4$ vertices such that $v_1,v_2,v_3$ 
  lie on the same face. Consider two instances of disjoint paths $(P_1,P_2,P_3)$ 
  and $(P'_1,P'_2,P'_3)$ between $v_1,v_2,v_3,v_4$ as described above, of different
  total parities, and let $P_1,P'_1$ have different parities. Such a pair of instances exists iff 
  a structure consisting of the following exists : 
  \begin{itemize}
   \item A simple odd list superface $F$ with starting and ending 
     vertices $b,e$ and paths $Q_{F},Q'_{F}$ (of different parities) as boundaries of $F$.
   \item Paths $Q_b$ from $v_1$ to $b$, $Q_e$ from $e$ to $v_2$, $Q_2$ from 
     $v_2$ to $v_3$, and $Q_3$ from $v_3$ 
     to $v_4$. 
   \item  $Q_F,Q'_F, Q_b, Q_e, Q_2, Q_3$ are all pairwise disjoint (except at vertices $b,e$). 
  \end{itemize}
  Similar claim with appropriate changes also holds if paths $P_2,P'_2$ are of different parities
  or if paths $P_3,P'_3$ are of different parities.
 \end{lemma}
 The odd list superface will act as a switch for parities.
 This gives a recipe, similar to \cite{Nedev99} to find a solution in polynomial time. 
 We can enumerate over all
 simple odd list superfaces as done in \cite{Nedev99} and use the algorithm
 of \cite{Schrijver94} to find the four pairwise disjoint paths described above
 in $G$, after removing boundary vertices of the simple odd list superface. If such
 an odd list superface and the corresponding paths exist, it gives a solution, 
 else we output that no solution exists. 
 One direction of the claim is straightforward: If there exists a simple odd list super
 face $F$ and paths $Q_b,Q_e,Q_2,Q_3$ as described above, then 
 the set of paths $(Q_e.(\text{one of }Q_{F},Q'_{F}).Q_b,Q_2,Q_3)$ form
 two instances of disjoint paths of different parities.
 In order to prove the other direction of the claim, we note some definitions
 and observations.
 
 Let $G$ be a plane graph and $v_1,v_2$ be two vertices on the outer face, and 
 $P_1,P'_1$ be two paths from $v_1$ to $v_2$. 
 By Jordan curve theorem, each path 
 $P_1,P_2$ divides $G$ into two regions. 
 If $P$ is a directed path starting and
 ending on a closed boundary, we call the region on the left of $P$ as 
 $\leftR{P}$ and the region along the right as $\rightR{P}$. 
 
 Let $P_1\cup P'_1$  denote the subgraph of $G$ formed by vertices and edges of paths $P_1,P'_1$.
 The leftmost undirected walk or the of subgraph $P_1\cup P'_1$, denote by $\leftm{P_1}$,
 is defined as the undirected walk obtained by traversing the clockwise 
 first edge(ignoring direction) 
 of $P_1\cup P'_1$ with respect to the parent edge at every step, until it reaches $v_2$ or starts
 repeating. 
 For the first step, we can use an imaginary dart from $v_1$ into the 
 outer face 
 as parent edge for reference. 
 The rightmost walk of $P_1\cup P'_1$, denoted by $\rightm{P_1}$ is defined similarly 
 by taking the counter-clockwise first edge at every step.
 We show some lemmas using these paths.
 \begin{claim}\label{claim:leftm_paths}
   Suppose $v_1,v_2$ lie on outer face of $G$ and $P_1,P'_1$ are paths
   from $v_1$ to $v_2$.
  \begin{enumerate} 
    \item The undirected walks $\leftm{P_1},\rightm{P_1}$ are simple, directed 
      paths from $v_1$ to $v_2$.
   \item Let $P_2$ be a path disjoint from $P_1$ and lying (strictly) in 
    $\rightR{P_1}$. Let $P'_2$ be a path disjoint from $P'_1$, lying (strictly) in $\rightR{P'_1}$. 
    Then the path $\leftm{P_1}$ is disjoint from both $P_2,P'_2$, and 
    lies in $\leftR{P_1}\cap \leftR{P'_1}$.
   \end{enumerate}
 \end{claim}
 \begin{proof}
  \begin{enumerate}
   \item We first show that $\leftm{P_1}$ is a simple path ignoring directions.
    \\ Let $\leftm{P_1}= \langle v_1,\ldots y,x,z\ldots w,x,\ldots \rangle$, where $x$ is the first 
    vertex on $\leftm{P_1}$ that repeats in the walk. Let $C$ denote the subcycle 
    $\langle x,z\ldots w,x\rangle$ of $\leftm{P_1}$. Suppose the edge $(y,x)$ belongs to
    $P_1$ (See~\cref{fig:leftm}). 
    The path $P_1[x,v_2]$ cannot touch $C$ at any point other than $x$, since $v_2$ lies
    on the exterior of $C$ and any edge touching $C$ from exterior at any point other than $x$
    would condradict the assumption that $C$ was obtained in the leftmost walk. But if it touches
    $C$ at $x$ only, then edges $(x,z),(w,x)$ must both belong to $P'_1$ and by similar argument
    they would have to leave $C$ contradicting the assumption that $C$ occured in a leftmost walk.
    Therefore $\leftm{P_1}$ is a simple undirected path.
    Next we show that $\leftm{P_1}$ is actually a directed path from
    $v_1$ to $v_2$. Suppose that is not the case, 
    which means it consists of segments that are alternately directed
    towards or away from $v_2$. Let $x,y$ be the first maximal segment directed away from $v_2$(i.e. 
    $\leftm{P_1} = \langle v_1,\ldots y,\ldots x_1,x,x_2\ldots v_2\rangle$, where all arcs in the 
    segment from $v_1$ till $y$ are directed towards $v_2$, all arcs in segment from  $y$ till $x$
    are directed away from $v_2$, and the arc $(x,x_2)$ is directed towards $v_2$). 
    Without loss of generality,
    let the arc $(x,x_1)$ be a part of $P_1$, which implies the 
    incoming arc to $x$ in $P_1$ is not on
    $\leftm{P_1}$ and lies in $\rightR{\leftm{P_1}}$. Consider the closed loop formed by segments 
    $P_1[v_1,x],\leftm{P_1}[x,y],\leftm{P_1}[v_1,y]$. 
    Since $P_1$ goes to $x_1$ after $x$, and $v_2$ lies
    on the exterior of the loop, $P_1$ must exit the loop after $x_1$ to reach $v_2$. 
    But cannot exit via that segment since that would contradict that 
    $\leftm{P_1}[v_1,y]$ is a subpath of the lefmost
    undirected path $\leftm{P_1}$. 
    Therefore every arc in $\leftm{P_1}$ must be directed from $v_1$ to $v_2$.
   \item This follows easily from definition of $\leftm{P_1}$.
  \end{enumerate}
 \end{proof}
 \begin{figure} 
 \begin{minipage}{\textwidth}
  \begin{minipage}{0.5\textwidth}
   \includegraphics[scale=0.65]{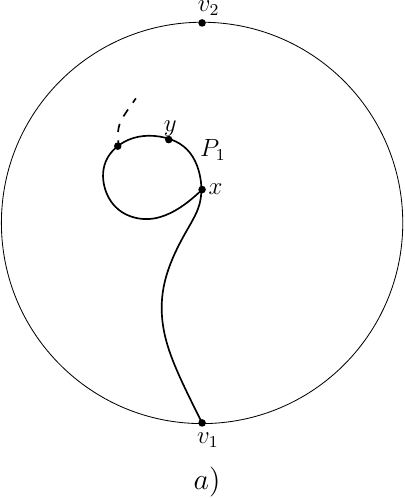}
  \end{minipage}
  \begin{minipage}{0.3\textwidth}
   \includegraphics[scale=0.65]{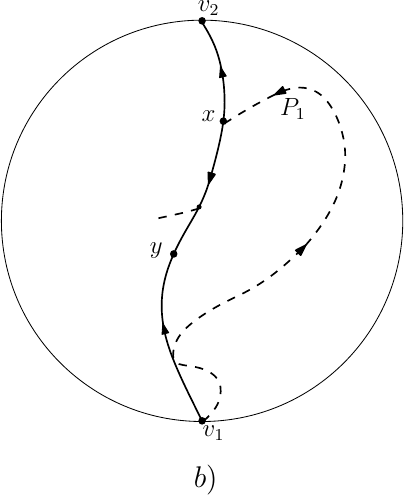}
  \end{minipage}
   \captionof{figure}{Suppose the undirected leftmost walk from $v_1$ in 
    $P_1\cup P'_1$ ends up in the cycle shown in $a)$. 
    Suppose the edge $(x,y)$ (undirected) belongs to $P_1$. Then one of the (undirected)
    segments $P_1[v_1,x],P_1[x,v_2]$ must have an edge touching the 
    cycle from its exterior as shown by the dashed path. 
    This would contradict the cycle being part of the leftmost walk. 
    Therefore the leftmost walk $\leftm{P_1}$ must be a path.
  In figure $b$), suppose the leftmost path $\leftm{P_1}$, drawn 
  in bold, has a segment from $x$ to $y$ directed away from $v_2$. 
    Suppose the outgoing edge of $x$ in this segment belongs to $P_1$. 
    Then $P_1[v_1,x]$ is a subpath of $P_1$ lying in $\rightR{\leftm{P_1}}$, 
    forming a closed loop with the undirected segment $\leftm{P_1}[v_1,x]$. 
    Therefore the only way for $P_1$ to reach $v_2$ is to
    exit the segment $\leftm{P_1}[x,y]$ into $\leftR{\leftm{P_1}}$, 
  contradicting that $\leftm{P_1}$ is the leftmost path.}\label{fig:leftm}
 \end{minipage}
 \end{figure}
 Now we show a tweaked version of Lemma $2$, and restate Lemma $1$ of \cite{Nedev99}.
 \begin{lemma}\label{lem:Nedev_ext}
  \begin{enumerate}
   \item \label{lem:Nedev_lem1_ext} (Extension of Lemma $2$ of~\cite{Nedev99}) Let $P_1,P'_1$ be
    paths as defined above, of different parity. Let $\leftm{P_1}$ 
    and $P_1$ be of different parities. Then there exists an odd list superface $F$
    in $\leftm{P_1}\cup P_1$ with begining and ending vetices $b,e$ respectively
    with boundaries $\leftm{P_1}[b,e], P_1[b,e]$, and paths $Q_b,Q_e$ from $v_1$ to
    $b$ and $e$ to $v_2$ respectively such that $Q_b,Q_e$ and $F$ are pairwise
    disjoint(except at some end points). Moreover $Q_b=\leftm{P_1}[v_1,b], 
    Q_e=\leftm{P_1}[e,v_2]$.
   \item \label{lem:Nedev_lem2}(Reiteration of Lemma $1$ of~\cite{Nedev99})If the odd 
    list superface $F$ found above in $G$ is not simple, we can find a simple odd list 
    superface $F_s$ which is contained inside the region of $F$, and extend $Q_b$ to $Q'_b$, 
    from $v_1$ to beginning of $F_s$, and $Q_e$ to $Q'_e$, from ending of $F_s$ to $v_2$, 
    such that the extensions are also contained inside $F$ and $F_s,Q'_b,Q'_e$ are all pairwise 
    disjoint.
   \end{enumerate}
 \end{lemma}
\begin{proof}
 \begin{enumerate}
  \item Let $x_1$ denote the first vertex after $v_1$ where 
    $\leftm{P_1}$ and say, $P_1$ diverge. 
   Let $y_1$ be the first vertex of $\leftm{P_1}$ after $x_1$ which also is in $P_1$.
   Let $y'_1$ be the first vertex of $P_1$ after $x_1$ which is also in $\leftm{P_1}$ and
   $y_1\neq y'_1$. Then $y'_1$ appears after $y_1$ in $\leftm{P_1}$. 
   The segments $\leftm{P_1}[x_1,y'_1] \cup (P_1[x_1,y'_1])$ forms a closed loop.
   Now the vertex $y_1$ lies on the segment $\leftm{P_1}[x_1,y'_1]$ of the loop 
   $\leftm{P_1}[x_1,y'_1]\cup P_1[x_1,y'_1]$. The simple path $P_1[y'_1,v_2]$ must continue from 
   $y_1$ to $v_2$ which is outside the loop, for which it must exit the loop via the segment
   $\leftm{P_1}[x_1,y'_1]$. But any edge going out of the loop from a vertex of 
   $\leftm{P_1}[x_1,y'_1)$ will contradict the assumption that $\leftm{P_1}$ is the leftmost
   walk. Therefore $y'_1=y_1$. Therefore we have a list superface in $\leftm{P_1}\cup P_1$ with
   beginning and ending vertices $x_1,y_1$ respectively and boundaries consisting of
   $\leftm{P_1}[x_1,y_1],P_1[x_1,y_1]$. Moreover, since $x_1$ was the first point of divergence
   and $y_1$ the first point after $x_1$ where both $\leftm{P_1},P_1$ meet each other, the
   segments $Q_b=\leftm{P_1}[v_1,x_1],Q_e=\leftm{P_1}[y_1,v_2]$ are mutually disjoint and
   disjoint from boundaries of the list superface. If both boundaries are of different parity,
   then we are done. 
   Else we keep the segment $\leftm{P_1}[v_1,y_1]$ as part of our initial segment, 
   and repeat the same argument starting from $y_1$ instead of $v_1$. At some point we must get
   a list superface satisfying all the conditions and with different boundaries since the parities
   of $\leftm{P_1},P_1$ are different. Therefore the lemma holds.
   \item The proof of this is the same as that of Lemma 2 of \cite{Nedev99}. 
 \end{enumerate}
\end{proof}
\begin{figure}
 \begin{minipage}{\textwidth}
  \begin{minipage}{0.5\textwidth}
   \includegraphics[scale=0.7]{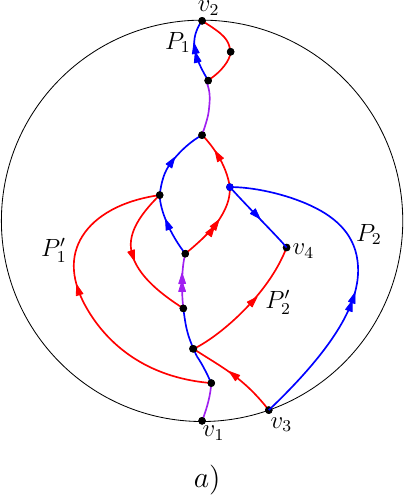}
  \end{minipage}
  \begin{minipage}{0.3\textwidth}
   \includegraphics[scale=0.7]{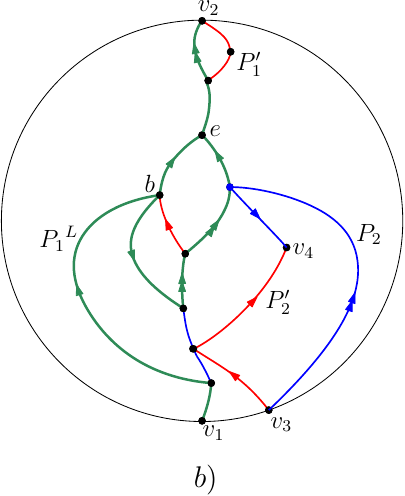}
  \end{minipage}
   \captionof{figure}{The red colored paths in Fig a) are $P'_1,P'_2$, and the blue colored
    paths are $P_1,P_2$(We use purple segments where they share edges). Single arrows denote odd
    length segments and double arrows even length segments. Segments without arrows can be 
    assumed to be of any parity. The paths highlighted in green in Fig b) denote the odd list
  superface and the segments $Q_b,Q_e,P'_2$ that we get in lemma~\ref{lem:Nedev_lem1_ext}.}
 \end{minipage}
\end{figure}
 Now we prove the other direction of~\cref{lem:3dpp_olsf}
 \begin{proof}
  We assume that $v_1,v_2,v_3$ lie on the outer face, and hence
  can consider them on an imaginary boundary that encloses the region $R$ where
  the graph is embedded.(see~\cref{fig:nedev_cases}). 
  We consider the cases of which of the pairs $P_i,P'_i$, $i\in \{1,2,3\}$, 
  consists of paths of different parity. In each case, we assume without loss of generality
  that $\leftm{P_i}$ is of same parity as $P'_i$, and hence of parity different
  from that of $P_i$
  \begin{enumerate}
   \item Case 1. $P_1,P'_1$ are paths of different parities. The vertices
    $v_3,v_4$ must lie on the same side of both $P_1,P'_1$, since they
    are connected by a path disjoint from $P_1$ and a path disjoint from $P'_1$. 
    W.l.o.g., we can assume that they lie strictly on the right side of $P_1,P'_1$. 
    By our assumption, $\leftm{P_1},P_1$ are of different parities.
    By claim~\ref{claim:leftm_paths}, the sets $(P_1,P_2,P_3)$ and
    $(\leftm{P_1},P'_2,P'_3)$ are also instances of disjoint paths of different
    total parities. Let $F$ be the odd list superface formed by $P_1,\leftm{P_1}$,
    as described in~\cref{lem:Nedev_ext}. 
    The boundary $\leftm{P_1}[b,e]$ of $F$, as well as segments $Q_b=\leftm{P_1}[v_1,b],
    Q_e=\leftm{P_1}[e,v_2]$ are disjoint from $P_2,P'_2,P_3,P'_3$ by 
    claim~\ref{claim:leftm_paths}. Since other boundary of $F$ is $P_1[b,e]$, 
    both the boundaries of $F$, and segments $Q_b,Q_e$ are disjoint from $P_2,P_3$. 
    Therefore we have the
    odd list superface $F$, and segments $Q_b,Q_e$ that are
    all mutually disjoint, satisfying all requirements of lemma~\ref{lem:3dpp_olsf}
    except possibly for $F$ being simple. Now if $F$ is not a simple odd list superface, then
    by lemma~\ref{lem:Nedev_ext}, there is a simple odd list superface $F_s$ contained
    inside $F$ and $Q_b,Q_e$ can be extended to beginning and ending of 
    $F_s$ respectively. Since $F_s$ and the extensions of $Q_b,Q_e$ lie
    inside $F$, they are also disjoint from $P_2,P_3$ and therefore all requirements
    of lemma~\ref{lem:3dpp_olsf} are satisfied.
   \item Case 2. $P_2,P'_2$ are paths of different parities. We have two subcases:
    \begin{enumerate}
     \item Both $v_1,v_4$ lie on the same side of $P_2,P'_2$, say right side.
      In this case the same argument as above works.
     \item Vertex $v_1$ lies on, say, the left side of $P_2,P'_2$, and $v_4$ on their
      right side. Since $\leftm{P_2},P_2$ are of different parities,
      let $F$ be the odd list superface formed by $\leftm{P_2},P_2$, with 
      $Q_b=\leftm{P_2}[v_2,b],Q_e=\leftm{P_2}[e,v_3]$ as described in the lemma
      above. 
      Then consider the paths $\leftm{P_1}, \leftm{P_2}[v_2,b], \leftm{P_2}[e,v_3],
      P_3$ and $F$. By claim~\ref{claim:leftm_paths}, $\leftm{P_1}$ is disjoint
      from both $\leftm{P_2},P_2,P_3$ and $\leftm{P_2}$ is disjoint from $P_3$. Now we
      can proceed as in above argument and satisfy conditions of the lemma.
    \end{enumerate}
  \item Case 3. $P_3,P'_3$ are paths of different parities. Since $v_4$ need not lie on 
    the outer face, $\leftm{P_3}$ or $\rightm{P_3}$ need not have the properties we
    showed in above lemma. Let $F$
    be an odd list superface formed by $P_3,P'_3$, with $b,e$ as beginning, 
    ending vertices, and $Q_b,Q_e$ paths from $v_3$ to $b$ and $e$ to $v_4$ as
    described by lemma $2$ of \cite{Nedev99} (Note that each of $Q_b,Q_e$ here
    can have edged from both $P_3,P'_3$). There are two subcases (others are 
    symmetrical).
    \begin{enumerate}
     \item Vertices $v_3,v_4$ lie on right side of $P_1,P'_1$ and also
      on the right side of $P_2,P'_2$. In this case consider the paths
      $\leftm{P_1}, \leftm{P_2}$. They are both mutually disjoint as well as
      disjoint from $F, Q_b, Q_e$ by claim~\ref{claim:leftm_paths}. Then we can 
      proceed as above.
     \item Vertices $v_3,v_4$ lie on right side of $P_1,P'_1$ and on left side 
      of $P_2,P'_2$. Then consider the paths $\leftm{P_1}, \rightm{P_2}$. They
      are both mutually disjoint and also disjoint from $F,Q_b,Q_e$ by 
      claim~\ref{claim:leftm_paths}. Now we can again proceed as above.
    \end{enumerate}
   \end{enumerate}
  Hence the claim holds.
 \end{proof}
\begin{minipage}{\textwidth}
 \begin{minipage}{\textwidth}
  \begin{minipage}{0.4\textwidth}
     \includegraphics[width=\textwidth]{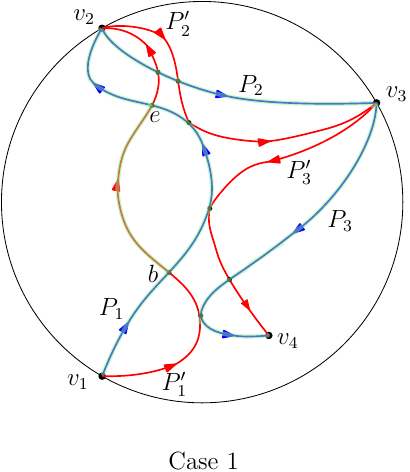}
      \label{fig:first}
  \end{minipage}
  \hfill
  \begin{minipage}{0.4\textwidth}    
     \includegraphics[width=\textwidth]{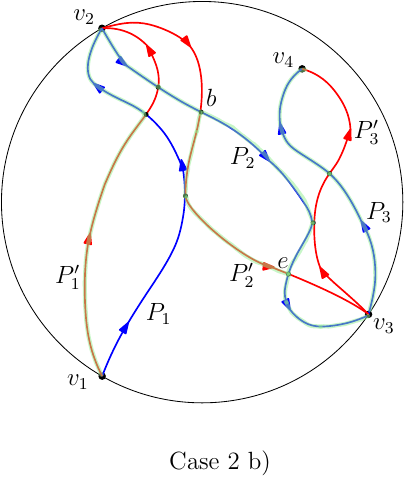}
      \label{fig:second}
  \end{minipage}
  \hfill
  \end{minipage}
  \begin{minipage}{0.4\textwidth}    
     \includegraphics[width=\textwidth]{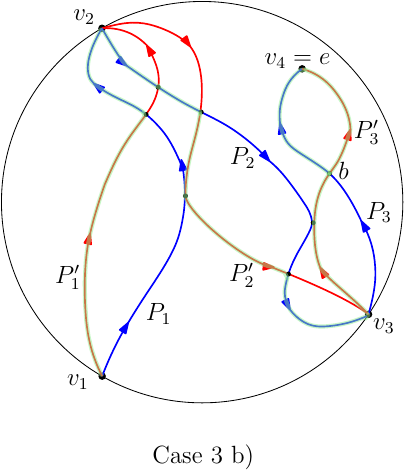}
      \label{fig:third}
  \end{minipage}
 \end{minipage}
 \captionof{figure}{Some cases of~\cref{lem:3dpp_olsf}. The 
   red colored paths are $P'_1,P'_2$, and the blue colored
    paths are $P_1,P_2$. The paths highlighted in green denote the odd list
  superface and the segments $Q_b,Q_e,P'_2$ that we get in lemma~\ref{lem:3dpp_olsf}.}
 \label{fig:nedev_cases}
\end{proof}
\subsection{Proof of lemma~\ref{lem:3dpp}, bounded treewidth case}\label{app:3dpp_btw}
We restate the lemma here and give the proof.
\threeDpp*
\begin{proof}{\emph{(Of Part 1)}}\\
Courcelle's theorem (\cite{COURCELLE}) states that every graph property definable in the monadic 
second-order logic of graphs can be decided in linear time on graphs of bounded treewidth.
The vocabulary of our logic consists of a unary relation symbol $V$ for the set of vertices,
a binary relation symbol $E$ for the set of edges,
and some constant symbols, like $s,t$ used to denote source, destination vertices.
The universe of discourse consists of vertices.
We use the fragment of monadic second order logic called $\sf{MSO_2}$ which allows quantification
over sets of vertices and edges but not more complex relations.
We use variables $A,B,P,P_1,P_2$ to denote sets of vertices, and variable $M$ to denote
set of edges.
For brevity, we use some shorthand notation for phrases like $\exists !x$ for
`there exists a unique $x$', $x \in P$ for `a vertex $x$ which is in vertex set $P$',
$x \in M$, for `a vertex $x$ which is in the set of vertices consisting of end points
of the edges in set $M$.' It is well known that these phrases can be easily expressed in
$\sf{MSO_2}$. We also use $V(M)$ to denote the set of vertices consisting of all end
points of the edges in $M$.
Consider the formula, $\phidisj{P_1}{P_2}$ which is true iff the sets $P_1,P_2$
are disjoint. 
We can write it in MSO as:
\begin{align*}
 \phidisj{P_1}{P_2} : \forall v, v\in P_1 \Rightarrow v \notin P_2
\end{align*}
We define a formula $\phipart{P}{A}{B}$, which is true iff the vertex sets $A,B$ form
a partition of the vertex set $P$, as:
\begin{align*}
 \phipart{P}{A}{B} : (\forall v(v\in P\Leftrightarrow (v\in A\vee v\in B))) \wedge (\phidisj{A}{B})
\end{align*}

The definition $\phipart{P}{A}{B}$ can be tweaked to get $\phipart{M}{A}{B}$ which is true
iff vertex sets $A,B$ for a partition of the vertex consisting of endpoints of edge set $M$.

Now we define a formula $\phieven{P}$ which is true iff the graph induced on vertex set $P$ has
a path from $s$ to $t$ of total length even. 
It is shown in\cite{THstack} that this can be expressed in $\sf{MSO_2}$. 
The idea is to express that there exists a set of edges $M$ in the graph induced by
$P$, whose vertices can be partitioned into vertex sets $(A,B)$, 
such that: 
\begin{itemize}
  \item Every edge of $M$ has one end point in $A$, and other in $B$. Both $s,t$ lie in $A$.  
  \item $s$ has only one neighbour in $M$, an out-neighbour, and $t$ has only one
        neighbour in $M$, an in-neighbour.
  \item Every other vertex of $M$ has exactly one out-neighbour and one in-neighbour.
\end{itemize}
It is clear that the graph induced by 
set of edges $M$ in $P$ must be a collection of disjoint cycles and an $s$-$t$ path. Since the graph
induced by $M$ is also bipartite, with $s,t$ lying in the same partition, it follows that the disjoint
cycles, as well as the $s$-$t$ path must be of even length.

Thus we can write the formula
for $\phieven{P}$ as:
\begin{equation*}
\begin{split}
  \exists &M ( (V(M)\subseteq P), \exists A,B ( \phipart{M}{A}{B} \wedge \\ 
     & \hspace{3mm} ( ((u,v)\in M)\Rightarrow ((u\in A \wedge v\in B)
   \vee(u\in B\wedge v\in A))) \wedge (s,t\in A)) \hspace{2mm} \wedge \\
  & ((\exists v (s,v)\in M)\wedge (\nexists v (v,s)\in M) 
   \wedge(\exists v ((v,t)\in M)) \wedge (\nexists v ((t,v)\in M))  ) \hspace{2mm} \wedge \\
 & (\forall u\in M, (\exists !u \in M, (u,v)\in M) \wedge (\exists !v \in M, (v,u)\in M)))
\end{split}
\end{equation*}
We can similarly define $\phiodd{P}$ which is true iff graph induced on $P$ 
has an $s$-$t$ path of odd length.
With $\phieven{},\phiodd{}$ and $\phidisj{}{}$, we can easily define a 
formula $\phidppeven$, which is true
iff the graph has two disjoint paths from $s_1$ to $t_1$, and from $s_2$ to $t_2$, with total parity
even as follows (we skip argument for brevity):
\begin{align*}
 \phidppeven  : \exists P_1,P_2 (\phidisj{P_1}{P_2} 
   \wedge ((\phieven{P_1}\wedge \phieven{P_2})& \;\vee \\
 ( & \phiodd{P_1}\wedge \phiodd{P_2})))
\end{align*}
We can define a similar formula as above for the case when total parity must be odd, as well as extend
it for three disjoint paths with total parity constraints.
\end{proof}
\subsection{Reducing even path and disjoint path with parity search to decision}\label{app:search_to_dec}
We now give a reduction from search version of even path problem to decision
version. 
Our algorithm will output the $s$-$t$ even path which is lexicographically 
least according to vertex indices.
A path $P_1$ is lexicographically lesser then $P_2$ if the among the first vertices
in $P_1,P_2$ where they diverge (assuming they start at the same vertex), the one in 
$P_1$ has a lesser index than the one in $P_2$.

We give an (unoptimized) algorithm below.

\begin{algorithm}[H]
\caption{Routine to find a path of parity $p$ between given vertices}
\label{parity_path_search}
\begin{algorithmic}[1]
\Procedure{\textsc{findParityPath}}{\textrm{$G,s,t,p$}}
 \If{s==t}
  \State Return <>;
 \EndIf\\ 
 \For{$v_{out}=\textsc{Each out neighbour of s in increasing index order}$}
   \If{$\textsc{searchParityPath}(G-\{s\},v_{out},t,1-p) == True$}
    \State Return $\langle s,v_{out} \rangle + \textsc{findParityPath}(G-\{s\},v_{out},t,1-p)$
   \EndIf
 \EndFor\\
 \State Return NULL\\
  
\EndProcedure
\end{algorithmic}
\end{algorithm}
The correctness of the algorithm can be proved by inducting on length of the unique
lexicographically least path of parity from $s$ to $t$.
The algorithm can be optimized to makes $\calO(n)$ queries to $\textsc{searchParityPath}$ routine,
and so the time complexity is $\calO(n.T(n))$ where $T(n)$ is the complexity of
decision version of parity path in $G$.

The same algorithm can be extended to find disjoint paths between $(s_i,t_i)$ with total
parity even with an oracle to decision version of the same problem. In that case, we will
output the lexicographically least instance of disjoint paths with total parity even, which
consists of solution obtained by  considering solutions with lex-least $(s_1,t_1)$ path, 
and then amongst those, the solutions with lex-least $(s_2,t_2)$ path and so on. Hence we can
also find disjoint paths with parity in time $\calO(n.T(n))$ where $T(n)$ is the complexity 
of decision version of disjoint paths with parity in $G$.
\subsection{Proof of lemma~\ref{lem:par_mim_net}}\label{app:par_mim_net_proof}
We restate the lemma here:
\abc*
\begin{proof}
We will now formalise the proof idea using some definitions and lemmas.
Let $\parConf$ denote the parity configuration of $\lnode$
with respect to terminals $\{v_1,v_2,v_3\}$. We will throughout assume that variables $i,j,k$ 
take distinct values in $\{1,2,3\}$. We will generally use $x_{ij},x_{ij}+1$, where
$x_{ij} \in \{0,1\}$ and addition is modulo $2$, to denote entries of sets 
$\dir{\parConf}{v_i}{v_j}, \via{\parConf}{v_i}{v_k}{v_j}$. We will refer to the 
entries in the \emph{Dir} sets as direct set entries. 
Given two parity configurations $\parConf, \parConf'$, we say $\parConf'$ is a
\emph{sub-configuration} of $\parConf$, if all \emph{Dir,Via} sets of $\parConf'$ are
subsets of the corresponding \emph{Dir,Via} sets of $\parConf$, i.e. 
$\dir{\parConf'}{v_i}{v_j} \subseteq \dir{\parConf}{v_i}{v_j}$, and
$\via{\parConf'}{v_i}{v_k}{v_j} \subseteq \via{\parConf}{v_i}{v_k}{v_j}, \forall i,j,k$.
We denote this as $\parConf' \subseteq_c \parConf$.

We now define \emph{bad pairs} of $\parConf$.

\begin{definition}
  Given the parity configuration $\parConf$, a pair of entries,
  $x_{ij} \in \dir{\parConf}{v_i}{v_j}$ and $x_{jk} \in \dir{\parConf}{v_j}{v_k}$, 
  are called a \emph{bad pair} if $x_{ij}+x_{jk} \notin \left(\dir{\parConf}{v_i}{v_k} 
  \bigcup \via{\parConf}{v_i}{v_j}{v_k}\right )$. A direct set entry $x_{ij}$ is called a 
  \emph{bad} entry
  if it is part of at least one bad pair in $\parConf$. We call direct set entries that
  are not part of \emph{any} bad pair as \emph{good} entries.
  We use the phrase \emph{bad pairs between} $\dir{\parConf}{v_i}{v_j}, 
  \dir{\parConf}{v_j}{v_k}$, to refer to the bad pairs formed by 
  entries of $\dir{\parConf}{v_i}{v_j},\dir{\parConf}{v_j}{v_k}$
\end{definition}

Next we define the \emph{bad kernel}, $\bk{\parConf}$  of $\parConf$, which is the 
sub-configuration of $\parConf$ consisting of all the bad entries of 
$\parConf$. Its closure, $\bkc{\parConf}$ adds a minimal number of entries to 
make it realisable. 
\begin{definition}
  Given a parity configuration, $\parConf$, its \emph{bad kernel}, denoted by $\bk{\parConf}$, 
 is the sub-configuration defined as:
 \begin{itemize}
   \item For every $i,j$ and $x_{ij} \in \{0,1\}$,  $x_{ij} \in \dir{\parConfBad}{v_i}{v_j}$ iff 
    $x_{ij}$ is a bad entry in $\dir{\parConf}{v_i}{v_j}$.
\end{itemize} 
  The \emph{closure} of $\parConfBad$, denoted by $\bkc{\parConf}$, is obtained from
  $\parConfBad$ by augmenting it as follows:
  \begin{enumerate}
   \item For every bad pair $(x_{ij},x_{jk})$ in $\parConfBad$, 
    if $x_{ij}+x_{jk}+1 \notin \dir{\parConfBad}{v_i}{v_k}$, then add $x_{ij}+x_{jk}+1$
    to $\dir{\parConfBad}{v_i}{v_k}$.
   \item For every good pair $(x_{ij},x_{jk})$ in $\parConfBad$, 
    if $x_{ij}+x_{jk} \notin \dir{\parConfBad}{v_i}{v_k}$, then add $x_{ij}+x_{jk}$
    to $\via{\parConfBad}{v_i}{v_j}{v_k}$.
  \end{enumerate}
\end{definition}

The following is a simple claim:
\begin{claim}
  Let $\parConf$ be a realizable parity configuration. Let $\bk{\parConf}$ be the bad kernel
  of $\parConf$, and $\bkc{\parConf}$, its closure. 
  Then $\bk{\parConf} \subc \bkc{\parConf} \subc \parConf$.
\end{claim}
\begin{proof}
  It is clear from definition of $\bkc{\parConf}$ that $\bk{\parConf} \subc 
  \bkc{\parConf}$.  
  Suppose $x_{ij},x_{jk}$ form a bad pair in $\parConf$, and hence are present in
  $\bk{\parConf}$. Since, $\parConf$ is realizable, there must a path from $v_i$ to
  $v_k$ in any graph with parity configuration $\parConf$. Since by definition of a bad pair,
  the parity of that path cannot be $x_{ij}+ x_{jk}$, it must be $x_{ij}+x_{jk}+1$.
  Therefore any entry added in step $1$ of construction of $\bkc{\parConf}$ must also
  be present in the corresponding set in $\parConf$.
  Similar argument also holds for entries added in step $2$ of construction of $\bkc{\parConf}$.
  Therefore the claim holds.
\end{proof}

We show that in order to find a parity mimicking network of any realizable parity configuration that
satisfies the required planarity conditions, it 
is sufficient to give a parity mimicking network obeying those conditions for 
the \emph{closure} of its \emph{bad kernel}. 

\begin{lemma}
 Let $\parConf$ be a parity configuration with respect to terminals 
 $\{v_1,v_2,v_3\}$ and let $\bkc{\parConf}$ denote the closure of the bad kernel
 of $\parConf$.  Suppose $L''$ is a planar parity mimicking network of the
 parity configuration $\bkc{\parConf}$, with terminals lying on outer face. We can 
 construct a planar parity mimicking network $L'$ for configuration $\parConf$, with terminals 
 lying on outer face, by adding edges to $L''$ using the following iterative operation:
 \begin{itemize}
  \item For every entry $x_{ij} \in \dir{\parConf}{v_i}{v_j}$, if there is not
  a direct path in $L''$ from $v_i$ to $v_j$ of parity $x_{ij}$ already, then
  add a path of parity $x_{ij}$ (length one or two) from $v_i$ to $v_j$, disjoint
  from all currently existing paths in the network.
 \end{itemize}
\end{lemma}
\begin{proof}
 By hypothesis, terminals $v_1,v_2,v_3$ all lie on the outer face of $L''$.
 It can be seen easily that we can repeatedly add direct paths from $v_i$
 to $v_j$ without intersecting others by drawing them on the outer face 
 (see~\cref{fig:bad_kernel_eg}). Hence $L'$ is planar with $v_1,v_2,v_3$ on the same face.
 As noted above, $\bkc{\parConf}$ is a sub-configuration of $\parConf$. 
 Since all paths of parity corresponding
 to bad entries of $\parConf$ have already been added 
 in $L''$, the remaining entries for which paths 
 are yet to be added are good entries. Now, at any step in the above procedure, if we 
 add a path of parity $x_{ij}$ from $v_i$ to $v_j$ disjoint from all existing paths in the network, the
 only possible extra terminal to terminal paths that can be formed are \emph{Via} paths from $v_i$
 to $v_k$ via $v_j$, of parity $x_{ij}+x_{jk}$, and from $v_k$ to $v_j$ via
 $v_i$, of parity $x_{ki}+x_{ij}$. By definition of a good entry,
 both $x_{ij}+x_{jk}, x_{ki}+x_{ij}$ must exist in $\left(\dir{\parConf}{v_i}{v_k} 
 \bigcup \via{\parConf}{v_i}{v_j}{v_k}\right )$ and $\left(\dir{\parConf}{v_k}{v_j} 
 \bigcup \via{\parConf}{v_k}{v_i}{v_j}\right )$ respectively. Therefore this does not create any
 paths of unwanted parities in $L'$, and hence we can safely
 construct $L'$ by this operation.
\end{proof}

Now all that remains to show is how to construct parity mimicking networks
for closures of all possible \emph{bad kernels}. 
For visual aid in figures, we will call the direct sets $\dir{\parConf}{v_1}{v_2},
\dir{\parConf}{v_2}{v_3},\dir{\parConf}{v_3}{v_1}$ as sets in \emph{upper row} and the sets 
$\dir{\parConf}{v_2}{v_1}$, $\dir{\parConf}{v_3}{v_2},$ 
$\dir{\parConf}{v_1}{v_3}$ as the sets in \emph{lower row}.
We will make a few observations regarding bad pairs of 
a realizable parity configuration $\parConf$ which follow easily from definition of
a \emph{bad pairs}. We observe that:
\begin{observation}\label{obs:bad_pairs_struct}
\begin{enumerate}
 \item There can be at most two bad pairs between $\dir{\parConf}{v_i}{v_j}$
  and $\dir{\parConf}{v_j}{v_k}$. This follows from the observation that at least one
  of $\{0,1\}$ must be present in \\ 
  $\left(\dir{\parConf}{v_i}{v_k} \bigcup \via{\parConf}{v_i}{v_j}{v_k} \right )$, 
  since there is some path from $v_i$ to $v_k$ if $\parConf$
  is realizable.
 \item If there are two bad pairs between $\dir{\parConf}{v_i}{v_j}$ and
  $\dir{\parConf}{v_j}{v_k}$, then each of $\dir{\parConf}{v_i}{v_j}$ and 
  $\dir{\parConf}{v_j}{v_k}$ has both $0,1$ as entries, and the bad pairs are disjoint.
  For example, if ($x_{ij},x_{jk}$) form a bad pair bewteen sets $\dir{\parConf}{v_i}{v_j}$, and 
  $\dir{\parConf}{v_j}{v_k}$, then the other bad pair between these sets, if 
  it exists, must be ($x_{ij}+1,x_{jk}+1$).
\item Suppose a direct set $\dir{\parConf}{v_i}{v_j}$ has both $0,1$ as entries. 
  Then there can be no bad pairs formed between 
  $\dir{\parConf}{v_i}{v_k}$ and $\dir{\parConf}{v_k}{v_j}$.
\item Bad pairs can be formed only between \emph{Dir}
  sets within upper row, or \emph{Dir} sets within lower row, not across.
\end{enumerate}
\end{observation}

To enumerate on the bad kernels, we can adopt without loss of generality, the following two
conventions:
\begin{itemize}
  \item Number of bad pairs between upper row sets $\geq$ Number of bad pairs between lower row sets.
  \item Number of bad pairs between $\dir{\parConf}{v_1}{v_2},\dir{\parConf}{v_2}{v_3} \geq$ 
   Number of bad pairs between $\dir{\parConf}{v_2}{v_3},\dir{\parConf}{v_3}{v_1}\geq $ 
   Number of bad pairs between $\dir{\parConf}{v_3}{v_1},\dir{\parConf}{v_1}{v_2},$
  \end{itemize}    
  The other cases are handled by symmetry.
We make the following claim:
\begin{claim}
 Any realisable parity configuration $\parConf$ can have at most $6$ bad pairs. 
\end{claim}
\begin{proof}
 Let the number of bad pairs in $\parConf$ be more than $6$. We can then assume using our
 conventions that the upper row sets have at least $4$ bad pairs, and two of them must be 
 between $\dir{\parConf}{v_1}{v_2},\dir{\parConf}{v_2}{v_3}$. This implies, by 
 parts $2$ and $3$ of observation~\ref{obs:bad_pairs_struct} above, that 
 $\dir{\parConf}{v_1}{v_2}$ and $\dir{\parConf}{v_2}{v_3}$ each
 have both $0,1$ as entries, and hence there cannot be any bad pairs between
 $\dir{\parConf}{v_1}{v_3}, \dir{\parConf}{v_3}{v_2}$ and between $\dir{\parConf}{v_2}{v_1},
 \dir{\parConf}{v_1}{v_3}$ in the lower row. Now there are two cases:
 \begin{itemize} 
  \item If $\dir{\parConf}{v_3}{v_1}$ has two bad entries, each forming bad pairs with
   other upper row sets, then there cannot be any bad pairs between the lower
   row sets. Since the total number of bad pairs cannot be more than $6$ in the
   upper row, this leads to a contradiction. 
 \item If $\dir{\parConf}{v_3}{v_1}$ has lesser than two bad entries,
   they can form at most two bad pairs. Then the total number
   of bad pairs in upper row is at most four, and in lower row at most two, which also 
   leads to a contradiction.
 \end{itemize}
   Hence total number of bad pairs in $\parConf$ cannot be more than $6$.
\end{proof}

This gives a bound on number of types of bad kernels we need to consider.
We can write the number of bad pairs of $\parConf$ as $(n_u,n_l)$, where $n_u$ is the
number of bad pairs in upper row and $n_l$ the number of bad pairs in the
lower row. Since $n_l\leq n_u\leq 6$ and $n_l+n_u\leq 6$, we can lexicographically
enumerate all cases from $(6,0)$ to $(1,0)$, and construct explicitly, the 
required mimicking networks for closure of each case.
We list the cases below starting from $(6,0)$ to $(1,1)$.
The cases lying in between $(6,0)$ and $(1,1)$ that are not drawn
are those which cannot occur as bad kernels. We give some examples
of such a cases, others have a similar argument.
Case remaining after $(1,1)$ is $(1,0)$, which is
shown in proof idea, so we do not draw it here.

In all of the figures below, the entries of the parity configuration
tables joined by the red lines denote bad pairs. In the corresponding
gadgets, if no parity expression is written beside an edge, then it
is a path of length one or two according to if it has a single or
a double arrow respectively. If a parity expression is written then
the length is according to the parity expression. The variables can
take values in $\{0,1\}.$
We only show $Via$ sets when required. Otherwise we assume them to be
empty by default.
The only case that is a bit tedious to verify is that of $(4,0)$, 
as shown in~\cref{fig:satanic}. 
For ease of verification, we have listed all the paths from $v_1$ to $v_3$
along with their lengths (counting single arrows as length $1$, double arrows
as length $2$.)
All of these must be of same parity in that particular case. Paths from $v_2$ to $v_1$
follow a symmetric pattern.
\begin{figure}[hbt!]
 \begin{minipage}[t]{\textwidth}
  \includegraphics[scale=0.6]{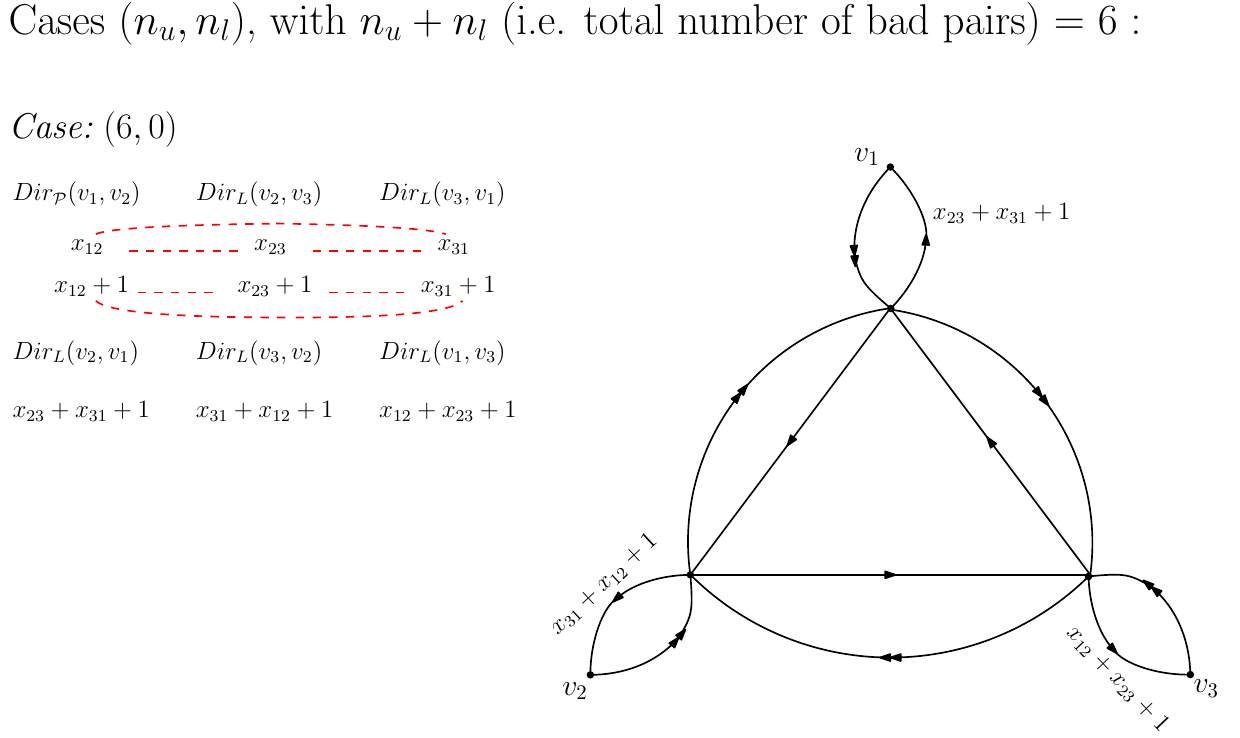}
 \end{minipage}
\end{figure}
\begin{figure}[hbt!]
 \begin{minipage}[t]{\textwidth}
  \includegraphics[scale=0.6]{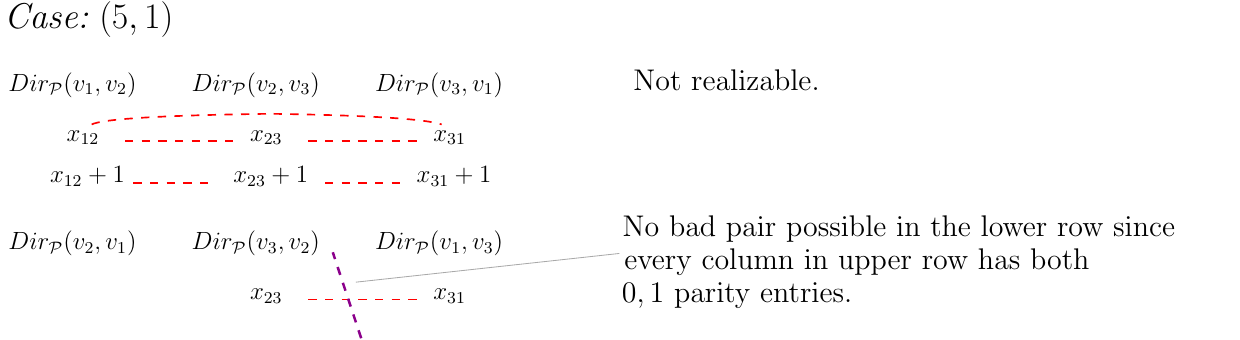}
 \end{minipage}
\end{figure}
\begin{figure}[hbt!]
 \begin{minipage}[t]{\textwidth}
  \includegraphics[scale=0.6]{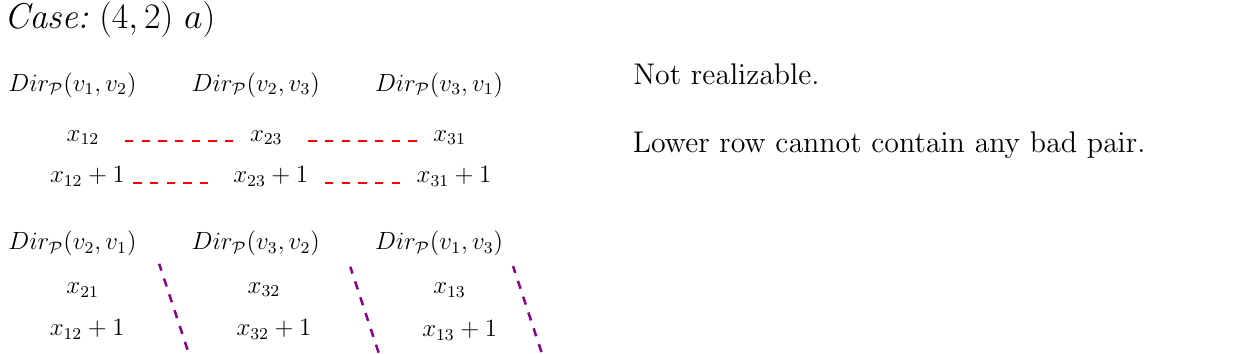}
 \end{minipage}
\end{figure}
\begin{figure}[hbt!]
 \begin{minipage}[t]{\textwidth}
  \includegraphics[scale=0.6]{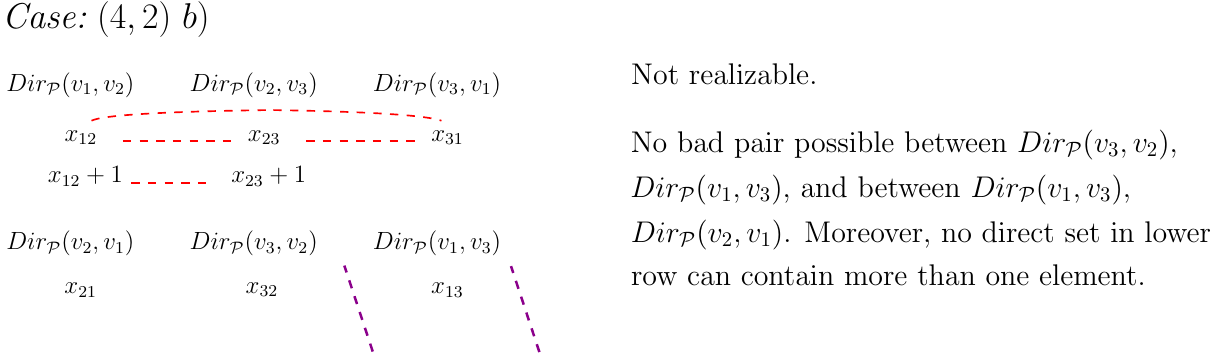}
 \end{minipage}
\end{figure}
\begin{figure}[hbt!]
 \begin{minipage}[t]{\textwidth}
  \includegraphics[scale=0.6]{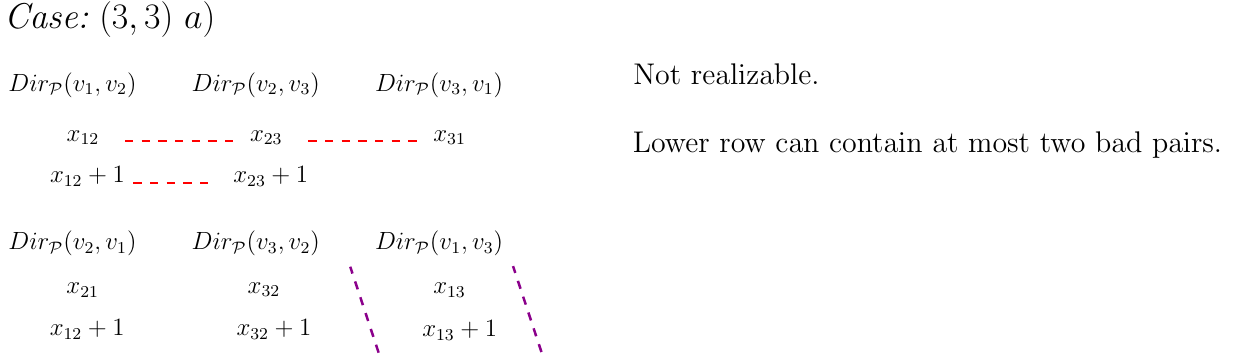}
 \end{minipage}
\end{figure}
\begin{figure}[hbt!]
 \begin{minipage}[t]{\textwidth}
  \includegraphics[scale=0.6]{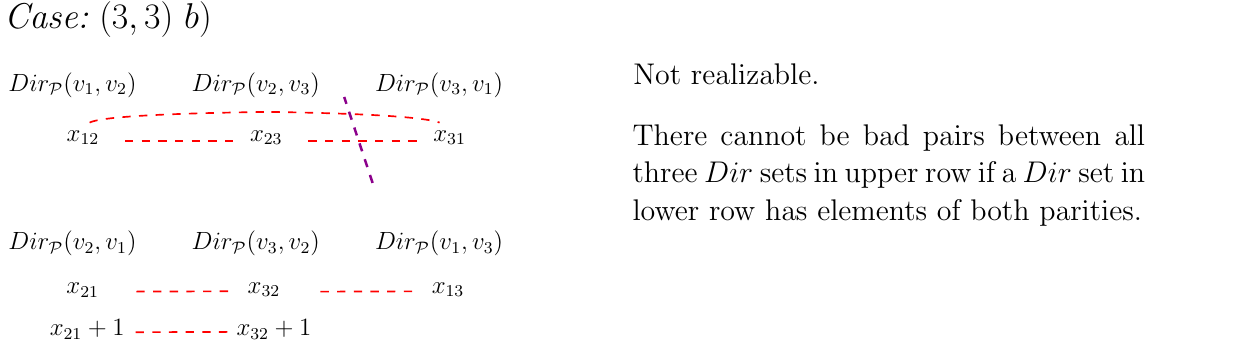}
 \end{minipage}
\end{figure}
\begin{figure}[hbt!]
 \begin{minipage}[t]{\textwidth}
  \includegraphics[scale=0.6]{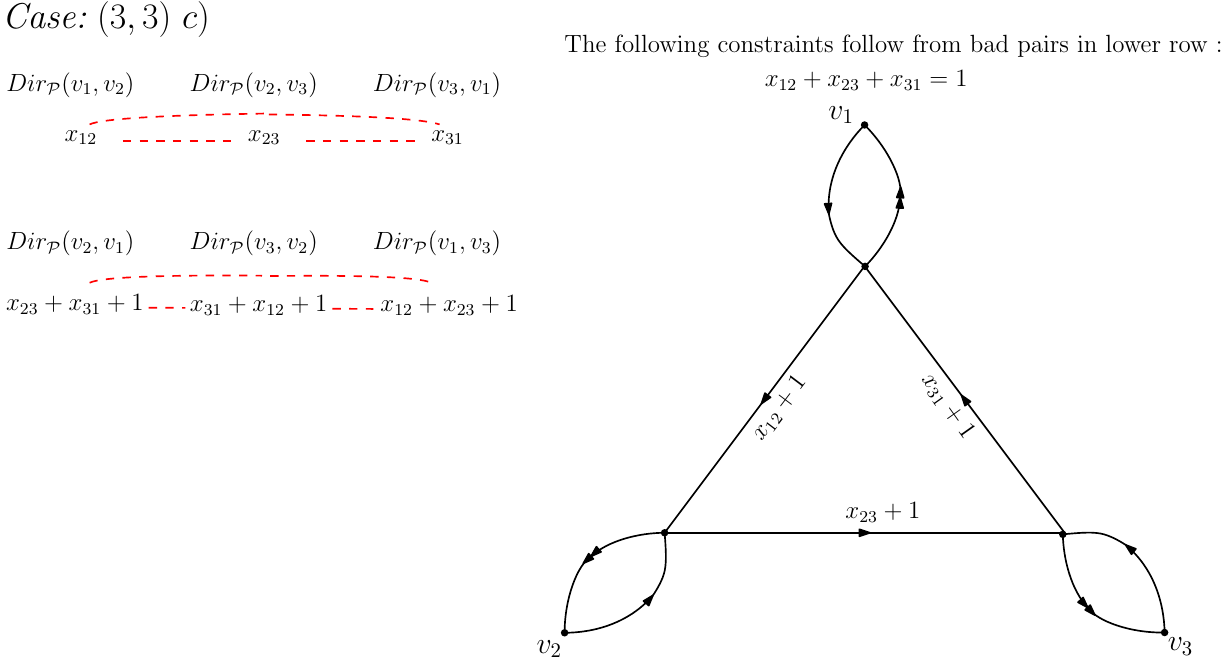}
 \end{minipage}
\end{figure}
\begin{figure}[hbt!]
 \begin{minipage}[t]{\textwidth}
  \includegraphics[scale=0.6]{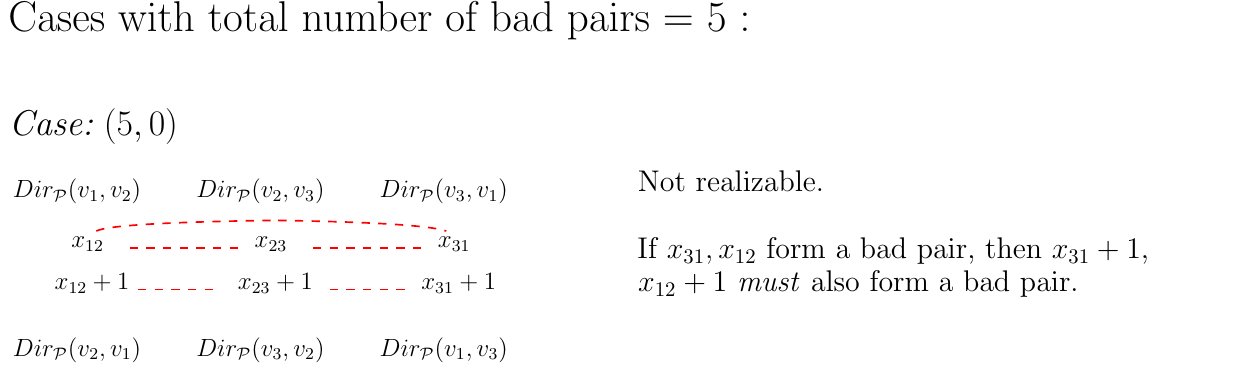}
 \end{minipage}
\end{figure}
\begin{figure}[hbt!]
 \begin{minipage}[t]{\textwidth}
  \includegraphics[scale=0.6]{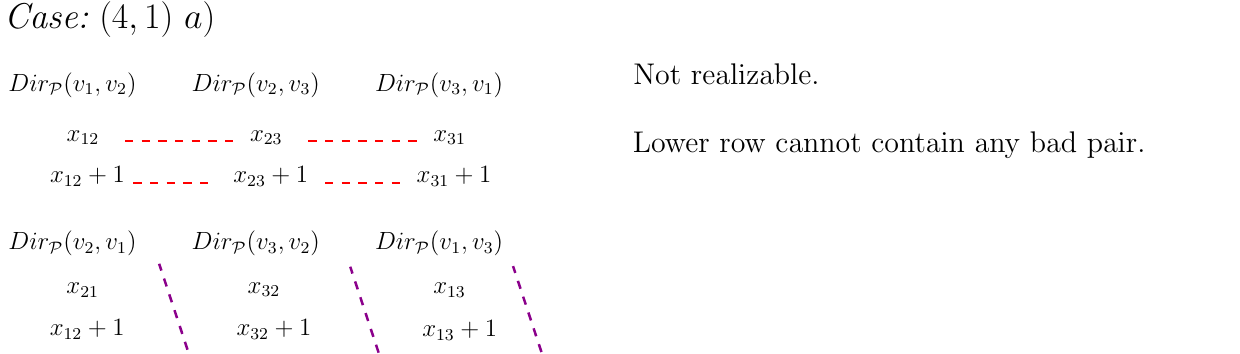}
 \end{minipage}
\end{figure}
\begin{figure}[hbt!]
 \begin{minipage}[t]{\textwidth}
  \includegraphics[scale=0.6]{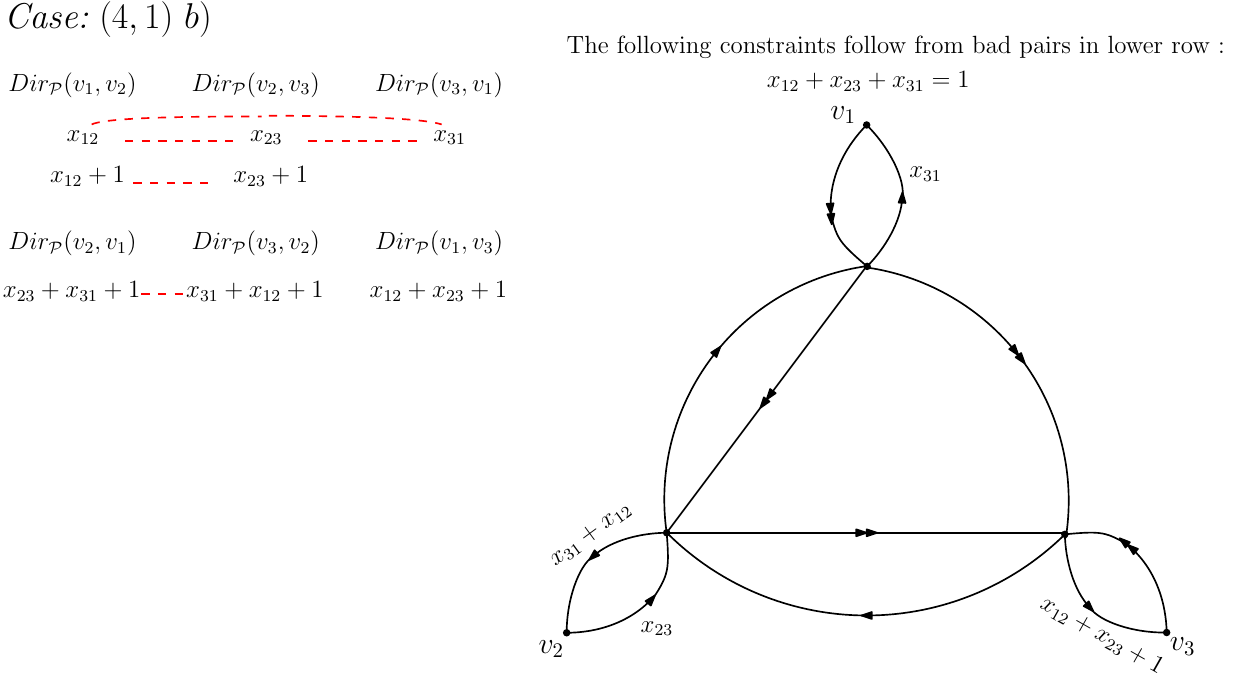}
 \end{minipage}
\end{figure}
\begin{figure}[hbt!]
 \begin{minipage}[t]{\textwidth}
  \includegraphics[scale=0.6]{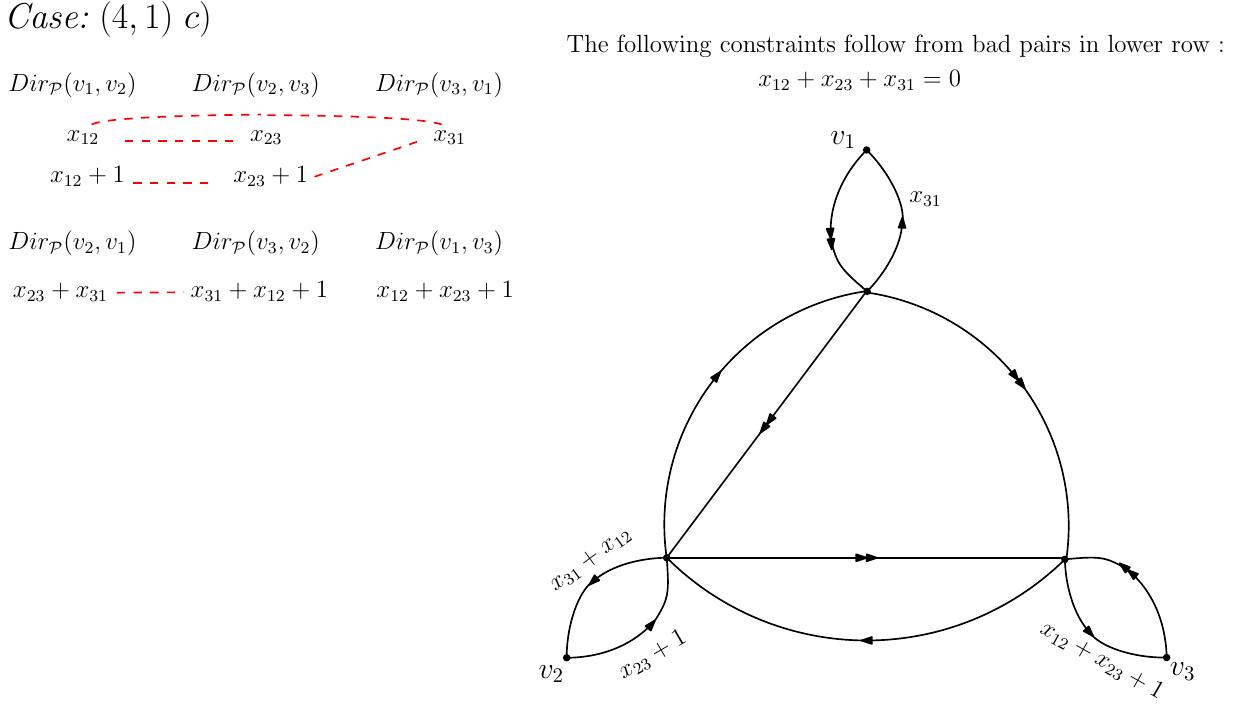}
 \end{minipage}
\end{figure}
\begin{figure}[hbt!]
 \begin{minipage}[t]{\textwidth}
  \includegraphics[scale=0.6]{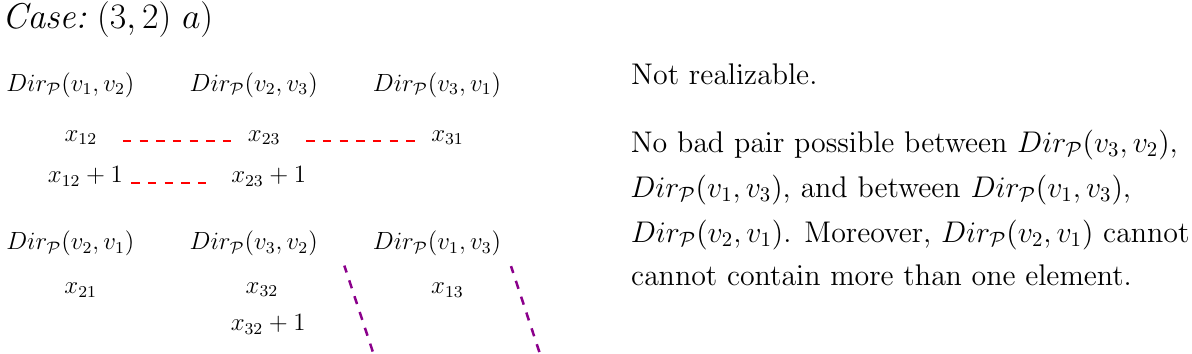}
 \end{minipage}
\end{figure}
\begin{figure}[hbt!]
 \begin{minipage}[t]{\textwidth}
  \includegraphics[scale=0.6]{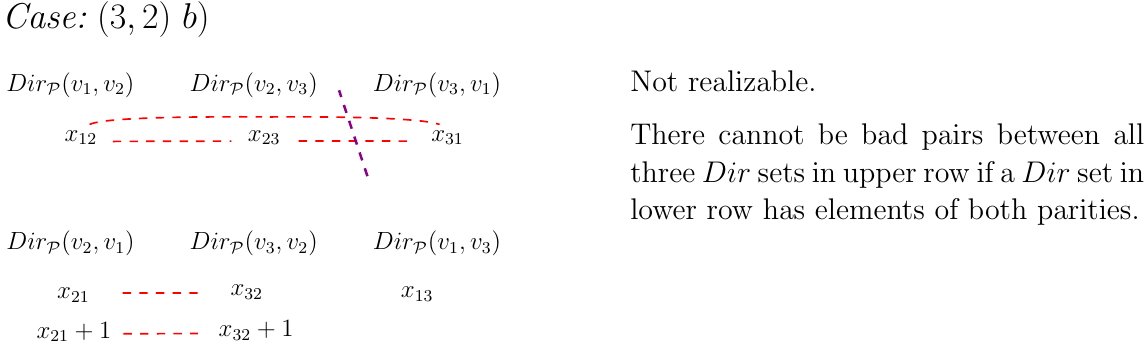}
 \end{minipage}
\end{figure}
\begin{figure}[hbt!]
 \begin{minipage}[t]{\textwidth}
  \includegraphics[scale=0.6]{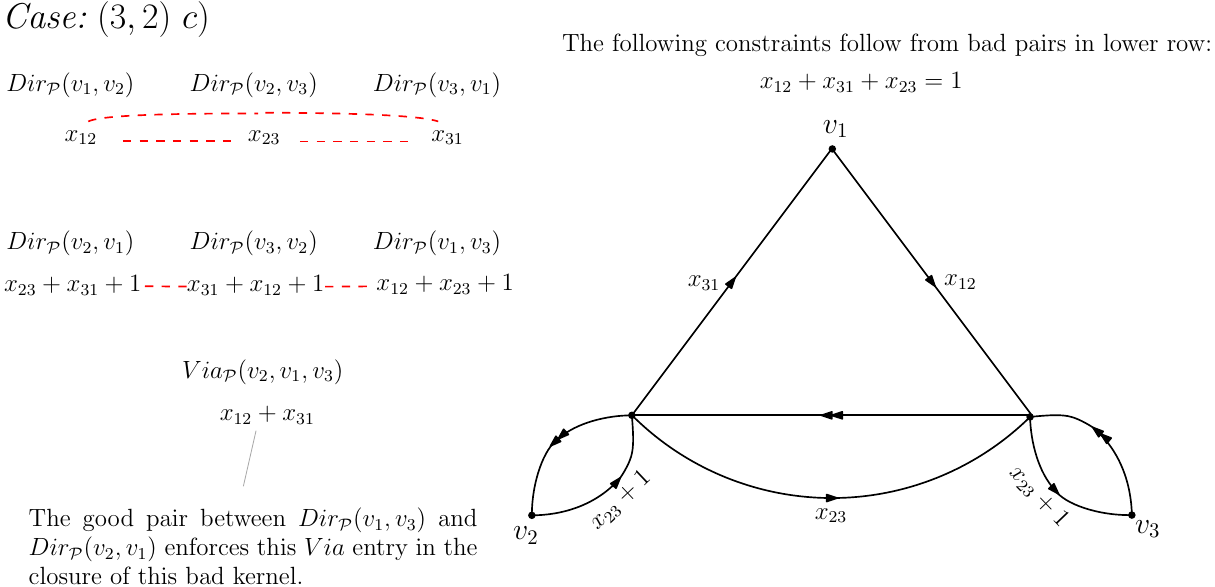}
 \end{minipage}
\end{figure}
\begin{figure}[hbt!]
 \begin{minipage}[t]{\textwidth}
  \includegraphics[scale=0.6]{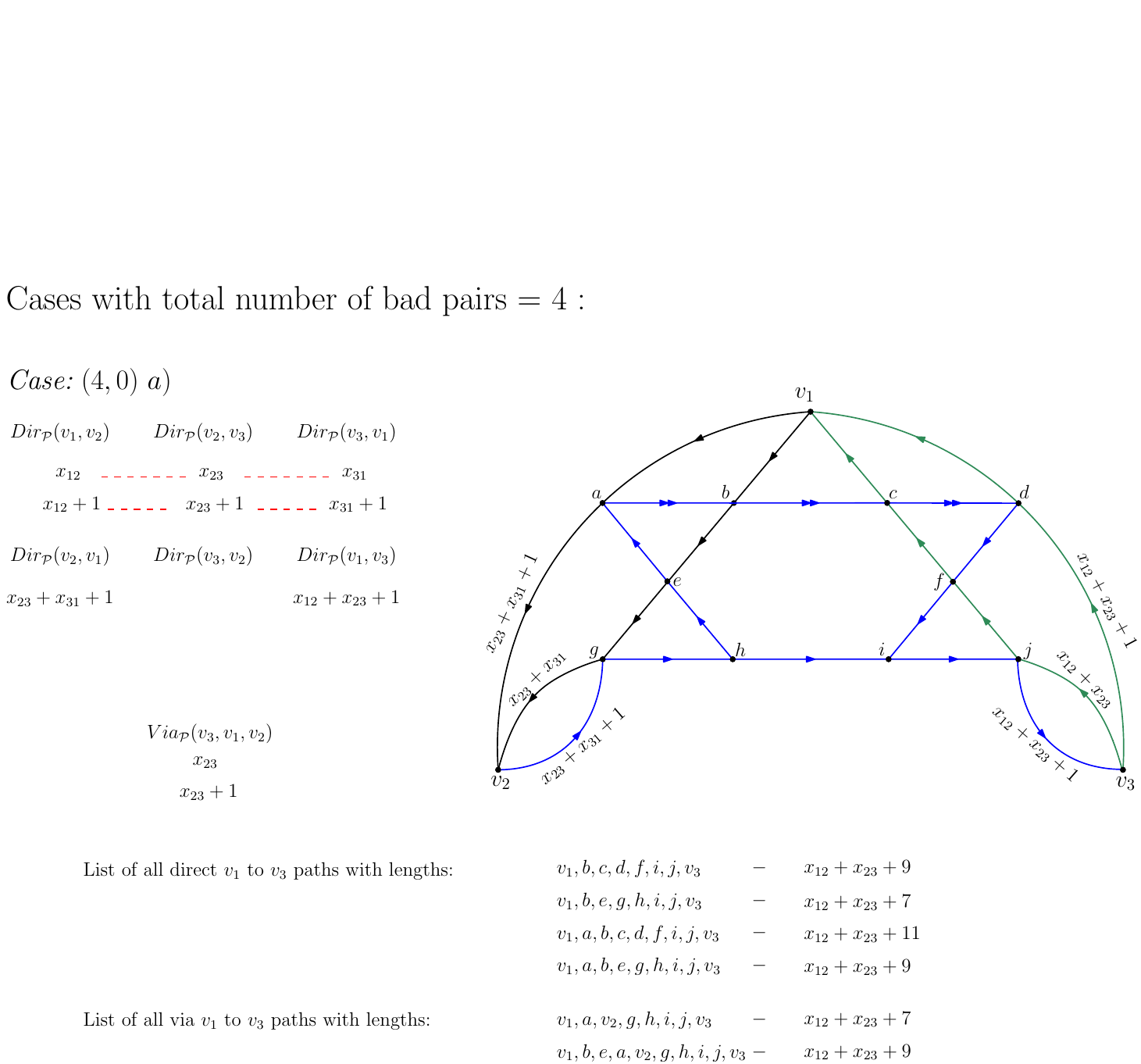}
\caption{Case $(4,0)$, with list of all paths from $v_1$ to $v_3$
 for ease of verification.}\label{fig:satanic}
 \end{minipage}
\end{figure}
\begin{figure}[hbt!]
 \begin{minipage}[t]{\textwidth}
  \includegraphics[scale=0.6]{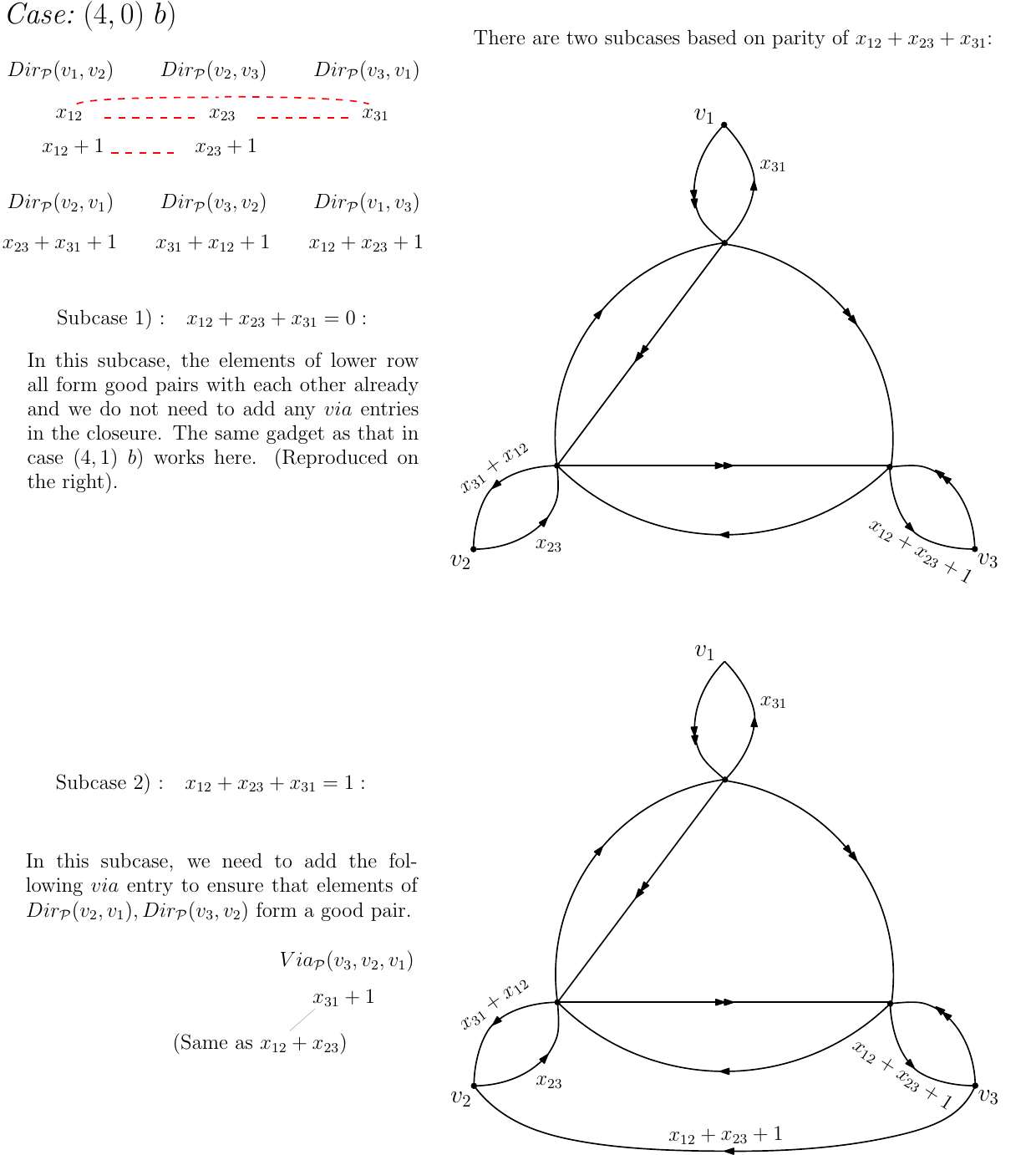}
 \end{minipage}
\end{figure}
\begin{figure}[hbt!]
 \begin{minipage}[t]{\textwidth}
  \includegraphics[scale=0.6]{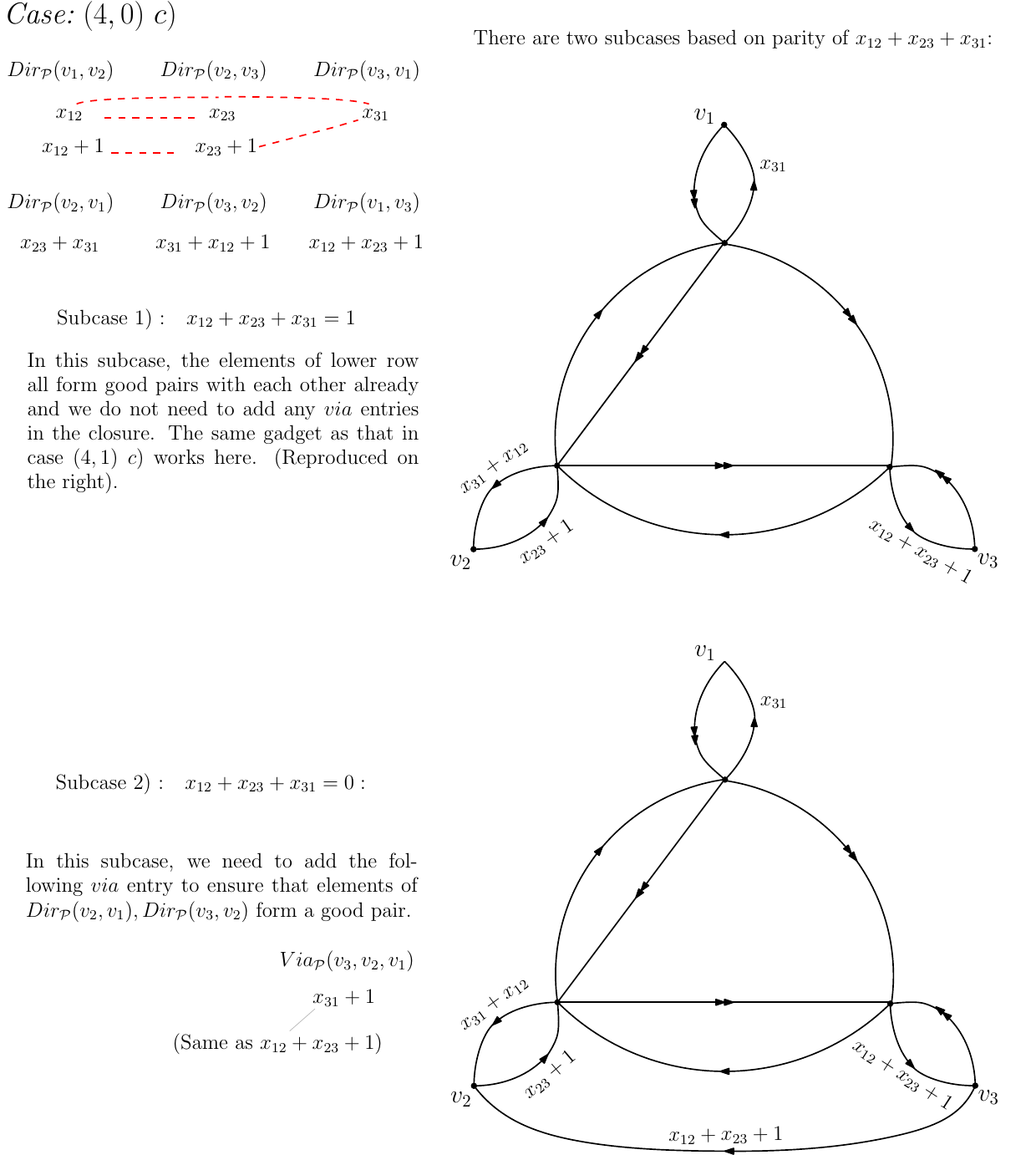}
 \end{minipage}
\end{figure}
\begin{figure}[hbt!]
 \begin{minipage}[t]{\textwidth}
  \includegraphics[scale=0.6]{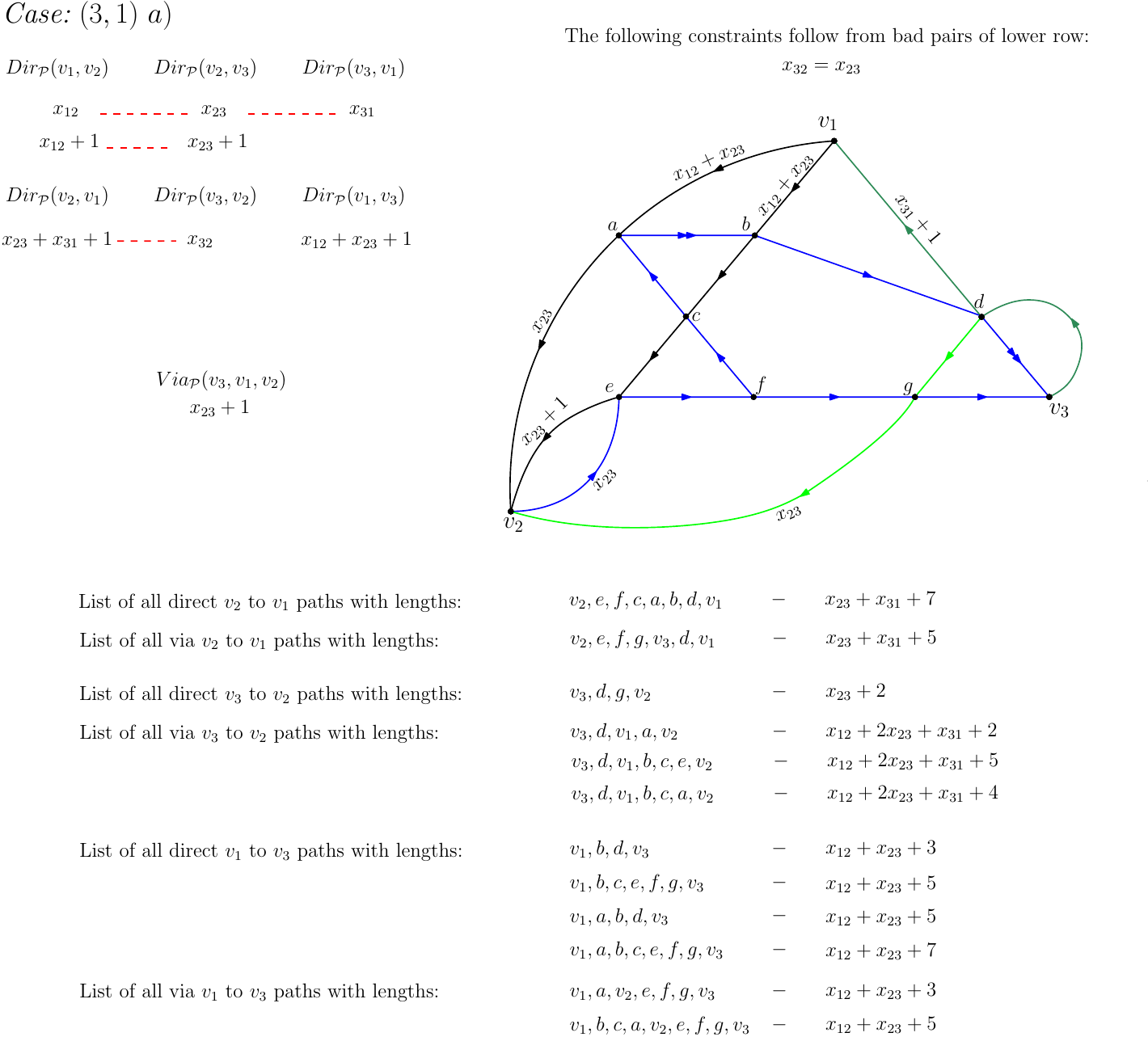}
 \end{minipage}
\end{figure}
\begin{figure}[hbt!]
 \begin{minipage}[t]{\textwidth}
  \includegraphics[scale=0.6]{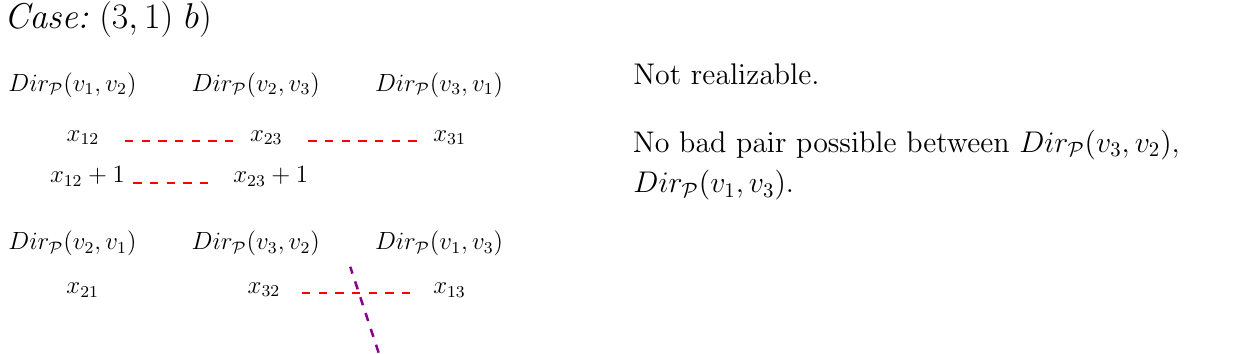}
 \end{minipage}
\end{figure}
\begin{figure}[hbt!]
 \begin{minipage}[t]{\textwidth}
  \includegraphics[scale=0.6]{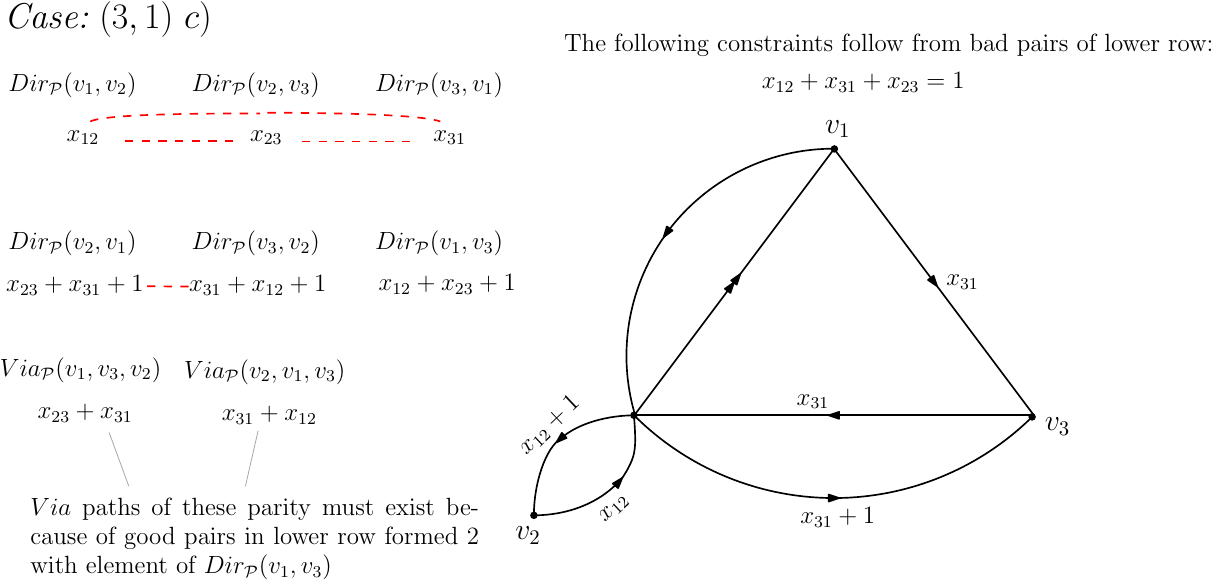}
 \end{minipage}
\end{figure}
\begin{figure}[hbt!]
 \begin{minipage}[t]{\textwidth}
  \includegraphics[scale=0.6]{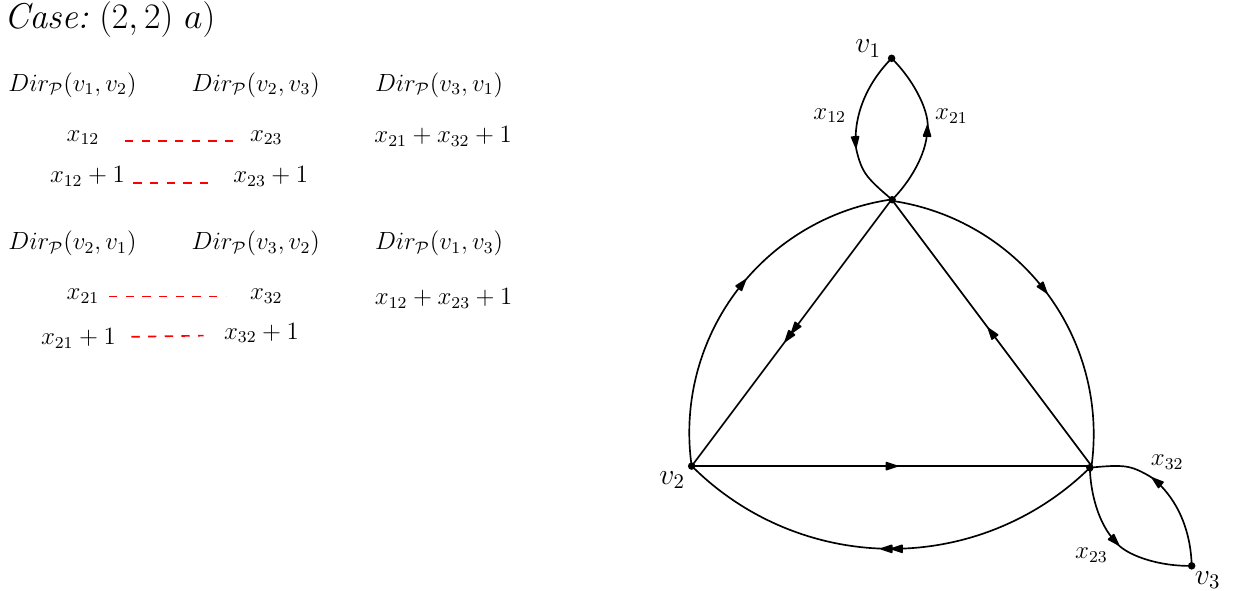}
 \end{minipage}
\end{figure}
\begin{figure}[hbt!]
 \begin{minipage}[t]{\textwidth}
  \includegraphics[scale=0.6]{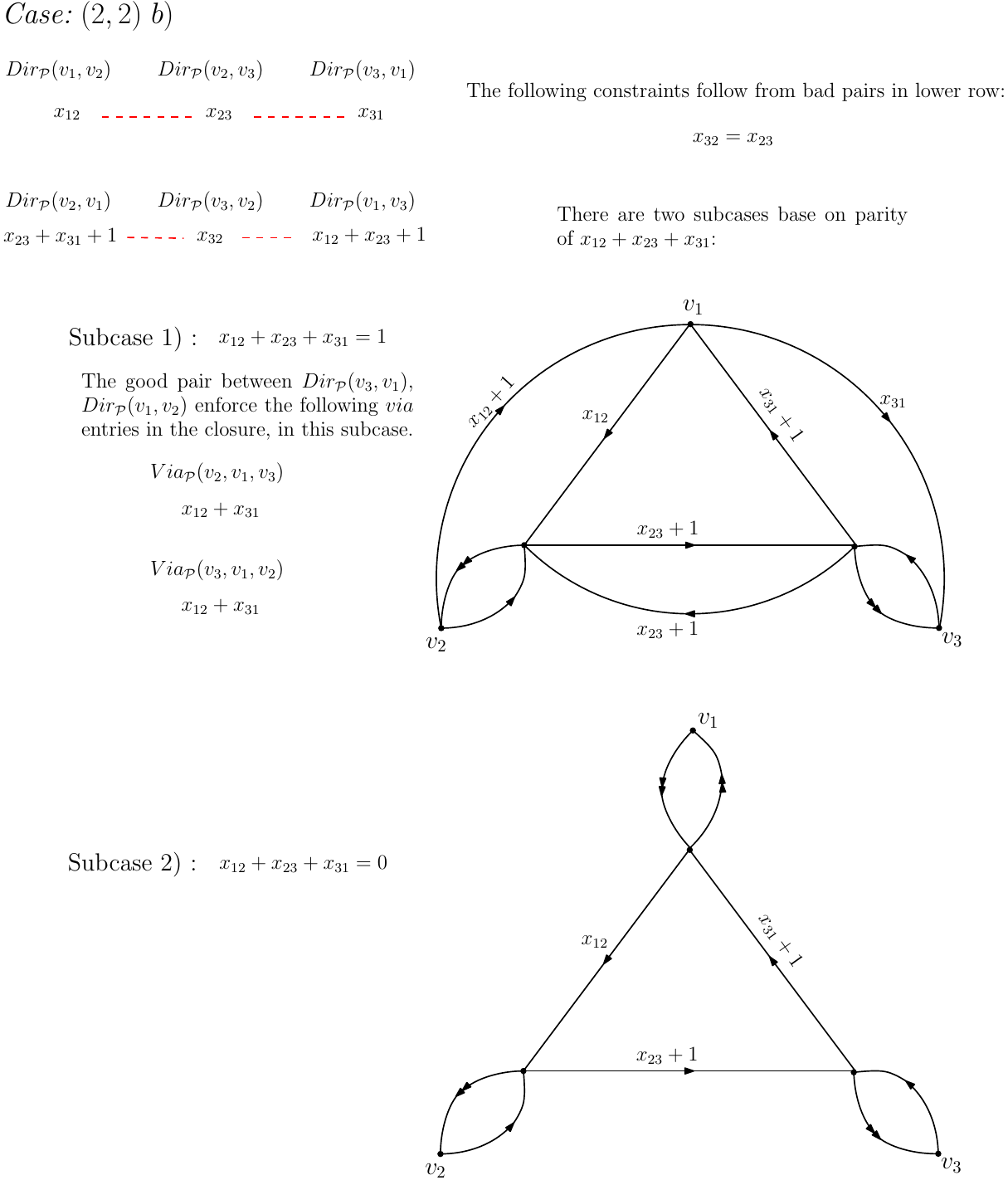}
 \end{minipage}
\end{figure}
\begin{figure}[hbt!]
 \begin{minipage}[t]{\textwidth}
  \includegraphics[scale=0.6]{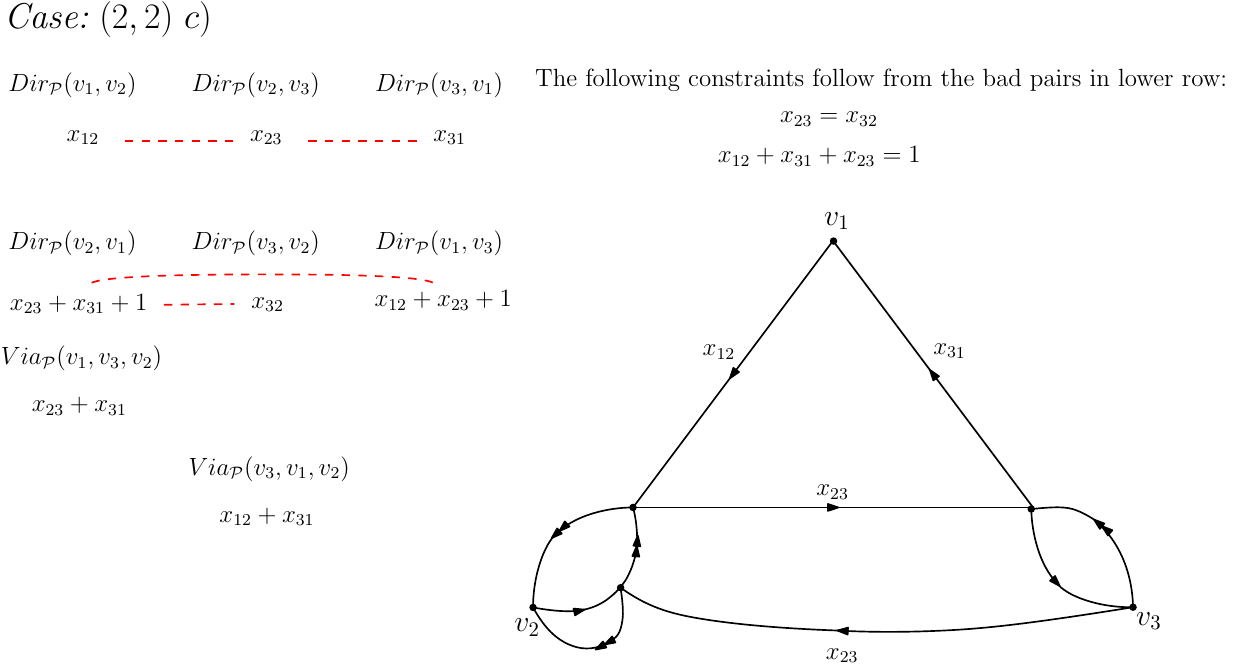}
 \end{minipage}
\end{figure}
\begin{figure}[hbt!]
 \begin{minipage}[t]{\textwidth}
  \includegraphics[scale=0.6]{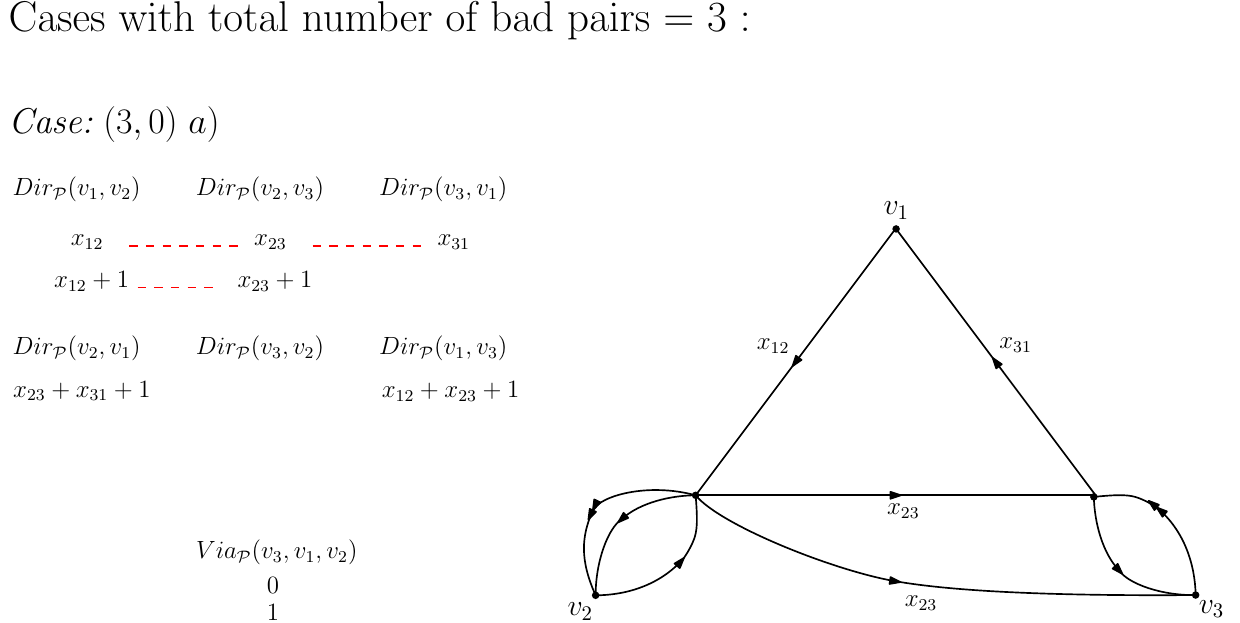}
 \end{minipage}
\end{figure}
\begin{figure}[hbt!]
 \begin{minipage}[t]{\textwidth}
  \includegraphics[scale=0.6]{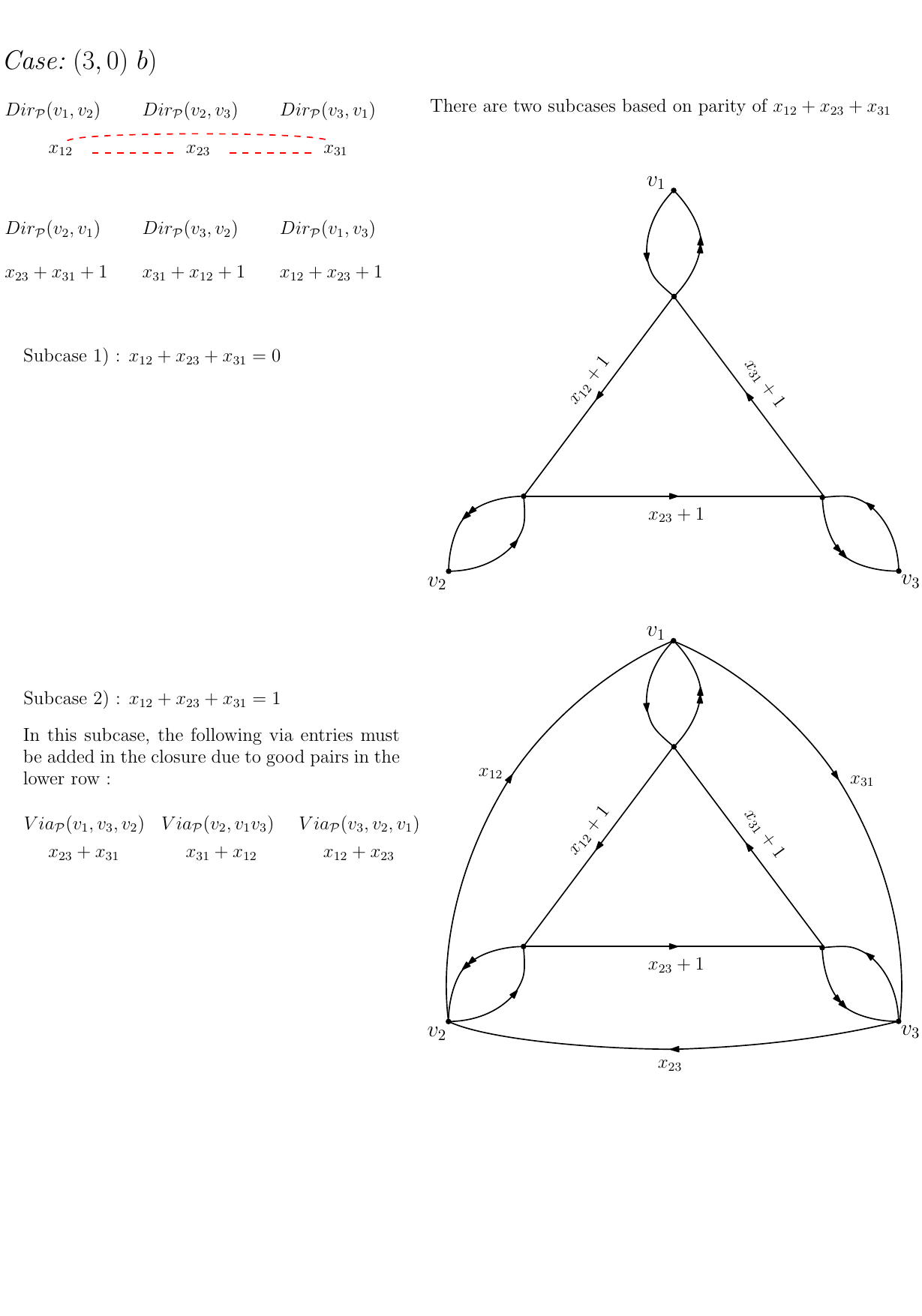}
 \end{minipage}
\end{figure}
\begin{figure}[hbt!]
 \begin{minipage}[t]{\textwidth}
  \includegraphics[scale=0.6]{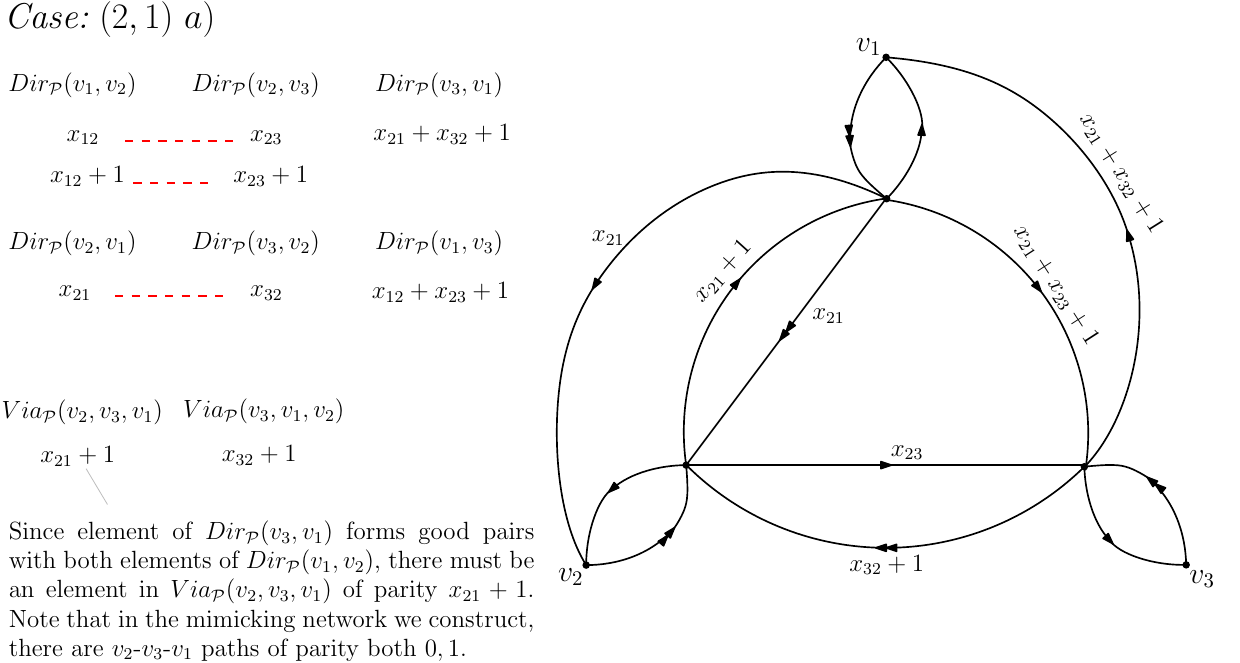}
 \end{minipage}
\end{figure}
\begin{figure}[hbt!]
 \begin{minipage}[t]{\textwidth}
  \includegraphics[scale=0.6]{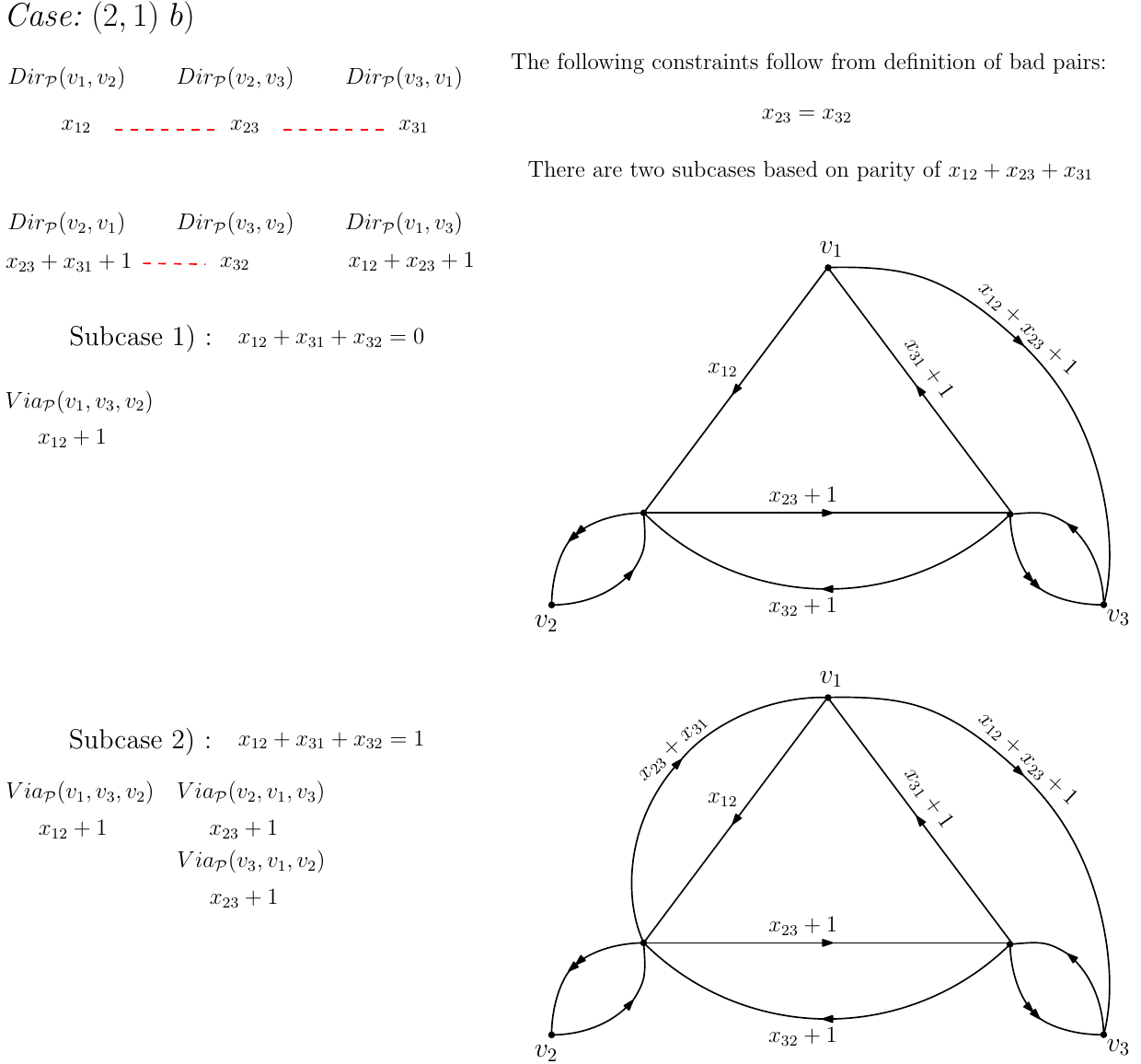}
 \end{minipage}
\end{figure}
\begin{figure}[hbt!]
 \begin{minipage}[t]{\textwidth}
  \includegraphics[scale=0.6]{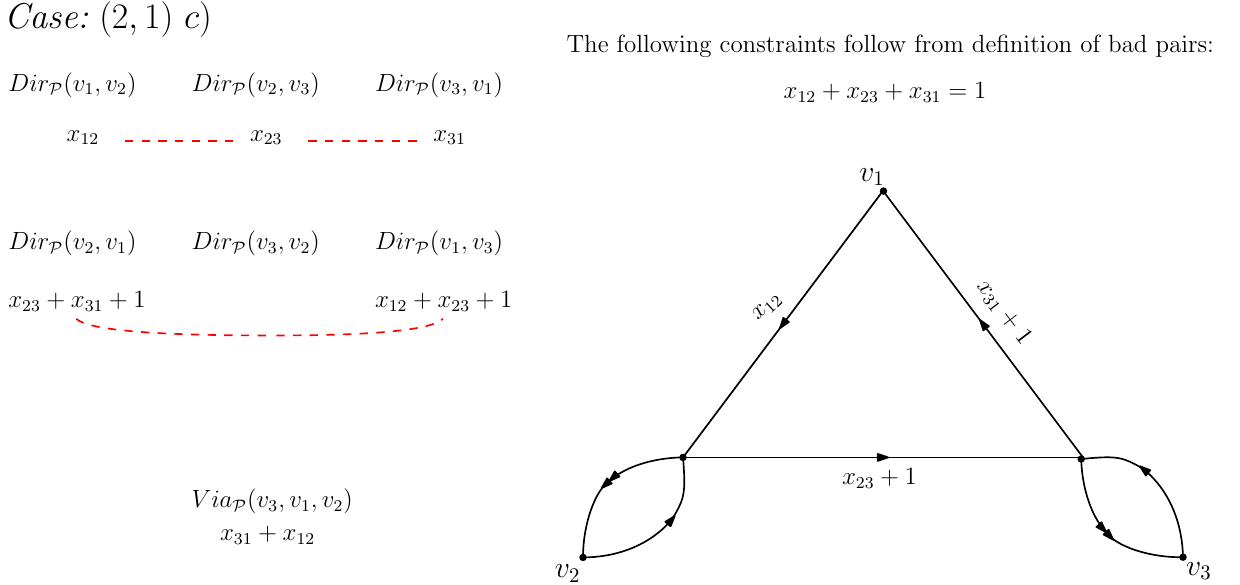}
 \end{minipage}
\end{figure}
\begin{figure}[hbt!]
 \begin{minipage}[t]{\textwidth}
  \includegraphics[scale=0.6]{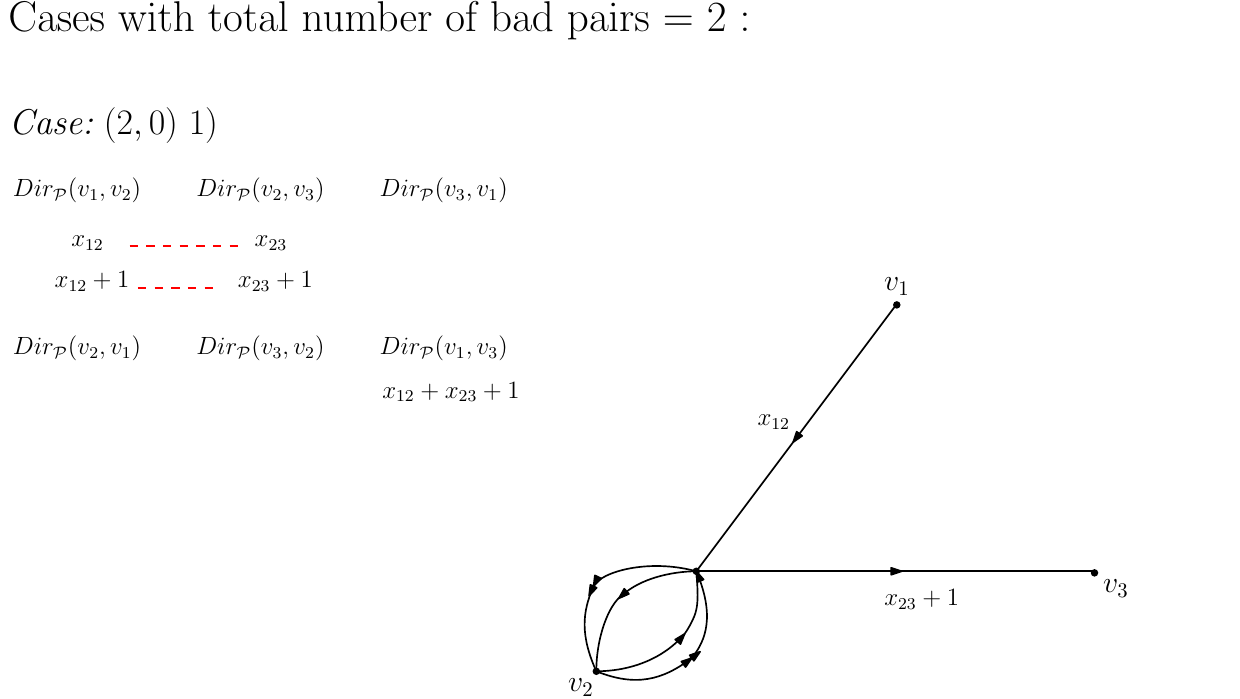}
 \end{minipage}
\end{figure}
\begin{figure}[hbt!]
 \begin{minipage}[t]{\textwidth}
  \includegraphics[scale=0.6]{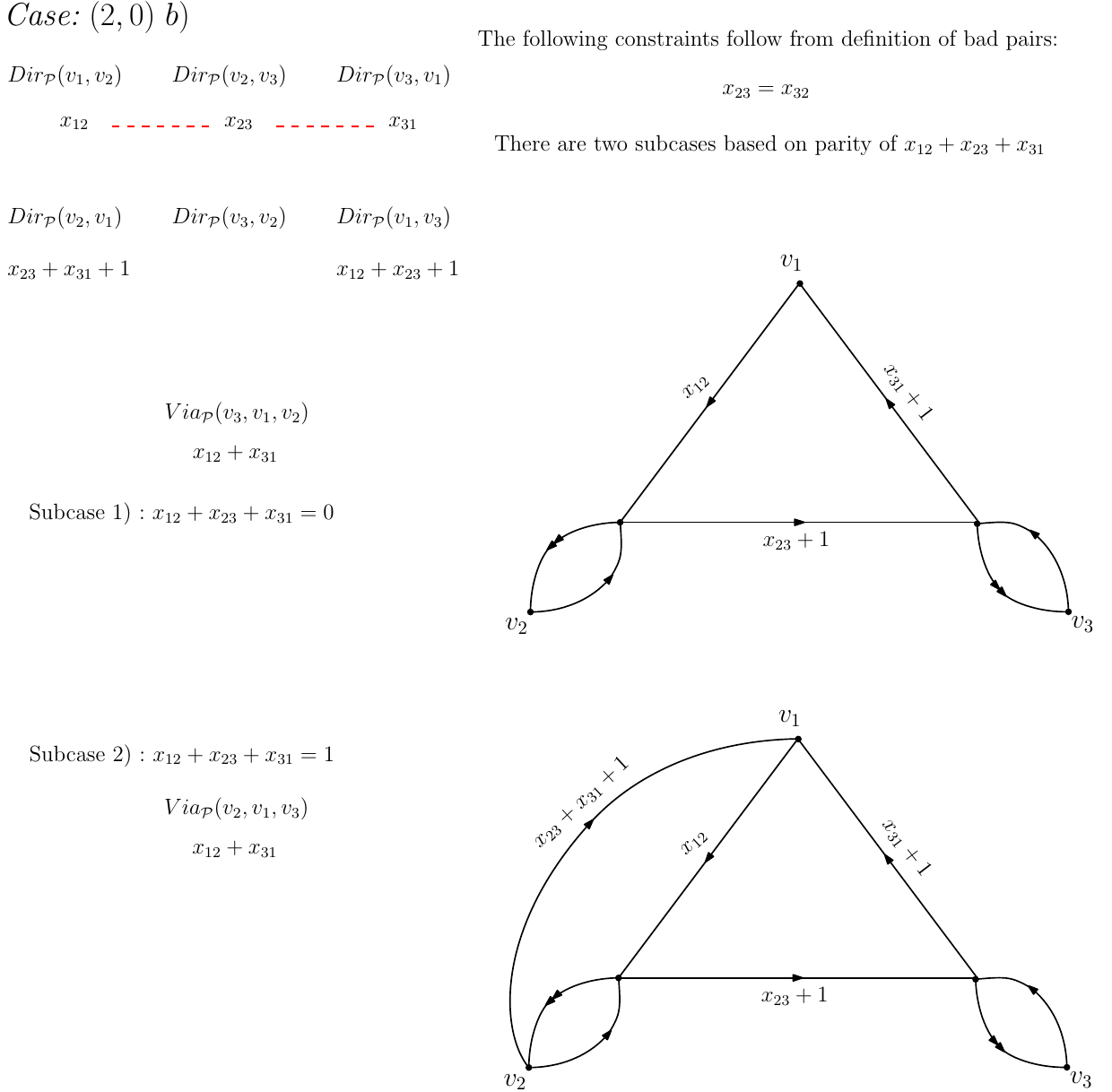}
 \end{minipage}
\end{figure}
\begin{figure}[hbt!]
 \begin{minipage}[t]{\textwidth}
  \includegraphics[scale=0.6]{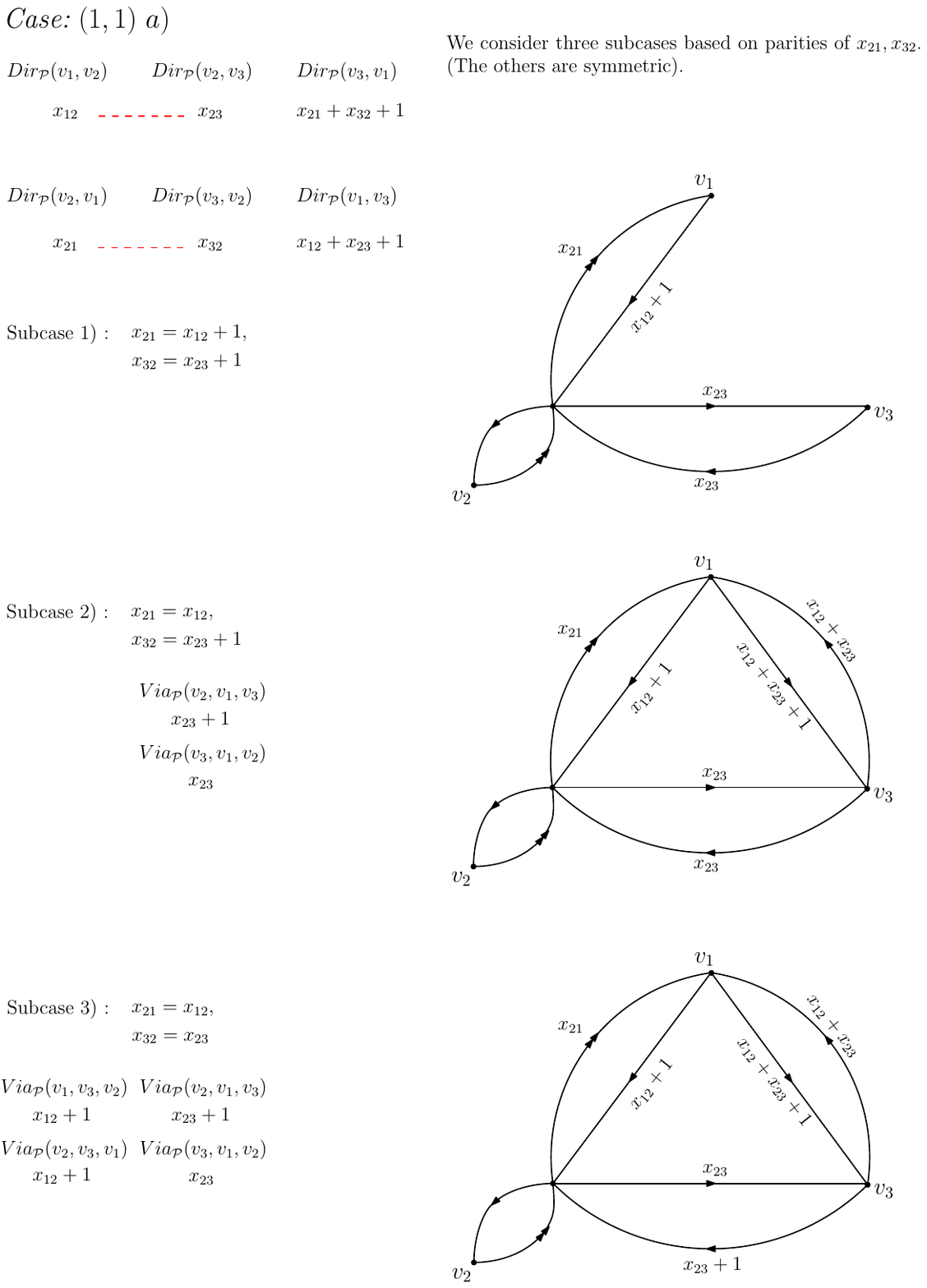}
 \end{minipage}
\end{figure}
\begin{figure}[hbt!]
 \begin{minipage}[t]{\textwidth}
  \includegraphics[scale=0.6]{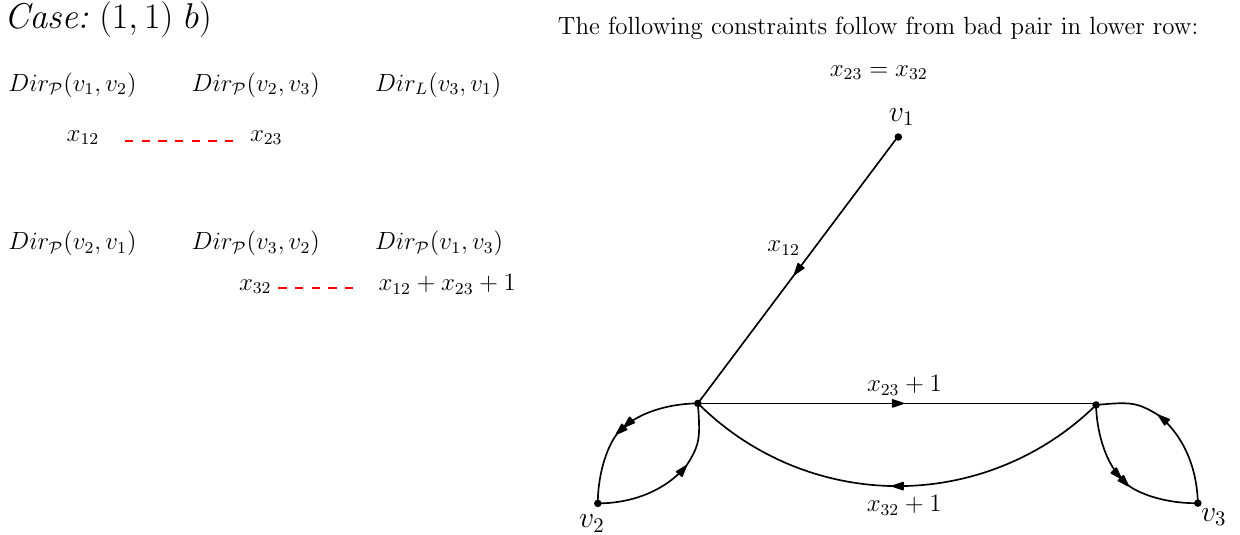}
 \end{minipage}
\end{figure}
\end{proof}

%

%
%
%
%
%
%
%
%
\end{document}